\documentclass{article}
\usepackage{kr}

\usepackage{times}
\usepackage{soul}
\usepackage{url}
\usepackage[hidelinks]{hyperref}
\usepackage[utf8]{inputenc}
\usepackage[small]{caption}
\usepackage{graphicx}
\usepackage{amsmath}
\usepackage{amsthm}
\usepackage{booktabs}
\usepackage{algorithm}
\usepackage{algorithmic}
\usepackage{xspace}
\usepackage{subcaption}
\usepackage{multirow}
\usepackage[export]{adjustbox}

\usepackage{listings}
\usepackage{xcolor}
\usepackage{amsfonts, amssymb, stmaryrd}
\usepackage{thm-restate} 

\newcommand{\inds}{{\bf{C}}}
\newcommand{\vars}{{\bf{V}}}
\newcommand{\preds}{{\bf{P}}}

\def\dllitecore{DL-Lite$_{\mn{core}}$\xspace}

\newcommand{\ans}{\vec{a}}

\newcommand{\brave}{\textup{brave}\xspace}
\newcommand{\AR}{\textup{AR}\xspace}
\newcommand{\GR}{\textup{GR}\xspace}
\newcommand{\IAR}{\textup{IAR}\xspace}

\newcommand{\bravemodels}[1]{\models_{\brave}^{#1}}
\newcommand{\armodels}[1]{\models_{\AR}^{#1}}
\newcommand{\iarmodels}[1]{\models_{\IAR}^{#1}}
\newcommand{\grmodels}{\models_{\GR}}

\newcommand{\conflicts}[1]{\mi{Conf}(#1)}
\newcommand{\causes}[1]{\mi{Causes}(#1)}

\newcommand{\reps}[1]{\mi{SRep}(#1)}
\newcommand{\greps}[1]{\mi{GRep}(#1)}
\newcommand{\preps}[1]{\mi{PRep}(#1)}
\newcommand{\creps}[1]{\mi{CRep}(#1)}
\newcommand{\xreps}[1]{\mi{XRep}(#1)}

\def\reachstrong{\mn{R}_\mn{s}}
\def\reachweak{\mn{R}_\mn{w}}

\newcommand{\minconf}[1]{Min(#1)}

\def\ptime{\textsc{PTime}\xspace}
\def\np{\textsc{NP}\xspace}
\def\conp{co\textsc{NP}\xspace}
\def\piptwo{\ensuremath{\Pi^{p}_{2}}\xspace}
\def\sigmaptwo{\ensuremath{\Sigma^{p}_{2}}\xspace}

\newcommand{\mn}[1]{\ensuremath{\mathsf{#1}}}
\newcommand{\mi}[1]{\ensuremath{\mathit{#1}}}
\newcommand{\mt}[1]{\ensuremath{\mathtt{#1}}}

\newcommand{\Bmc}{\ensuremath{\mathcal{B}}}
\newcommand{\Cmc}{\ensuremath{\mathcal{C}}}
\newcommand{\Dmc}{\ensuremath{\mathcal{D}}}
\newcommand{\Emc}{\ensuremath{\mathcal{E}}}
\newcommand{\Fmc}{\ensuremath{\mathcal{F}}}
\newcommand{\Gmc}{\ensuremath{\mathcal{G}}}

\newcommand{\Lmc}{\ensuremath{\mathcal{L}}}
\newcommand{\Kmc}{\ensuremath{\mathcal{K}}}

\newcommand{\Rmc}{\ensuremath{\mathcal{R}}}

\newcommand{\Tmc}{\ensuremath{\mathcal{T}}}

\def\naf{\ensuremath{\raise.17ex\hbox{\ensuremath{\scriptstyle\mathtt{\sim}}}}\xspace}

\newcommand{\eg}{e.g.,~}

\newcommand{\ie}{i.e.,~}
\newcommand{\wrt}{w.r.t.~}
\newcommand{\cf}{cf.~}

\newcommand{\resp}{resp.~}

\newcommand{\impl}{\ {:\!\text{-}}\  }
\newcommand{\ASPQ}{ASP(Q)\xspace}
\newcommand{\fix}[2]{\mathit{fix_{#1}(#2)}}

\newtheorem{example}{Example}
\newtheorem{theorem}{Theorem}

\newtheorem{claim}{Claim}
\newtheorem{proposition}{Proposition}

\newtheorem{definition}{Definition}

\title{Using ASP(Q) to Handle Inconsistent Prioritized Data}

\author{%
Meghyn Bienvenu$^1$\and
Camille Bourgaux$^2$\and
Robin Jean $^1$\and
Giuseppe Mazzotta$^3$ \\
\affiliations
$^1$Univ. Bordeaux, CNRS, Bordeaux INP, LaBRI, UMR 5800, Talence, France\\
$^2$DI ENS, ENS, CNRS, PSL University \& Inria, Paris, France\\
$^3$University of Calabria, Rende, Italy\\
\emails
\{meghyn.bienvenu, robin.jean\}@u-bordeaux.fr,
camille.bourgaux@ens.fr,
giuseppe.mazzotta@unical.it
}

\begin{document}

\maketitle

\begin{abstract}
We explore the use of answer set programming (ASP) and its extension with quantifiers, ASP(Q), for inconsistency-tolerant querying of prioritized data, where a priority relation between conflicting facts is exploited to define three notions of optimal repairs (Pareto-, globally- and completion-optimal). We consider the variants of three well-known semantics (AR, brave and IAR) that use these optimal repairs, and for which query answering is in the first or second level of the polynomial hierarchy for a large class of logical theories. Notably, this paper presents the first implementation of globally-optimal repair-based semantics, as well as the first implementation of the grounded semantics, which is a tractable under-approximation of all these optimal repair-based semantics. Our experimental evaluation sheds light on the feasibility of computing answers under globally-optimal repair semantics and the impact of adopting different semantics, approximations, and encodings. 
\end{abstract}

\section{Introduction}

Repair-based semantics are a prominent means of obtaining meaningful answers to queries posed over some data which is inconsistent \wrt some logical theory, both in the relational database and ontology-mediated query answering setting (\cf \cite{DBLP:conf/pods/Bertossi19,DBLP:journals/ki/Bienvenu20} for brief overviews). In this context, a repair is a subset-maximal subset of the data consistent with the logical theory. The most well-known repair-based semantics, called AR in the KR community, requires that the query holds in every repair, while the less cautious brave semantics requires that it holds in some repair, and the more cautious IAR semantics that it holds in the intersection of all repairs \cite{ArenasBC99,LemboLRRS10,Bienvenu_TractableApproximation_long}. 
Several notions of preferred repairs have been proposed 
to take into account some preference information and consider only a subset of all the possible repairs to evaluate the queries (\cf \cite{DBLP:conf/rweb/Bourgaux25} for a survey). In particular, \citeauthor{DBLP:journals/amai/StaworkoCM12}~\shortcite{DBLP:journals/amai/StaworkoCM12} introduced three kinds of optimal repairs based on a priority relation between conflicting facts, which have attracted a lot of interest in the last decade with extensive complexity analyses \cite{DBLP:conf/icdt/KimelfeldLP17,DBLP:journals/tcs/KimelfeldLP20,DBLP:conf/kr/BienvenuB20,DBLP:conf/kr/BienvenuB23}, two implementations \cite{DBLP:conf/kr/BienvenuB22,DBLP:conf/kr/BienvenuBIJ2025}, and a  
framework for specifying and computing priority relations \cite{DBLP:conf/kr/BienvenuBIJ2025}. 

However, only two of these three kinds of  
repairs have actually been implemented in the 
existing systems,  
namely 
the Pareto- and completion-optimal repairs. 
This is due to 
the higher complexity of reasoning with the third kind of 
repairs, 
called globally-optimal. Indeed, for a large class of database constraints and ontology languages, the data complexity of query answering under the variants of AR, IAR and brave that use Pareto- or completion-optimal repairs is in the first level of the polynomial hierarchy, while it is in the second level for the globally-optimal repair-based variants. 

Most implementations of  
repair-based semantics encode query entailment under some intractable semantics into SAT, binary 
integer programming, or answer set programming (ASP) to take advantage of the efficient solvers that exist for these problems (see, \eg \cite{DBLP:journals/tkde/GrecoGZ03,DBLP:journals/tods/EiterFGL08,DBLP:journals/tplp/MannaRT13,DBLP:journals/pvldb/KolaitisPT13,DBLP:conf/aaai/BienvenuBG14,DBLP:conf/sat/DixitK19} for relevant examples and \cite[Table 8]{DBLP:conf/rweb/Bourgaux25} for an overview).  
The two implementations of optimal repair-based semantics follow this path: the one by \citeauthor{DBLP:conf/kr/BienvenuB22}~\shortcite{DBLP:conf/kr/BienvenuB22} is based on the use of tractable approximations and SAT solvers, and that by \citeauthor{DBLP:conf/kr/BienvenuBIJ2025}~\shortcite{DBLP:conf/kr/BienvenuBIJ2025} uses ASP. A notable difference between the two systems is that the former focuses on the case where the conflicts (\ie the minimal sets of facts inconsistent \wrt the logical theory) are of size at most two while the latter handles conflicts of arbitrary size.

ASP (with disjunctive rules) can be used to model problems up to the second level of the polynomial hierarchy, meaning that semantics based on globally-optimal repairs could theoretically be implemented in ASP. However, this requires advanced techniques, such as saturation~\cite{DBLP:journals/amai/EiterG95}, which complicate modeling \textit{``beyond the capabilities of the common ASP laymen''}~\cite{DBLP:journals/tplp/GebserKS11}.  
Moreover, efficiency-wise it has been shown that in many cases of practical interest alternative formalisms are preferred to disjunctive ASP~\cite{DBLP:conf/lpnmr/AmendolaCRT22}. 
Among alternative formalisms, answer set programming with quantifiers (ASP(Q)) uses quantifiers over answer sets of ASP programs and provides a natural and compact way of modeling problems in the entire polynomial hierarchy \cite{DBLP:journals/tplp/AmendolaRT19}. 
ASP(Q) has already been used to tackle  
problems related to planning \cite{DBLP:journals/tplp/FandinnoLRSS21,DBLP:conf/padl/FaberMC22} and abstract argumentation \cite{DBLP:conf/iclp/000124}.

We investigate 
the use of ASP(Q) and ASP to implement optimal repair-based semantics. 
Compared to the previous ASP-based implementation by \citeauthor{DBLP:conf/kr/BienvenuBIJ2025}~\shortcite{DBLP:conf/kr/BienvenuBIJ2025}, we handle globally-optimal repairs thanks to 
ASP(Q) and improve the system performance by using tractable approximations in a similar way as \citeauthor{DBLP:conf/kr/BienvenuB22}~\shortcite{DBLP:conf/kr/BienvenuB22}. Moreover, we explore the use of another tractable approximation of the optimal repair-based semantics, called the grounded semantics \cite{DBLP:conf/kr/BienvenuB20}, which has never been implemented so far. Our experimental evaluation shows that using globally-optimal repairs is considerably more challenging
than for Pareto- or completion-optimal repairs, while at the same time 
demonstrating the utility of different optimizations, in particular, 
the surprising effectiveness of the grounded semantics. 

Proofs and additional details on the experiments are provided 
in the appendix. 
All materials to reproduce the experiments are available from \url{https://github.com/rjean007/InconsistentPrioritizedData-ASPQ}.

\section{Preliminaries}
In this section, we introduce the relevant background on inconsistency-tolerant semantics and ASP(Q).

\subsection{Optimal Repair-Based Semantics}\label{prelims:reps}
We recall the framework of inconsistency-tolerant querying of prioritized knowledge bases. 
All notions will be illustrated in Example~\ref{ex:prelimPrioKB}. 
Let $\preds$, $\inds$ and $\vars$ be three disjoint sets of predicates, constants and variables, respectively. We assume that each predicate has some arity $n\geq 1$, and let $\preds_n$ be the set of the $n$-ary predicates in $\preds$. 

\subsubsection{Knowledge Bases, Conflicts, Repairs}
A \emph{knowledge base (KB)} $\Kmc=(\Dmc,\Tmc)$ consists of a \emph{dataset} $\Dmc$ and a \emph{logical theory}~$\Tmc$: $\mathcal{D}$ is a finite set of \emph{facts} of the form $P(c_1, \ldots, c_n)$ with $P\in\preds_n$, $c_i\in\inds$ for $1 \leq i \leq n$, and $\Tmc$ is a finite set of first-order logic (FOL) sentences built from $\preds$, $\inds$ and~$\vars$.  
Typically, $\Tmc$ will be either an \emph{ontology} (\eg formulated in some description logic) or a set of \emph{database constraints}. 
In particular, we consider description logics of the \emph{DL-Lite family} \cite{calvaneseetal:dllite} and \emph{denial constraints} of the form $
\alpha_1 \wedge \ldots \wedge \alpha_n\rightarrow\bot$, where each $\alpha_i$ is a relational or inequality atom, 
which include 
\emph{functional dependencies (FDs)}. 
 A KB $\Kmc=(\Dmc,\Tmc)$ is \emph{consistent}, and $\Dmc$ is called \emph{$\Tmc$-consistent}, if $\Dmc \cup \Tmc$ has some model. Otherwise, $\Kmc$ is \emph{inconsistent}, denoted $\Kmc \models \bot$. 
A \emph{conflict} of $\Kmc=(\Dmc,\Tmc)$ is an inclusion-minimal subset $\Cmc \subseteq \Dmc$ such that $(\Cmc,\Tmc) \models \bot$. 
The set of conflicts of $\Kmc$ is denoted $\conflicts{\Kmc}$. 
A \emph{(subset) repair} of $\Kmc$ is an inclusion-maximal subset $\Rmc \subseteq \Dmc$ such that $(\Rmc,\Tmc) \not \models \bot$.  The set of repairs of $\Kmc$ is denoted $\reps{\Kmc}$.

\subsubsection{Prioritized KBs, Optimal Repairs} 
A \emph{priority relation} $\succ$ for a KB $\Kmc=(\Dmc,\Tmc)$ is an acyclic
binary relation over the facts of $\Dmc$ such that $\alpha\succ\beta$ implies that $\{\alpha,\beta\}\subseteq \Cmc$ for some $\Cmc\in\conflicts{\Kmc}$. 
It is \emph{total} if for every pair $\alpha\neq\beta$ such that $\{\alpha,\beta\}\subseteq \Cmc$ for some $\Cmc\in\conflicts{\Kmc}$, either $\alpha\succ\beta$ or $\beta\succ\alpha$. 
A \emph{completion} of $\succ$ is a total priority relation $\succ'\ \supseteq \ \,\succ$. 
A \emph{prioritized KB} $\Kmc_\succ$ is a KB $\Kmc$ with a priority relation $\succ$ for~$\Kmc$. 
Three kinds of optimal repairs are defined:

\begin{definition} 
Let $\Kmc_\succ$ be a prioritized KB with $\Kmc=(\Dmc,\Tmc)$ and $\Rmc \in \reps{\Kmc}$. 
\begin{itemize}
    \item A \emph{Pareto improvement} of $\Rmc$ is a $\Tmc$-consistent $\Bmc\subseteq\Dmc$ such that there is $ \beta\in\Bmc\setminus\Rmc$ with $\beta\succ\alpha$ for every $\alpha\in\Rmc\setminus\Bmc$. 
    \item A \emph{global improvement} of $\Rmc$ is a $\Tmc$-consistent $\Bmc\subseteq\Dmc$ such that $\Bmc \neq \Rmc$ and for every $\alpha \in \Rmc \setminus \Bmc$ there exists $\beta \in \Bmc \setminus \Rmc$ such that $\beta \succ \alpha$.
\end{itemize}
The repair $\Rmc$ is:
\begin{itemize}
    \item \emph{Pareto-optimal} if there is no Pareto improvement of $\Rmc$;
    \item \emph{globally-optimal} if there is no global improvement of $\Rmc$;
    \item \emph{completion-optimal} if $\Rmc$ is a globally-optimal repair of $\Kmc_{\succ'}$, for some completion $\succ'$ of $\succ$. 
\end{itemize}
We denote by $\preps{\Kmc_\succ}$, $\greps{\Kmc_\succ}$ and $\creps{\Kmc_\succ}$ the sets of  
Pareto-, globally- and completion-optimal repairs.
\end{definition}
\noindent It is known that $\creps{\Kmc_\succ} \subseteq \greps{\Kmc_\succ}
\subseteq\preps{\Kmc_\succ}$.

 \subsubsection{Queries, Repair-Based Semantics} 
A \emph{conjunctive query} (CQ) is a conjunction of 
atoms $P(t_1, \ldots, t_n)$ ($P\in\preds_n$, $t_i\in\inds\cup\vars$),
where some variables may be existentially quantified. 
Given a query $q(\vec{x})$, with free variables~$\vec{x}$, and a tuple of constants $\vec{a}$ such that $|\vec{a}|=|\vec{x}|$, $q(\vec{a})$ denotes the first-order sentence obtained by replacing each variable in~$\vec{x}$ by the corresponding constant in~$\vec{a}$. 
A \emph{(certain) answer} to $q(\vec{x})$ over $\Kmc$ is a tuple~$\vec{a}$ of constants such that $q(\vec{a})$ holds in every model of $\Kmc$, denoted $\Kmc \models q(\vec{a})$. 
When the KB is inconsistent, we consider the following alternative semantics, parameterized by the considered type of repair.
\begin{definition}
Fix $X\in \{S,P,G,C\}$ and consider a prioritized KB $\Kmc_\succ$ with $\Kmc=(\Dmc,\Tmc)$, query $q(\vec{x})$, and tuple of constants~$\ans$. 
Then $\ans$ is an answer to $q(\vec{x})$ over $\Kmc_\succ$ 
\begin{itemize}
    \item under \emph{X-brave semantics}, denoted $\Kmc_\succ \bravemodels{X} q(\ans)$, if $(\Rmc,\Tmc) \models q(\ans)$ for some $\Rmc \in \xreps{\Kmc_\succ}$;
    \item under \emph{X-AR semantics}, denoted $\Kmc_\succ \armodels{X} q(\ans)$, if $(\Rmc,\Tmc) \models q(\ans)$ for every $\Rmc \in \xreps{\Kmc_\succ}$;
    \item under \emph{X-IAR semantics}, denoted $\Kmc_\succ \iarmodels{X} q(\ans)$, if $(\Bmc,\Tmc) \models q(\ans)$ where $\Bmc=\bigcap_{\Rmc \in \xreps{\Kmc_\succ}} \Rmc$.
\end{itemize}
\end{definition}
\noindent It is known that $\Kmc_\succ \iarmodels{X} q \Rightarrow \Kmc_\succ \armodels{X} q  \Rightarrow \Kmc_\succ \bravemodels{X} q$. 
A \emph{cause} for $q(\vec{a})$ w.r.t.\ $\Kmc=(\Dmc,\Tmc)$ is an inclusion-minimal $\Tmc$-consistent subset $\Cmc \subseteq \Dmc$ such that $(\Cmc, \Tmc) \models q(\vec{a})$. 
The set of causes for $q(\vec{a})$ w.r.t.\ $\Kmc$ is denoted by $\causes{q(\ans),\Kmc} $. 
We will use the following characterizations of the semantics:
\begin{itemize}
    \item $\Kmc_\succ \bravemodels{X} q(\ans)$ iff there exist $\Rmc \in \xreps{\Kmc_\succ}$ and $\Cmc\in\causes{q(\ans),\Kmc}$ such that $\Cmc\subseteq\Rmc$;
    \item $\Kmc_\succ \not\armodels{X} q(\ans)$ iff there exists $\Rmc \in \xreps{\Kmc_\succ}$ such that for every $\Cmc\in\causes{q(\ans),\Kmc}$, $\Cmc\not\subseteq\Rmc$;
    \item $\Kmc_\succ \iarmodels{X} q(\ans)$ iff there exists $\Cmc\in\causes{q(\ans),\Kmc}$ such that $\Cmc\subseteq\bigcap_{\Rmc \in \xreps{\Kmc_\succ}} \Rmc$. 
\end{itemize}  
Theorems \ref{ub-thm} 
and \ref{lb-thm} summarize known 
results on the \emph{data complexity} (where the sizes of the logical theory $\Tmc$ and query $q(\vec{x})$ are assumed to be fixed) of query answering under these semantics (\cf survey \cite[Table 6]{DBLP:conf/rweb/Bourgaux25}).

\begin{theorem}\label{ub-thm} 
Let $\Lmc$ be an FOL fragment for which KB consistency and query entailment are in \ptime. 
Query entailment for $\Lmc$ KBs is in $\sigmaptwo$ under G-brave semantics, in $\piptwo$ under G-AR and G-IAR semantics, and for $X\in \{S,P,C\}$, it is in $\np$ under X-brave semantics, and in $\conp$ under X-AR and X-IAR semantics.  
\end{theorem}

\begin{theorem}\label{lb-thm}
Let $\Lmc$ be any FOL fragment that extends \dllitecore\ or FDs. 
Query entailment for $\Lmc$ KBs is $\sigmaptwo$-hard under G-brave semantics, $\piptwo$-hard under G-AR and G-IAR semantics, $\np$-hard under X-brave semantics for $X\in \{P,C\}$, and $\conp$-hard under X-AR semantics  for $X\in \{S,P,C\}$ and 
under X-IAR semantics for $X\in \{P,C\}$.
\end{theorem}

\subsubsection{Attack Relation, Grounded Semantics} 
The grounded semantics for prioritized KBs comes from the area of abstract argumentation. For prioritized KB $\Kmc_\succ$ with $\Kmc=(\Dmc,\Tmc)$, the \emph{attack relation} $\rightsquigarrow \subseteq (2^{\Dmc} \setminus \{\emptyset\}) \times \Dmc$ is defined by: 
$$\rightsquigarrow=\{(\Cmc\setminus\{\alpha\},\alpha) \mid \Cmc\in\conflicts{\Kmc}, \alpha\in\Cmc, \forall \beta \in \Cmc, \alpha \not \succ \beta\}.$$
We write $\Bmc\rightsquigarrow\alpha$ for $(\Bmc,\alpha)\in\rightsquigarrow$. The characteristic function $\Gamma:2^{\Dmc}\mapsto 2^{\Dmc}$ is defined by $$\Gamma(\Bmc) = \{\alpha \mid \Emc\rightsquigarrow\alpha\Rightarrow \exists \Fmc\subseteq\Bmc,\beta\in\Emc\text{ s.t. }\Fmc\rightsquigarrow \beta \}.$$ The \emph{grounded repair} of $\Kmc_\succ$ is the inclusion-minimal $\Gmc\subseteq\Dmc$ such that $\Gmc$ is $\Tmc$-consistent and $\Gmc=\Gamma(\Gmc)$, or equivalently, the least fixpoint of $\Gamma$. 
A tuple $\ans$ is an answer to $q(\vec{x})$ over $\Kmc_\succ$ under \emph{grounded semantics}, denoted $\Kmc_\succ\grmodels q(\ans)$, if $(\Gmc,\Tmc)\models q(\ans)$. 
It is known that $\Kmc_\succ\grmodels q(\ans)$ implies $\Kmc_\succ\iarmodels{P} q(\ans)$, so answers that hold under the grounded semantics hold under X-IAR for $X\in\{P,G,C\}$. 
Moreover, if $\Lmc$ is an FOL fragment for which the conflicts size is bounded independently from the data and KB consistency and query entailment are in \ptime, then grounded query entailment for $\Lmc$ KBs is in \ptime\ \cite{DBLP:conf/kr/BienvenuB20}.

\begin{example}\label{ex:prelimPrioKB}
Consider the KB $\Kmc=(\Dmc,\Tmc)$ with:
\begin{align*}
\Dmc=\{&A(a), B(a), C(a), D(a), A(b), B(b), C(b)\}\\
\Tmc=\{&A(x)\wedge B(x)\rightarrow \bot,\ 
B(x)\wedge C(x)\rightarrow \bot, \\&
C(x)\wedge D(x)\rightarrow \bot,\
D(x)\wedge A(x)\rightarrow \bot\}
\end{align*}
The conflicts and repairs of $\Kmc$ are as follows:
\begin{align*}
\conflicts{\Kmc}\!=\!\{&\{A(a),B(a)\}, \{B(a),C(a)\}, \{C(a),D(a)\}, \\&\{D(a),A(a)\}, \{A(b),B(b)\}, \{B(b),C(b)\}\}\\
\reps{\Kmc}\!=\!\{& \{A(a), C(a), B(b) \},\{B(a), D(a), B(b) \}\\
\{A(a), C&(a),A(b), C(b) \},\{B(a), D(a), A(b), C(b) \} \}
\end{align*}
Let us define a priority relation $\succ$ for $\Kmc$ by $A(a)\succ B(a)$, $C(a)\succ D(a)$, $A(b)\succ B(b)$, and $B(b)\succ C(b)$. 
\begin{align*}
&\creps{\Kmc_\succ}=\greps{\Kmc_\succ}=\{\{A(a), C(a),A(b), C(b) \}\}\\
&\preps{\Kmc_\succ}=\greps{\Kmc_\succ}\cup\{\{B(a), D(a), A(b), C(b) \}\}
\end{align*}
We thus obtain, \eg $\Kmc_\succ\armodels{G}A(a)$ while $\Kmc_\succ\not\armodels{P}A(a)$, $\Kmc_\succ\bravemodels{P}B(a)$ while $\Kmc_\succ\not\bravemodels{G}B(a)$, and $\Kmc_\succ\iarmodels{X}C(b)$ for $X\in\{P,G,C\}$. 
The grounded repair of $\Kmc_\succ$ is $\Gmc=\{A(b), C(b) \}$. Indeed, the attack relation $\rightsquigarrow$ is as written below, so $\Gamma(\emptyset)=\{A(b)\}$, since $A(b)$ is the only fact of $\Dmc$ that is not attacked, $\Gamma(\{A(b)\})=\{A(b),C(b)\}$, and $\Gamma(\{A(b),C(b)\})=\{A(b),C(b)\}$. 
\begin{align*}
&\{A(a)\}\rightsquigarrow B(a) & \{B(a)\}\rightsquigarrow C(a) && \{C(a)\}\rightsquigarrow B(a) \\  &\{C(a)\}\rightsquigarrow D(a) &
\{D(a)\}\rightsquigarrow A(a) && \{A(a)\}\rightsquigarrow D(a) \\  &\{A(b)\}\rightsquigarrow B(b) & \{B(b)\}\rightsquigarrow C(b)
\end{align*}
Note that $\Gmc$ is indeed included in the intersection of the Pareto-optimal repairs (which is actually equal to $\Gmc$ here).
\end{example}

\subsection{Answer Set Programming with Quantifiers}
We now recall basic notions about Answer Set Programming (ASP) and ASP with Quantifiers (ASP(Q)).
\subsubsection{ASP}
In ASP, \emph{atoms} take the form $\mt{p(t_1,\ldots,t_n)}$ where $\mt{p}$ is a predicate of arity $n\geq 0$ and each \emph{term} $\mt{t_i}$ is either a constant (integer or alphanumeric string starting with lowercase letter) or a variable (alphanumeric string starting with uppercase letter). 
A \emph{literal} is an atom $a$ or its negation $\mt{not}~a$, where $\mt{not}$ represents \emph{negation as failure}. 
It is \emph{negative} if it is of the form $\mt{not}~a$, \emph{positive} otherwise. 
An \emph{ASP program}\footnote{We consider core ASP, without disjunction in rule heads.} is a finite set of \emph{rules} of the form
$h\impl l_1,\ldots,l_n$ (with $n\geq 0$), whose \emph{head} $h$ is an atom, and whose \emph{body} $l_1,\ldots,l_n$ is interpreted as the conjunction of the literals $l_1,\ldots,l_n$. 
Every rule must be \emph{safe}, \ie each variable appearing in it appears in some positive body literal. 
A \emph{constraint} is a rule with an empty head ($\impl l_1,\ldots,l_n$), and a \emph{fact} is a rule with an empty body ($h\impl$). 
We also allow \emph{choice rules} of the forms
$\{h\} \impl l_1,\ldots,l_n$ and $1\{h_1;h_2\}1 \impl l_1,\ldots,l_n$: the former is used to choose to either add or omit $h$, and the latter enforces that  precisely one of $h_1$ and $h_2$ holds, when the rule body is satisfied. 
Choice rules are syntactic sugar which do not increase the expressive power  
but enable  
compact and intuitive modeling~\cite{DBLP:journals/tplp/CalimeriFGIKKLM20}. 
An ASP expression (atom, rule, program, etc.) is \emph{ground} if it contains no variable.

Given an ASP program $P$, the \emph{Herbrand Universe} of $P$ is the set $\mathcal{U}_P$ of constants appearing in $P$ and the \emph{Herbrand Base} of $P$ is the set $\mathcal{B}_P$ of ground atoms constructed from predicates of $P$ and constants in $\mathcal{U}_P$. 
We denote by $\mi{ground}(P)$ the set of all possible ground rules obtained from rules in $P$ by proper variable substitution with constants in $\mathcal{U}_P$ (cf.~\cite{DBLP:journals/tplp/CalimeriFGIKKLM20}). 
An \emph{interpretation} is a set of atoms $I\subseteq \mathcal{B}_P$.
A positive (resp.~negative) ground literal $l = a$ (resp.~$l=\mt{not}~a$) is true \wrt an interpretation $I$ if $a \in I$ (resp.~$a\notin I$), and false otherwise.
A conjunction of ground literals is true \wrt $I$ if all the literals are true \wrt $I$, and false otherwise. 
An interpretation $I$ satisfies a ground rule $r$ if the body of $r$ is false or its head is true \wrt $I$. 
It is a \emph{model} of an ASP program $P$ if it satisfies every $r \in \mi{ground}(P)$. 
The \emph{GL-reduct}~\cite{DBLP:journals/ngc/GelfondL91} of 
$P$ \wrt 
$I$ is the program $P^I$ obtained from $\mi{ground}(P)$ by $(i)$ removing each rule having at least one negative body literal false \wrt $I$; and $(ii)$ removing negative literals from the 
remaining rules.
An \emph{answer set} of $P$ is an interpretation $I$ such that $I$ is a subset-minimal model of $P^I$. 
We denote by $AS(P)$ the set of all answer sets of $P$. A program $P$ is \emph{coherent} if it has some answer set.

\subsubsection{ASP(Q)}
ASP with quantifiers (ASP(Q))~\cite{DBLP:journals/tplp/AmendolaRT19} extends ASP by allowing quantification over answer sets of different ASP programs. 
An \emph{ASP(Q) program} is an expression of the form:
\begin{equation}\label{eq:aspq}
    \Box_1 P_1 \ldots\Box_n P_n:C
\end{equation}
where $C$ is a stratified\footnote{Stratified programs respect some conditions that prevent default negation to be involved in recursion.} ASP program~\cite{DBLP:books/sp/CeriGT90} with constraints and for each $i \in \{1,\ldots,n\}$, $\Box_i \in \{\exists^{st},\forall^{st}\}$ is a quantifier and $P_i$ is an ASP program. 
Given an \ASPQ program $\Pi$ of form~(\ref{eq:aspq}), an ASP program $P$, and an interpretation $I$, define the following set of facts and constraints and ASP(Q) program, respectively: 
\begin{align*}
    \fix{P}{I}=&\{a\impl \mid a \in \mathcal{B}_{P}\cap I\} \cup \{\impl a \mid a \in \mathcal{B}_{P}\setminus I\}\\
    \Pi_{P,I}=&\Box_1 P_1 \cup \fix{P}{I}\ldots \Box_n P_n:C
\end{align*}
The \ASPQ semantics is defined inductively:
\begin{itemize}
    \item $\exists^{st} P:C$ is coherent if and only if there exists $M\in AS(P)$ such that $C\cup\fix{P}{M}$ is coherent;
    \item $\forall^{st} P:C$ is coherent if and only if for each $M\in AS(P)$, $C\cup\fix{P}{M}$ is coherent;
    \item $\exists^{st}P~\Pi$ is coherent if and only if there exists $M \in AS(P)$ such that $\Pi_{P,M}$ is coherent;
    \item $\forall^{st}P~\Pi$ is coherent if and only if for each $M\in AS(P)$, $\Pi_{P,M}$ is coherent.
\end{itemize}
The \emph{quantified answer sets} of an existential \ASPQ program $\exists^{st} P~\Pi$, with $\Pi$ of form~(\ref{eq:aspq}), are all $M \in AS(P)$ such that $\Pi_{P,M}$ is coherent. 
\begin{example}\label{ex:aspq}
Let $\Pi=\exists^{st} P_1\forall^{st} P_2:C$ where:

\begin{minipage}{.2\textwidth}
    $$
    P_1=\left\{\begin{array}{l}
        \mt{a\impl not~b}  \\
        \mt{b\impl not~a}  \\
        \mt{c\impl not~d}  \\
        \mt{d\impl not~c}  \\
    \end{array}\right\}
    $$    
    \end{minipage}
    \begin{minipage}{.2\textwidth}
    $$
    P_2=\left\{\begin{array}{l}
        \mt{e\impl a, not~f}  \\
        \mt{f\impl a, not~e}  \\
        \mt{\impl e, d}  \\
        \mt{\impl f, d}  \\
    \end{array}\right\}
    $$   
    \end{minipage}
    $$
    C=\left\{\begin{array}{l}
        \mt{\impl f}  \\
    \end{array}\right\}
    $$  
      Let us check whether $\Pi$ is coherent, \ie whether there is an answer set $M_1$ of $P_1$ such that \mbox{$\forall^{st} P_2\cup \fix{P_1}{M_1}:C$} is coherent. We have $AS(P_1)=\{\{\mt{a},\mt{c}\},\{\mt{a},\mt{d}\},\{\mt{b},\mt{c}\},\{\mt{b},\mt{d}\}\}$.
    Consider first $M_1 = \{\mt{a},\mt{c}\}$ and let $P_2' = P_2 \cup \fix{P_1}{M_1}$. $P'_2$ has two answer sets: $M_2' = \{\mt{a},\mt{c},\mt{e}\}$ and $M_2'' = \{\mt{a},\mt{c},\mt{f}\}$, so for $\forall^{st} P'_2:C$ to be coherent, both $C\cup \fix{P_2'}{M_2'}$ and $C\cup \fix{P_2'}{M_2''}$ have to be coherent. 
    However, $C\cup \fix{P_2'}{M_2''}$ is incoherent (it contains both $\mt{\impl f}$ and $\mt{f}\impl$).
    Thus, $M_1 = \{\mt{a},\mt{c}\}$ is not a quantified answer set of $\Pi$ and so not a witness for the coherence of $\Pi$. 
    Consider now $M_1 = \{\mt{b},\mt{c}\}$. Then $P_2' = P_2\cup\fix{P_1}{M_1}$ has exactly one answer set, $M_2 = \{\mt{b},\mt{c}\}$. In this case $C\cup \fix{P_2'}{M_2}$ is coherent as $\mt{f}$ is false \wrt $M_2$. Thus, $M_1 = \{\mt{b},\mt{c}\}$ is a quantified answer set of $\Pi$ and $\Pi$ is coherent.

    Let us now consider $\Pi'$ of the form $\exists^{st} P_1\forall^{st} P_2:C'$ where $C' = \{\mt{\impl not~fail}\}$.
    Observe that $\mt{fail}$ does not appear in any rule head of $P_1$, $P_2$, and $C'$. Hence, $C'\cup \fix{P_2\cup\fix{P_1}{M_1}}{M_2}$ will be incoherent for every choice of $M_1\in AS(P_1)$ and $M_2\in AS(P_2\cup\fix{P_1}{M_1})$, since $\mt{fail}$ is false \wrt any answer set of these programs. 
    Nonetheless, this does not imply that $\Pi'$ is incoherent as well. 
    Indeed, $M_1 = \{\mt{a},\mt{d}\}$ is a quantified answer set of $\Pi'$ because $P_2' = P_2\cup\fix{P_1}{M_1}$ is incoherent, so the universal quantification over answer sets of $P_2'$ is trivially satisfied.
\end{example}

\section{Encoding Semantics in ASP(Q)}\label{sec:encodings}

In this section, we present our approach to compute the answers that hold under the considered optimal repair-based semantics using ASP and ASP(Q) from their causes, the conflicts and the priority relation. 
Following \citeauthor{DBLP:conf/kr/BienvenuBIJ2025}~\shortcite{DBLP:conf/kr/BienvenuBIJ2025}, we assume that the input is given by ASP facts on the following predicates: \mt{conf} and \mt{cause} are unary predicates that store identifiers of the KB conflicts $\conflicts{\Kmc}$ and query answer causes $\causes{q(\ans),\Kmc}$, respectively, \mt{inConf} and \mt{inCause} are binary predicates such that \mt{inConf(C,A)} (\resp \mt{inCause(C,A)}) means that the fact with identifier \mt{A} belongs to the conflict (\resp cause) with identifier \mt{C}, 
and \mt{pref} is a binary predicate such that \mt{pref(A,B)} means that $\alpha\succ\beta$, where $\alpha$ and $\beta$ are the facts with identifiers \mt{A} and \mt{B} respectively.

\begin{table*}
\scalebox{0.85}{
\begin{tabular*}{1.15\textwidth}{l l }
\toprule
$\Pi_{\mi{ReachAll}}$ &
\mt{reachable(A) \impl inCause(C,A).} \\
&\mt{reachable(A) \impl inConf(C,A).}
\\
\midrule
$\Pi_{\mi{Attack}}$ & 
\mt{non\_attacking(C, A) \impl conf(C), inConf(C, A), inConf(C, B), A} != \mt{B, pref(A, B).} \\
& \mt{attacks(C, A) \impl conf(C), inConf(C, A), not~non\_attacking(C, A).}
\\
\midrule
$\Pi_{\mi{ReachS}}$ & $\Pi_{\mi{Attack}}$\\
& \mt{reachable(A) \impl cause(C), inCause(C,A).}\\
& \mt{reachable(A) \impl conf(C), attacks(C,B), reachable(B), inConf(C,A).}
\\
\midrule
$\Pi_{\mi{ReachW}}$ & \mt{weak\_attacks(C, A) \impl conf(C), inConf(C, A), inConf(C, B), A} != \mt{B, not~pref(A, B).}
\\
& \mt{reachable(A) \impl cause(C), inCause(C,A).}\\
& \mt{reachable(A) \impl conf(C), weak\_attacks(C,B), reachable(B), inConf(C,A).}
\\
\midrule
$\Pi_{\mi{ReachBin}}$ & \mt{reachable(A) \impl cause(C), inCause(C,A).}\\
&\mt{reachable(A) \impl conf(C), inConf(C,B), reachable(B), inConf(C,A), not~pref(B,A).}
\\
\midrule
$\Pi_{\mi{SubRep}}$&\mt{\{inRepair(A)\}\impl reachable(A)}.\\
&\mt{solved(C) \impl inConf(C,A), not~inRepair(A).}\\
&$ \impl $ \mt{conf(C), not~solved(C).}\\
\midrule
$\Pi_{\mi{Rep}}$&$\Pi_{\mi{SubRep}}$\\
&\mt{safe(C) \impl inConf(C, A), inConf(C, B), not~A = B, not~inRepair(A), not~inRepair(B).}\\
&\mt{keepOut(A) \impl inConf(C, A),not~inRepair(A), not~safe(C).}\\
&$ \impl $ \mt{reachable(A), not~inRepair(A), not~keepOut(A).}
\\
\midrule
$\Pi_{\mi{SatIfCause}}$&\mt{violatedCause(C) \impl inCause(C,A), not~inRepair(A).}\\
&\mt{sat \impl cause(C), not~violatedCause(C).}
\\
\midrule
$\Pi_{\mi{SomeCause}}$&$\Pi_{\mi{SatIfCause}}$\\
&$ \impl $ \mt{not~sat.}
\\
\midrule
$\Pi_{\mi{NoCause}}$&$\Pi_{\mi{SatIfCause}}$\\
&$ \impl $ \mt{sat.}
\\
\midrule
$\Pi_{\mi{GImp}}$ & \mt{\{global\_imp(A)\}\impl reachable(A)}.
\\
&\mt{solved\_global(C) \impl inConf(C,A), not~global\_imp(A).}\\
&$ \impl $ \mt{conf(C), not~solved\_global(C).}
\\
&\mt{impMinusRepair(A) \impl global\_imp(A), not~inRepair(A).}\\
&\mt{repairMinusImp(A) \impl inRepair(A), not~global\_imp(A).}\\
&\mt{diff \impl impMinusRepair(A).}\\
&\mt{diff \impl repairMinusImp(A).}\\
&\mt{fake\_improvement \impl not~diff.}\\
&\mt{ok(A) \impl repairMinusImp(A),impMinusRepair(B),pref(B,A).}\\
&\mt{fake\_improvement \impl repairMinusImp(A), not~ok(A).}\\
&\mt{global\_improvement \impl not~fake\_improvement.}
\\
&$ \impl $ \mt{not~global\_improvement.}
\\
\midrule
$\Pi_{\mi{POpt}}$ &$\Pi_{\mi{Attack}}$\\
&\mt{valid(A) \impl reachable(A), inRepair(A).}\\
&\mt{invalid\_att(C, A) \impl reachable(A), attacks(C, A), inConf(C, B), not~inRepair(B), not~A} = \mt{B.}\\
&\mt{valid(A) \impl reachable(A), conf(C), not~inRepair(A), attacks(C, A), not~invalid\_att(C, A).}\\
&$ \impl $ \mt{reachable(A), not~valid(A).}
\\
\midrule
$\Pi_{\mi{POptBin}}$ & \mt{valid(A) \impl reachable(A), inRepair(A).}\\
&\mt{valid(A) \impl reachable(A), conf(C), not~inRepair(A),  inConf(C, A), not~pref(A,B), inConf(C, B), inRepair(B).}\\
&$ \impl $ \mt{reachable(A), not~valid(A).}
\\
\midrule
$\Pi_{\mi{Compl}}$ & \mt{pref\_comp(A, B) \impl reachable(A), reachable(B), pref(A, B).}\\
&\mt{1\{pref\_comp(A, B); pref\_comp(B, A)\}1 \impl reachable(A), reachable(B), inConf(C, A), inConf(C, B),}\\& \hspace*{9cm}\mt{not~pref(A, B), not~pref(B, A), not~A} = \mt{B.}\\
&\mt{trans\_cl\_comp(A, B) \impl pref\_comp(A, B).}\\
&\mt{trans\_cl\_comp(A, B) \impl trans\_cl\_comp(A, Y), pref\_comp(Y, B).}\\
&$ \impl $ \mt{trans\_cl\_comp(A, A).}
\\
\midrule
$\Pi_{\mi{COpt}}$ & $\Pi_{\mi{Compl}}$\\
&\mt{valid(A) \impl reachable(A), inRepair(A).}\\
&\mt{invalid\_att(C, A) \impl reachable(A), not~inRepair(A), inConf(C, A), not~A} = \mt{B, inConf(C, B), not~inRepair(B).}\\
&\mt{invalid\_att(C, A) \impl reachable(A), not~inRepair(A), inConf(C, A), inConf(C, B), not~A} = \mt{B, pref\_comp(A, B).}\\
&\mt{valid(A) \impl reachable(A), not~inRepair(A), inConf(C, A), not~invalid\_att(C, A).}\\
&$ \impl $ \mt{reachable(A), not~valid(A).}
\\
\bottomrule
\end{tabular*}
}
\caption{Logic programs used to built the ASP(Q) programs that filter query answers that hold under X-brave or X-AR semantics from facts on predicates \mt{conf}, \mt{inConf}, \mt{pref}, \mt{cause} and \mt{inCause}. Intuitively, variables $\mt{A,B}$ are intended for some fact identifiers, and $\mt{C}$ for some conflict or cause identifier. 
}\label{tab:programs-semantics-asp(q)}
\end{table*}

\subsection{Naive Encodings}
The logic programs we assemble to build ASP(Q) encodings for semantics based on optimal repairs are given in Table~\ref{tab:programs-semantics-asp(q)}. 

 \subsubsection{G-Brave and G-AR}  
Using the programs from Table~\ref{tab:programs-semantics-asp(q)}, we build the following ASP(Q) programs for G-brave and G-AR semantics (where $\Pi_1\Pi_2$ stands for $\Pi_1\cup\Pi_2$):
\begin{align*}
\Pi^G_{\mi{brave}}=&\exists^{st} \Pi_{\mi{ReachAll}}\Pi_{\mi{Rep}}\Pi_{\mi{SomeCause}}\forall^{st} \Pi_{\mi{GImp}}:C\\
\Pi_{\mi{AR}}^{G}=&\exists^{st} \Pi_{\mi{ReachAll}}\Pi_{\mi{Rep}}\Pi_{\mi{NoCause}}\forall^{st} \Pi_{\mi{GImp}}:C
\end{align*}
with $C=\{\impl\mt{not~fail}\}$.

Let us first explain how $\Pi^G_{\mi{brave}}$ works. 
We want to check whether there exists $\Rmc\in\greps{\Kmc_\succ}$ such that $\Cmc\subseteq\Rmc$ for some $\Cmc\in\causes{q(\ans),\Kmc}$. 
First, applying the existential quantifier over $\Pi_{\mi{ReachAll}}\cup \Pi_{\mi{Rep}} \cup \Pi_{\mi{SomeCause}}$ serves to search for the existence of a repair that contains a cause. 
Since the facts that do not appear in any conflict belong to every repair, in what follows, we do not distinguish between actual repairs of $\Kmc$ and repairs of $(\Dmc',\Tmc)$, where $\Dmc'$ contains all facts that occur in some conflict or some cause (\ie facts whose identifiers are mentioned in our input ASP programs).
\begin{itemize}
\item 
The rules in $\Pi_{\mi{ReachAll}}$ compute the set of all such 
facts, 
encoded as atoms of the form $\mt{reachable(A)}$, where $\mt{A}$ is a fact identifier (we will see next how to restrict this set of facts using some notions of reachability, hence the name). 
\item The rules in $\Pi_{\mi{Rep}}$ compute a repair $\Rmc$ of the set of facts whose identifiers appear in the \mt{reachable(A)} atoms. 
The rules in $\Pi_{\mi{SubRep}}\subseteq \Pi_{\mi{Rep}}$ compute a $\Tmc$-consistent set $\Rmc$ of facts as follows. The choice rule in $\Pi_{\mi{SubRep}}$ guesses for each relevant fact $\alpha$ with identifier $\mt{A}$ whether $\alpha\in\Rmc$, which is encoded as $\mt{inRepair(A)}$. 
Remaining rules in $\Pi_{\mi{SubRep}}$ enforce $\Tmc$-consistency: $\mt{solved(C)}$ is derived if $\mt{C}$ is the identifier of a conflict which contains a fact $\alpha$ with identifier $\mt{A}$ such that $\alpha\notin\Rmc$, so the constraint $\mt{\impl conf(C),not~solved(C)}$ imposes that at least one fact from each conflict is excluded. 
The rules in $\Pi_{\mi{Rep}}\setminus\Pi_{\mi{SubRep}}$ ensure $\subseteq$-maximality. 
Specifically, they identify each conflict $\Cmc$ with identifier $\mt{C}$ that contains at least two facts not in $\Rmc$, encoded by $\mt{safe(C)}$, and for every non-safe conflict $\Cmc$, they mark the only fact $\alpha\in\Cmc\setminus\Rmc$ with identifier $\mt{A}$ using \mt{keepOut(A)}. 
The last constraint imposes that every $\alpha\notin\Rmc$ is marked by \mt{keepOut}, meaning that there is a conflict $\Cmc$ such that $\alpha\in\Cmc$ and $\Cmc\setminus\{\alpha\}\subseteq \Rmc$, \ie $\Rmc\cup\{\alpha\}$ is $\Tmc$-inconsistent.
\item The rules in $\Pi_{\mi{SomeCause}}$ ensure $\Rmc$ contains some~cause. 
The rules in $\Pi_{\mi{SatIfCause}}$ derive \mt{violatedCause(C)} for every cause $\Cmc$ with identifier \mt{C} such that $\Cmc\not\subseteq\Rmc$, and the atom $\mt{sat}$ is derived if there exists some cause $\Cmc$ for which \mt{violatedCause(C)} is not derived (hence $\Cmc\subseteq\Rmc$). 
The constraint in $\Pi_{\mi{SomeCause}}\setminus\Pi_{\mi{SatIfCause}}$ imposes that $\mt{sat}$ is true, so at least one cause must be included.
\end{itemize}
Hence, an answer set of $\Pi_{\mi{ReachAll}}\cup \Pi_{\mi{Rep}}\cup \Pi_{\mi{SomeCause}}$ corresponds to some $\Rmc\in\reps{\Kmc}$ with $(\Rmc,\Tmc)\models q(\ans)$. 
The second program of $\Pi^G_{\mi{brave}}$, $\Pi_{\mi{GImp}}$, is then used to search for a global improvement of $\Rmc$.
\begin{itemize}
    \item The three first rules guess a $\Tmc$-consistent set of facts $\Bmc$, in the same way as $\Pi_{\mi{SubRep}}$ did, using $\mt{global\_imp}$ to store the identifiers of the facts in $\Bmc$. The remaining rules verify whether $\Bmc$ is indeed a global improvement of $\Rmc$. 
   \item 
    Two rules compute $\Bmc\setminus\Rmc$ and $\Rmc\setminus\Bmc$, using predicates $\mt{impMinusRepair}$ and $\mt{repairMinusImp}$, respectively. This is used to (i) check whether $\Bmc\neq \Rmc$ with the two rules that derive $\mt{diff}$ if either $\Bmc\setminus\Rmc\neq\emptyset$ or $\Rmc\setminus\Bmc\neq\emptyset$ and (ii) check whether for every $\alpha \in \Rmc \setminus \Bmc$, there exists $\beta \in \Bmc \setminus \Rmc$ such that $\beta \succ \alpha$ with the rule that derives $\mt{ok(A)}$ if $\alpha\in \Rmc\setminus\Bmc$ with identifier $\mt{A}$ satisfies this condition. 
 The atom $\mt{fake\_improvement}$ is derived when (i) or (ii) is not satisfied, \ie if $\mt{diff}$ is false or if $\mt{ok(A)}$ is false for the identifier \mt{A} of some $\alpha\in\Rmc\setminus\Bmc$. 
Hence, if $\mt{fake\_improvement}$ is false, 
$\Bmc$ is a global improvement of $\Rmc$, and $\mt{global\_improvement}$ is derived. 

   \item The final constraint enforces that $\mt{global\_improvement}$ has been derived, so that $\Bmc$ is a global improvement of $\Rmc$. 
\end{itemize}
Finally, $C$ contains only the constraint $\mt{\impl not~fail}$ and $\mt{fail}$ does not appear in the head of any rule. Hence, as explained in Example~\ref{ex:aspq}, if we let $P_1=\Pi_{\mi{ReachAll}}\cup \Pi_{\mi{Rep}} \cup \Pi_{\mi{SomeCause}}$ and $P_2=\Pi_{\mi{GImp}}$, a quantified answer set $M$ of $\Pi^G_{\mi{brave}}$ must be an answer set of $P_1$ such that $P_2\cup\fix{P_1}{M}$ is incoherent.  
The next proposition follows.

\begin{restatable}{proposition}{PropBrave}
    $\mathcal{K}_{\succ} \bravemodels{G} q(\vec{a})$ iff 
    $\Pi^G_{\mi{brave}}$ is coherent.
\end{restatable}

The program for G-AR, $\Pi_{\mi{AR}}^{G}$, is exactly as $\Pi^G_{\mi{brave}}$ except that $\Pi_{\mi{SomeCause}}$ is replaced by $\Pi_{\mi{NoCause}}$, which ensures that the repair $\Rmc$ built by $\Pi_{\mi{ReachAll}}\cup \Pi_{\mi{Rep}}$ does not contain any cause for $q(\ans)$. Indeed, $\mt{\impl sat}\in \Pi_{\mi{NoCause}}$ requires that $\mt{sat}$ is not derived by $\Pi_{\mi{SatIfCause}}$, so that there is no cause $\Cmc$ such that $\Cmc\subseteq\Rmc$.  

\begin{proposition}
    $\mathcal{K}_{\succ} \armodels{G} q(\vec{a})$ iff 
    $\Pi_{\mi{AR}}^{G}$ is incoherent.
\end{proposition}

We also considered an alternative program of the form $\forall^{st}P\exists^{st}P':C'$ for G-AR, but as it was less efficient in practice, we present it only in 
the appendix.

 \subsubsection{P- and C- Brave and AR} For $X\in\{P,C\}$, since the data complexity of query answering under X-brave and X-AR semantics is in the first level of the polynomial hierarchy, we 
 use the following plain ASP programs, which verify whether there exists an optimal repair that either contains some cause or does not contain any cause for the query.  
 \begin{align*}
 \Pi_{\mi{brave}}^X=& \Pi_{\mi{ReachAll}}\cup\Pi_{\mi{SubRep}}\cup\Pi_{\mi{XOpt}}\cup\Pi_{\mi{SomeCause}}\\
\Pi_{\mi{AR}}^X=&\Pi_{\mi{ReachAll}}\cup\Pi_{\mi{SubRep}}\cup\Pi_{\mi{XOpt}}\cup\Pi_{\mi{NoCause}}
\end{align*}
As explained in the case $X=G$, $\Pi_{\mi{ReachAll}}\cup\Pi_{\mi{SubRep}}$ guesses a $\Tmc$-consistent set $\Rmc$ of facts. 
The programs $\Pi_{\mi{POpt}}$ and $\Pi_{\mi{COpt}}$ ensure that $\Rmc$ is a Pareto- or completion-optimal repair, respectively, and were used by \citeauthor{DBLP:conf/kr/BienvenuBIJ2025}~\shortcite{DBLP:conf/kr/BienvenuBIJ2025}.
\begin{itemize}
\item The rules in $\Pi_{\mi{Attack}}\subseteq \Pi_{\mi{POpt}}$ compute the attack relation $\rightsquigarrow$. They produce $\mt{attacks(C,A)}$ if $\mt{C}$ and \mt{A} are identifiers of a conflict $\Cmc$ and fact $\alpha$ such that $\Cmc\setminus\{\alpha\}\rightsquigarrow\alpha$. 
The next rules in $\Pi_{\mi{POpt}}$ derive $\mt{valid(A)}$ for the identifier $\mt{A}$ of a fact $\alpha$ if (i) $\alpha\in\Rmc$, or (ii) $\alpha$ is attacked by $\Cmc\setminus\{\alpha\}$ for some conflict $\Cmc$  
such that $\Cmc\setminus\{\alpha\}\subseteq\Rmc$. 
The constraint $\mt{\impl reachable(A), not~valid(A)}$ thus ensures that every $\alpha\notin\Rmc$ is attacked by a set of facts included in $\Rmc$. This means that $\Rmc$ is $\subseteq$-maximal and Pareto-optimal: otherwise, there would be $\alpha\notin\Rmc$ such that $\Rmc\cup\{\alpha\}\setminus\{\beta\mid\alpha\succ\beta\}$ is $\Tmc$-consistent and $\alpha$ would not be attacked by any subset of $\Rmc$.

\item The rules in $\Pi_{\mi{Compl}}\subseteq\Pi_{\mi{COpt}}$ guess a completion $\succ'$ of $\succ$, using $\mt{pref\_comp(A,B)}$ to indicate that $\alpha\succ'\beta$, where $\alpha$ and $\beta$ have identifiers $\mt{A}$ and $\mt{B}$. 
The next rules in $\Pi_{\mi{COpt}}$ derive $\mt{valid(A)}$ for the identifier $\mt{A}$ of a fact $\alpha$ if (i) $\alpha\in\Rmc$, or (ii) $\alpha$ is in a conflict $\Cmc$ such that  $\Cmc\setminus\{\alpha\}\subseteq \Rmc$ and  $\alpha\not\succ'\beta$ for every $\beta\in\Cmc\setminus\{\alpha\}$. 
The constraint $\mt{\impl reachable(A), not~valid(A)}$ thus ensures that every $\alpha\notin\Rmc$ is attacked (\wrt the attack relation defined \wrt $\succ'$) by a set of facts included in $\Rmc$.
\end{itemize}
Finally, as before, 
$\Pi_{\mi{SomeCause}}$ (\resp $\Pi_{\mi{NoCause}}$) enforces that $\Rmc$ contains a cause (\resp does not contain any cause). 
 
\begin{proposition}
    For $X\in\{P,C\}$, $\mathcal{K}_{\succ} \bravemodels{X} q(\vec{a})$ iff $\Pi_{\mi{brave}}^X$ is coherent, and $\mathcal{K}_{\succ} \armodels{X} q(\vec{a})$ iff $\Pi_{\mi{AR}}^X$ is incoherent.
\end{proposition}

\subsubsection{X-IAR} 
We compute 
$\bigcap_{\Rmc\in\xreps{\Kmc_\succ}}\Rmc$ by checking for each fact 
whether it holds under X-AR (since a fact holds under X-AR iff it is in every optimal repair). It is then possible to use this set to evaluate the queries under X-IAR.

\begin{table*}
\scalebox{0.85}{
\begin{tabular*}{1.15\textwidth}{l l }
\toprule
$\Pi_{\Gamma(\emptyset)}$ &  \mt{unsafe(A, 0) \impl attacks(C, A).}\\
&\mt{safe(A, 0) \impl inConf(C, A), not~unsafe(A, 0).}
\\
\midrule
$\Pi_{\Gamma incr}$ & \mt{safe(A, t) \impl safe(A, t-1).}
\\
       & \mt{non\_subset(C, A, t) \impl conf(C), inConf(C, A), inConf(C, B), A} != \mt{B, not~safe(B, t-1).}
\\
       & \mt{subset(C, A, t) \impl conf(C), inConf(C, A), not~non\_subset(C, A, t).}
\\
        & \mt{protected(C, A, t) \impl attacks(C, A), inConf(C, B), A} != \mt{B, attacks(C2, B), subset(C2, B, t).} 
\\
       & \mt{unsafe(A, t) \impl attacks(C, A), not~protected(C, A, t).} 
\\
       & \mt{safe(A, t) \impl inConf(C, A), not~unsafe(A, t).}
\\
       & \mt{continue(t) \impl safe(A, t), not~safe(A,t-1).}\\
\bottomrule
\end{tabular*}
}
\caption{Logic programs used to compute the grounded repair from facts on predicates \mt{conf}, \mt{inConf}, and \mt{attacks} (computed by $\Pi_{\mi{Attack}}$).}\label{tab:programs-grounded}
\end{table*}

\subsection{Localization}
To avoid considering the whole dataset in the encodings, we \emph{localize} them to relevant facts, defined from the query causes using some notions of \emph{reachability} \wrt the conflicts and priority relation. We define two such notions of reachability (strong and weak). 
The first one was already 
used to localize SAT or ASP encodings by \citeauthor{DBLP:conf/kr/BienvenuB22}~\shortcite{DBLP:conf/kr/BienvenuB22} and \citeauthor{DBLP:conf/kr/BienvenuBIJ2025}~\shortcite{DBLP:conf/kr/BienvenuBIJ2025}, and the proofs of the theorems below are strongly inspired by the proofs of correctness of the SAT encodings for non-binary conflicts given in the extended version of \cite{DBLP:conf/kr/BienvenuB22}. 
\begin{definition}
Given a prioritized KB $\Kmc_\succ$ with $\Kmc=(\Dmc,\Tmc)$  in $\Kmc_\succ$ and $\Bmc\subseteq\Dmc$, $\reachstrong(\Bmc)$ is the set of facts reachable from $\Bmc$ in the directed hypergraph\footnote{A node $\gamma$ is reachable from a set of nodes $\Bmc$ in a directed hypergraph if $\gamma\in\Bmc$ or there exists a node $\beta$ reachable from $\Bmc$ and some edge $(\beta,\Emc)$ such that $\gamma\in\Emc$.} whose edges are $\{(\alpha,\Emc)\mid \Emc\rightsquigarrow\alpha\}=\{(\alpha,\Emc)\mid \Emc\cup\{\alpha\}\in\conflicts{\Kmc}, \forall \beta \in \Emc, \alpha \not \succ \beta\}$ and $\reachweak(\Bmc)$ is the set of facts reachable from $\Bmc$ in the directed hypergraph whose edges are $\{(\alpha,\Emc)\mid \Emc\cup\{\alpha\}\in\conflicts{\Kmc}, \exists \beta \in \Emc, \alpha \not \succ \beta\}$. 
\end{definition} 
Note that $\reachstrong(\Bmc)\subseteq \reachweak(\Bmc)$ and that the these two sets coincide when the conflicts are \emph{binary}: in this case, reachability is done in the oriented graph whose edges are $\{(\alpha,\beta)\mid \{\alpha,\beta\}\in\conflicts{\Kmc}, \alpha\not\succ\beta\}$. 

\begin{restatable}{theorem}{ThmLocalizationGlobal}\label{ThmLocalizationGlobal}
Let $X\in\{P,G,C\}$, $\Bmc\subseteq\Dmc$, $\Bmc_r=\reachweak(\Bmc)$, $\Kmc^{\Bmc_r}=(\Bmc_r,\Tmc)$ and $\succ^{\Bmc_r}$ be the restriction of $\succ$ to $\Bmc^r$. 
\begin{itemize}
\item If $\Rmc\in\xreps{\Kmc_\succ}$, then $\Rmc\cap\Bmc_r\in\xreps{\Kmc^{\Bmc_r}_{\succ^{\Bmc_r}}}$.
\item If $\Rmc\in \xreps{\Kmc^{\Bmc_r}_{\succ^{\Bmc_r}}}$, then there exists $\Rmc'\in\xreps{\Kmc_\succ}$ such that $\Rmc=\Rmc'\cap\Bmc_r$. 
\end{itemize}
\end{restatable}

\begin{restatable}{theorem}{ThmLocalizationParetoComp}\label{ThmLocalizationParetoComp} 
If $\Bmc_r=\reachweak(\Bmc)$ is replaced by $\Bmc_r=\reachstrong(\Bmc)$, Theorem~\ref{ThmLocalizationGlobal} still holds for $X\in\{P,C\}$.
\end{restatable}

We can thus replace $\Pi_{\mi{ReachAll}}$ by $\Pi_{\mi{ReachW}}$ or $\Pi_{\mi{ReachS}}$ (if $X\in\{P,C\}$) in the naive encodings, in order to restrict the set of facts from $\Dmc$ considered (whose identifiers are stored in \mt{reachable}). Indeed, these programs compute $\reachstrong(\Bmc)$ and $\reachweak(\Bmc)$ for $\Bmc=\bigcup_{\Cmc\in\causes{q(\ans),\Kmc}}\Cmc$, respectively. 
We leave open whether Theorem~\ref{ThmLocalizationParetoComp} holds for $X=G$.

\subsection{Simplified Encoding for Binary Conflicts}

When the conflicts are of size exactly 2 (we assume that any self-inconsistent fact has been removed from the dataset), we can simplify the encodings as follows. 

\subsubsection{Reachable Facts} First, as explained earlier, the set of relevant (\mt{reachable}) facts is easier to compute with binary conflicts: we replace $\Pi_{\mi{ReachW}}$ or $\Pi_{\mi{ReachS}}$ by $\Pi_{\mi{ReachBin}}$. 

\subsubsection{Repairs} In the encodings for semantics based on globally-optimal repairs, which  build a full, $\subseteq$-maximal, repair, we modify $\Pi_{\mi{Rep}}$ by replacing 
the two rules of $\Pi_{\mi{Rep}}$ that compute the facts that are necessary to keep out of the repair 
by the simpler rules (with one negation instead of four):
\begin{align*}
\mt{dangerous(C)}\text{ $ \impl $ }&\mt{conf(C),inConf(C, A), inRepair(A).}\\
\mt{keepOut(A)}\text{ $ \impl $ }&\mt{dangerous(C),inConf(C, A),}\\&\mt{not~inRepair(A).}
\end{align*}

\subsubsection{Pareto-Optimality} We use again the fact that when the conflicts are binary the attack relation is straightforward ($\{\beta\}\rightsquigarrow\alpha$ iff $\{\alpha,\beta\}\in\conflicts{\Kmc}$ and $\alpha\not\succ\beta$), and that including a single conflicting fact is sufficient to exclude a fact, to replace $\Pi_{\mi{POpt}}$ by $\Pi_{\mi{POptBin}}$.

\subsection{Approximations}

Another way to answer queries under optimal repair-based semantics more efficiently is to use \emph{approximations} to compute subsets or supersets of the answers and thus avoid relying on more demanding programs for some answers. 

\subsubsection{Tractable Under-Approximations of P-IAR} 
We use a preprocessing step that  identifies a subset of the answers that hold under P-IAR (hence under all the considered semantics). 
We consider two tractable such under-approximations: the \emph{grounded semantics} (recalled in Section~\ref{prelims:reps}) and the \emph{trivially P-IAR answers} \cite{DBLP:conf/aaai/BienvenuBG14,DBLP:conf/kr/BienvenuB22}, which have some cause such that none of its fact is attacked (\wrt 
$\rightsquigarrow$). Note that the trivially P-IAR answers are 
those that hold \wrt $\Gamma(\emptyset)$. 
Table~\ref{tab:programs-grounded} shows the ASP programs used to compute $\Gamma(\emptyset)$ then incrementally compute the grounded repair. 
Using incremental solving mirrors the fixpoint definition of grounded semantics and allows for a very direct encoding.  
Once we have computed one of these two sets, we use it to filter the answers that have some cause included in it. 

\subsubsection{P-AR and P-Brave} We also investigate the use of semantics whose data complexity is in the first level of the polynomial hierarchy, \eg based on Pareto-optimal repairs, to obtain lower or upper bounds on semantics based on globally-optimal repairs: $\Kmc_\succ\armodels{P} q(\ans)$ implies $\Kmc_\succ\armodels{G} q(\ans)$ (hence also $\Kmc_\succ\bravemodels{G} q(\ans)$) and $\Kmc_\succ\not\bravemodels{P} q(\ans)$ implies $\Kmc_\succ\not\bravemodels{G} q(\ans)$ (hence also $\Kmc_\succ\not\armodels{G} q(\ans)$).

\section{Experiments}
Our experimental evaluation aims at (i) evaluating the impact of adopting globally-optimal repairs, (ii) assessing the grounded semantics, and (iii) comparing the different approaches of the same problem in terms of runtime. In more detail, we consider the following questions:
\begin{itemize}
\item What proportion of the intersection of optimal repairs is given by the grounded repair? And by the trivially P-IAR facts (\ie $\Gamma(\emptyset)$)? 
What is the overhead in term of runtime to compute the grounded repair instead of $\Gamma(\emptyset)$?

\item Given a semantics, what is the impact of 
localization or binary conflicts-specific encodings?

\item What is the impact of using globally-optimal repairs instead of Pareto- or completion-optimal ones, both in terms of the answers obtained and runtime overhead?

\end{itemize}

\subsection{Experimental Setting}

\begin{table}\footnotesize
\begin{tabular*}{0.47\textwidth}{l@{\extracolsep{\fill}}lrrrrr}
\toprule
 &&$\!\cap_{S}\!$& $\!\Gamma^1\!\setminus\!\cap_{S}\!$ & $\!\Gamma^2\!\setminus\!\Gamma^1\!$ & $\!\Gamma^3\!\setminus\!\Gamma^2\!$ & $\!\Gamma^4\!\setminus\!\Gamma^3\!$ \\
 \midrule
\mn{u1c1} & \multirow{2}{*}{$\ \succ^{ns}$} & 73351 &921 & 1365 & 0 & 0
\\
\mn{u1c50} &&43779& 7289 & 21319 & 1329 & 8
\\
\midrule
\mn{u1c1} &\multirow{2}{*}{$\ \succ^{ss}$}&73351& 1225 & 881 & 0 & 0
\\
\mn{u1c50} &&43779& 8493 & 9368 & 11 & 0
\\
\midrule
\mn{u1c1} &\multirow{2}{*}{$\ \succ^{nb}$}& 73131 & 989 & 1406 & 4 & 0
\\
\mn{u1c50} && 43612 & 10777 & 18629 & 52 & 0
\\
\bottomrule
\end{tabular*}
\caption{Number of facts in the different parts of the grounded repair: $\cap_{S}=\bigcap_{\Rmc\in\reps{\Kmc}} \Rmc$, $\Gamma^i=\Gamma^i(\emptyset)$.
}\label{tab:grounded-repair-incremental}
\end{table}

\begin{table}\footnotesize
    \begin{tabular*}{0.47\textwidth}{l@{\extracolsep{\fill}}rrrr}
        \toprule
        & \multicolumn{2}{c}{$\succ^{ss}$} & \multicolumn{2}{c}{$\succ^{ns}$} \\
        & $\!\Gmc\!\setminus\!\cap_{S}\!$ & $\!\Gamma^1\!\setminus\!\cap_{S}\!$ 
        & $\!\Gmc\!\setminus\!\cap_{S}\!$ & $\!\Gamma^1\!\setminus\!\cap_{S}\!$ \\
        \midrule
        \mn{u1c1}  & 1.58   & 0.16 & 2.10   & 0.16 \\
        \mn{u1c20} & 14.92  & 1.26 & 21.22  & 1.39 \\
        \mn{u1c50} & 61.80  & 3.53 & 126.39 & 3.80 \\
        \midrule
        \mn{u5c1}  & 7.85   & 0.68 & 10.70  & 0.68 \\
        \mn{u5c20} & 129.31 & 7.12 & 240.34 & 7.11 \\
        \mn{u5c50} & 304.11 & 14.55 & oom & 14.69 \\
        \midrule
        \mn{u20c1} & 23.75  & 2.41 & 47.05  & 2.39 \\
        \mn{u20c20} & oom  & 36.07 & oom  &  33.13\\
        \bottomrule
    \end{tabular*}
    \caption{Time (in seconds) to compute the facts that belong to the grounded repair $\Gmc$ or to $\Gamma^1=\Gamma(\emptyset)$ (among those that belong to some conflict) from the precomputed attack relation.}
    \label{tab:time-grounded-compact}
\end{table}

\begin{table}[t]\footnotesize
\begin{tabular*}{0.47\textwidth}{l@{\extracolsep{\fill}}lrrrrr}
\toprule
&&$\!\cap_{S}\!$& $\!\Gmc\!\setminus\!\cap_{S}\!$ & $\!\cap_{P}\!\setminus\!\Gmc\!$ & $\!\cap_{G}\!\setminus\!\cap_{P}\!$ & $\!\cap_{C}\!\setminus\!\cap_{G}\!$ \\
\midrule
\mn{u1c1} &\multirow{3}{*}{ $\ \succ^{ns}$ }&73351 &2286 & 0 & 0 & 0
\\
\mn{u1c10} &&65310& 10134 & 3 & 22 & 1
\\
\mn{u1c20}& &56646 &18507 & 82 & 46 & 7
\\
\bottomrule
\end{tabular*}
\caption{
Size of grounded repair $\Gmc$ and optimal repair intersections: $\cap_{X}=\bigcap_{\Rmc\in\xreps{\Kmc_\succ}}\Rmc$ for $X\in\{S, P,G,C\}$.
}\label{tab:grounded-vs-iar-facts}
\end{table}

Following \citeauthor{DBLP:conf/kr/BienvenuBIJ2025}~\shortcite{DBLP:conf/kr/BienvenuBIJ2025}, 
we translated into ASP programs a subset of the \mn{ORBITS} benchmark \cite{DBLP:conf/kr/BienvenuB22} which provides several conflict sets, priority relations, and potential answers associated with their causes built from the 
\mn{CQAPri} benchmark \cite{DBLP:phd/hal/Bourgaux16}, a synthetic benchmark adapted from LUBM$^\exists_{20}$ \cite{DBLP:conf/semweb/LutzSTW13} to evaluate inconsistency-tolerant query answering over DL-Lite KBs.   
To experiment also with non-binary conflicts, 
\citeauthor{DBLP:conf/kr/BienvenuBIJ2025}~\shortcite{DBLP:conf/kr/BienvenuBIJ2025} added a denial constraint that yields conflicts of size 10 and built some priority relations for this case using preference rules. 
\subsubsection{Datasets} The datasets of the \mn{CQAPri} benchmark are named 
\mn{uXcY}, with \mn{X} and \mn{Y} related to the size and the proportion of facts involved in some conflicts respectively, and are such that $\mn{uXcY}\subseteq\mn{uXcY'}$ for $\mn{Y}\leq\mn{Y'}$ and $\mn{uXcY}\subseteq\mn{uX'cY}$ for $\mn{X}\leq\mn{X'}$. 
Since our focus is on the use of the more demanding globally-optimal repairs, we mostly use the smallest datasets \mn{u1cY} with $\mn{Y}\in\{1, 5, 10, 20, 30, 50\}$, which contain from 75K to 78K facts. 
The proportion of facts involved in some (binary) conflict in these datasets ranges from 3\% to 44\% and the corresponding conflict sets contain from 2K to 81K conflicts, involving 2K to 34K facts. Forty non-binary conflicts are generated (in all datasets) when the additional denial constraint is considered. 
We also run a few experiments with some larger datasets \mn{u5cY} with $\mn{Y}\in\{1, 5, 10, 20\}$, with 12K to 231K conflicts involving 12K to 137K facts.  
\subsubsection{Priority Relations}In the binary conflicts case, we use the two priority relations of the \mn{ORBITS} benchmark: $\succ^{ss}$ is built from priority levels, hence is such that globally-, Pareto- and completion-optimal repairs coincide, while $\succ^{ns}$ is not score-structured and allows us to distinguish between the three kinds of repairs. Moreover, $\succ^{ss}$ assigns a priority between the two facts of about 80\% of the conflicts and this proportion is about 60\%-65\% for $\succ^{ns}$. For the non-binary conflicts case, we use a priority relation $\succ^{nb}$ resulting from some preference rules given by \citeauthor{DBLP:conf/kr/BienvenuBIJ2025}~\shortcite{DBLP:conf/kr/BienvenuBIJ2025} which assigns a priority between two facts that belong to some conflict in about 90\% of the cases. 
\subsubsection{Queries}
We use the 8 queries \citeauthor{DBLP:conf/kr/BienvenuBIJ2025}~\shortcite{DBLP:conf/kr/BienvenuBIJ2025} selected from the \mn{CQAPri} benchmark for having a lower number of potential answers.  The total number of potential answers over all queries ranges from 1,525 (on \mn{u1c1}) to 8,206 (on \mn{u5c20}).

\subsubsection{Setup}
All experiments were executed on a machine equipped with an Intel(R) Xeon(R) CPU E7-8880 v4 @ 2.20GHz, running Debian GNU/Linux 12, with memory and CPU (i.e., user+system) limited to 8GB and 600s. Time and memory usage have been measured with pyrunlim\footnote{pyrunlim is available at \url{https://github.com/alviano/python.git}}. 
As \ASPQ solver we used the \textsc{casper} system~\cite{DBLP:conf/aaai/CuteriMR26}\footnote{Other \ASPQ systems are available, such as \textsc{qasp}~\cite{DBLP:conf/lpnmr/AmendolaCRT22} and \textsc{pyqasp}~\cite{DBLP:journals/tplp/FaberMR23}, but \textsc{casper} performed better in our preliminary evaluation.}, and as ASP solver we used \textsc{clingo} \cite{DBLP:journals/corr/GebserKKS17}.

\subsubsection{Preprocessing}Reported times exclude the computation of the attack relation 
using $\Pi_{\mi{Attack}}$, as it can be considered a query-independent preprocessing task. 
Computing the attack relation took at most 5 seconds for the \mn{u1cY} datasets (small datasets), about 30 seconds for \mn{u5c20} (
largest dataset we used for the query answering task) and up to about one hour
for 
\mn{u20c50} (dataset with 2M facts, 46\% of facts involved in some conflict, and 3M conflicts, which we considered only for the task of computing the grounded repair).

\begin{table*}
    \centering\footnotesize
    \resizebox{\textwidth}{!}{    
    \begin{tabular}{ll|rr|r|rlrlrl|rlrlrl}
        \hline
           &                      & Triv & GR$\setminus$Triv  & Pot$\setminus$GR  & \multicolumn{2}{c}{P-AR$\setminus$GR}    & \multicolumn{2}{c}{G-AR$\setminus$GR}  & \multicolumn{2}{c|}{C-AR$\setminus$GR}  & \multicolumn{2}{c}{P-brave$\setminus$GR}   & \multicolumn{2}{c}{G-brave$\setminus$GR}   & \multicolumn{2}{c}{C-brave$\setminus$GR}\\
        \hline
        \mn{u1c1} & \multirow{5}{*}{$\succ^{ns}$}                     & 1465	       & 59	     &   1    & 0 &(0)	            & 0 &(0)	            & 0 &(0)	            & 0 &(0)                 & 0 &(0)                 & 0 &(0)\\
        \mn{u1c20} &                     & 937	            & 584	 &   35        & 1 &(0)	            & 0 &(10)            & 1 &(0)	            & 17 &(0)                & 16 &(6)                & 17 &(0)\\
        \mn{u1c50} &                     & 625	            & 875	&    141        & 26 &(0)	            & 9 &(82)            & 20 &(56)           & 79 &(0)                & 28 &(81)               & 38 &(40)\\
        \mn{u5c1} &                      & 7637	        & 316	&       10     & 0 &(0)	            & 0 &(0)	            & 0 &(0)             & 2 &(0)                 & 2 &(0)                 & 2 &(0)\\
        \mn{u5c20} &                     & 4947	        & 3035	 &        224   & 37 &(0)              & 6 &(96)	        & 32 &(32)           & 91 &(0)                & 51 &(151)              & 83 &(18)\\
        \hline
        \mn{u1c1} & \multirow{3}{*}{$\succ^{nb}$}             & 1447	        & 76	  &      2    & 0 &(0)               & 0 &(0)             & 0 &(0)             & 0 &(0)                 & 0 &(0)                 & 0 &(0)\\
        \mn{u1c20} &             & 940	            & 571	    & 45       & 0 &(0)               & 0 &(2)             & 0 &(0)             & 1 &(0)                 & 0 &(2)                 & 1 &(0)\\
        \mn{u1c50} &            & 680	            & 823	 &      138     & 0 &(0)               & 0 &(43)            & 0 &(0)             & 20 &(0)                & 13 &(37)               & 20 &(0)\\
        \hline
    \end{tabular}
    }
    \caption{Number of answers found trivially P-IAR (Triv), grounded (GR) but not trivially P-IAR, potential answers (Pot) not grounded, and X-AR or X-brave but not grounded. The number of potential answers for which we ran out of time (600s) is given in parenthesis. 
    }\label{tab:number-answers}
\end{table*}
\begin{figure*}
\centering
\includegraphics[width=0.4\textwidth,valign=t]{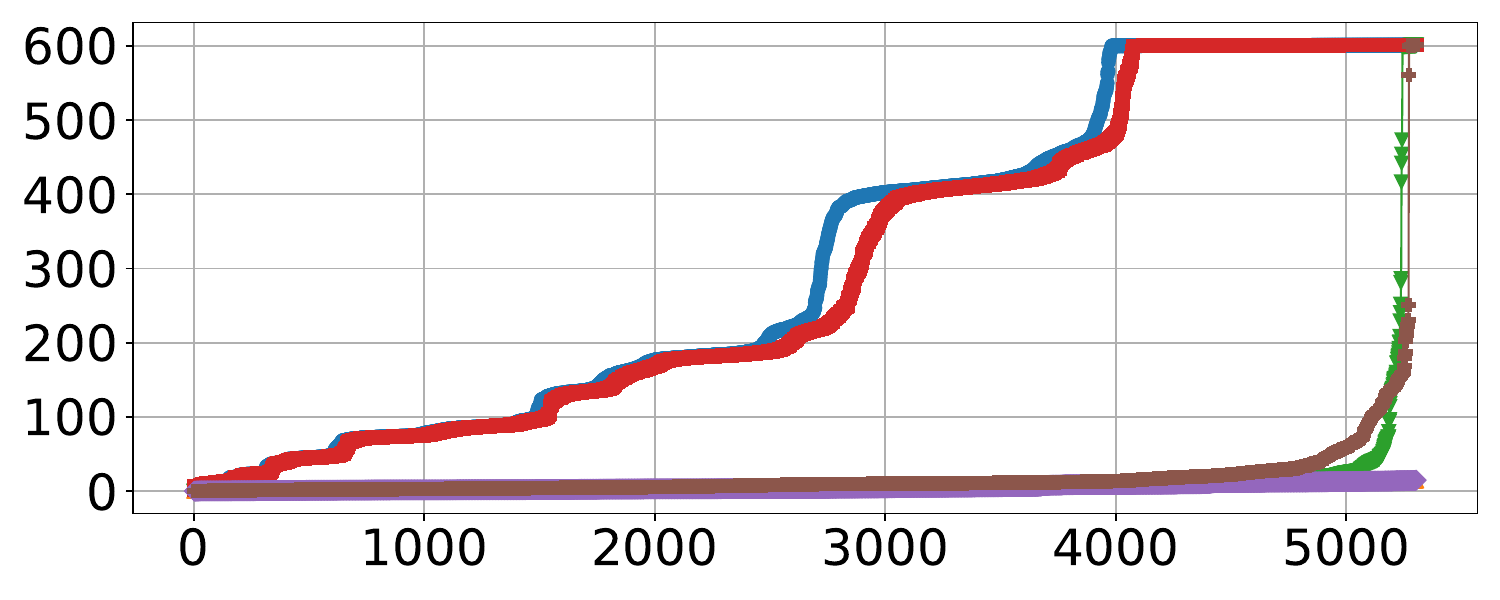}
\quad
\includegraphics[width=0.1\textwidth,valign=t]{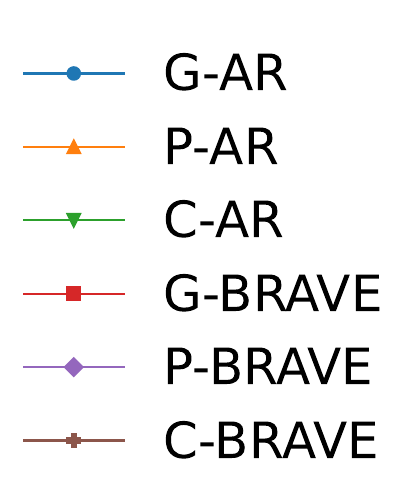}
\includegraphics[width=0.4\textwidth,valign=t]{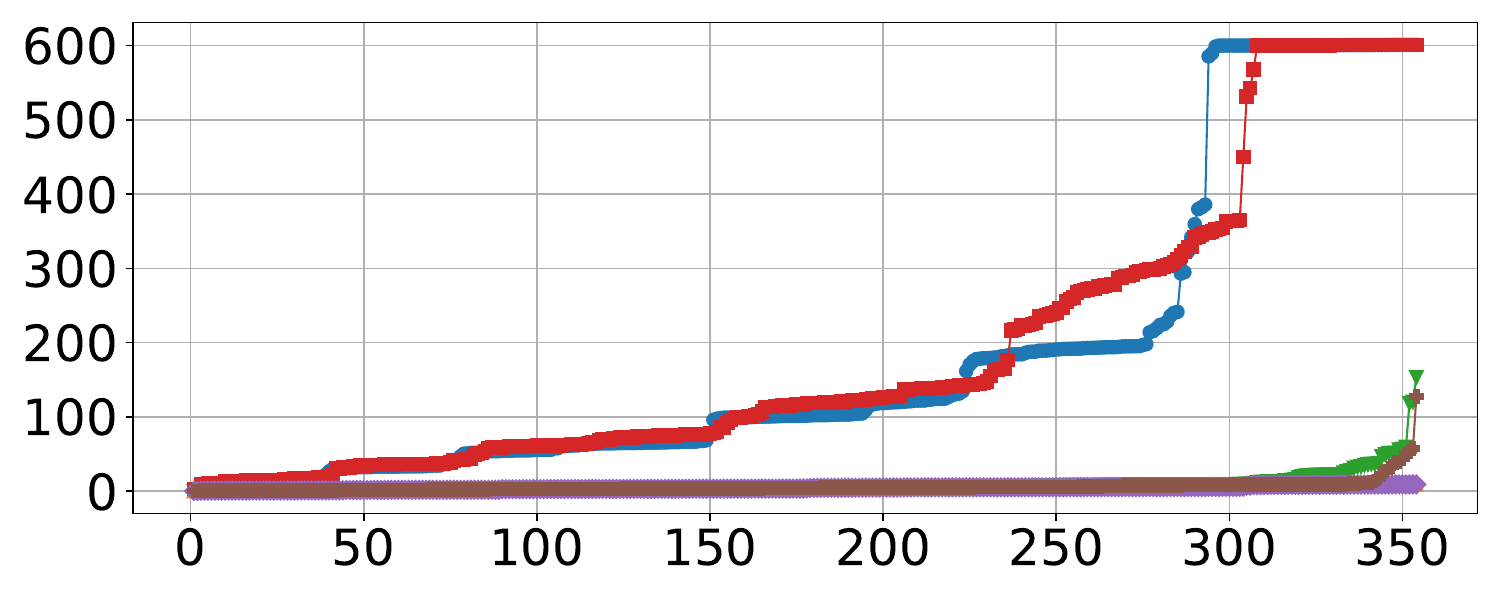}
\caption{Time (in seconds, y-axis) to decide for each $q(\ans)$, dataset \mn{uXcY}, and priority relation, whether $q(\ans)$ holds under X-AR/X-brave, with instances sorted by increasing solving time 
along the x-axis (\ie a point $(i,j)$ indicates that $i$ instances could be solved within $j$ seconds).\\(left) binary conflicts case (\mn{u1cY} and \mn{u5cY}, two priority relations) (right) non-binary conflicts case (\mn{u1cY}, one priority relation)}\label{fig:overall-runtime}
\end{figure*}

\subsection{Experimental Results}
In what follows, we summarize our main observations.
\begin{figure}
\includegraphics[width=0.155\textwidth]{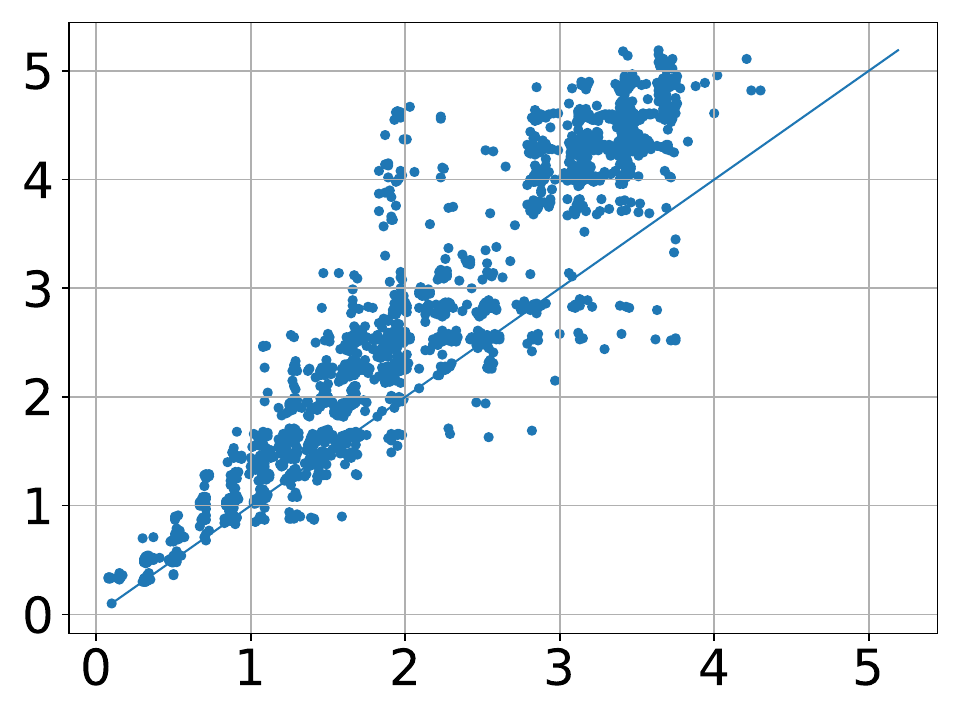}
\includegraphics[width=0.155\textwidth]{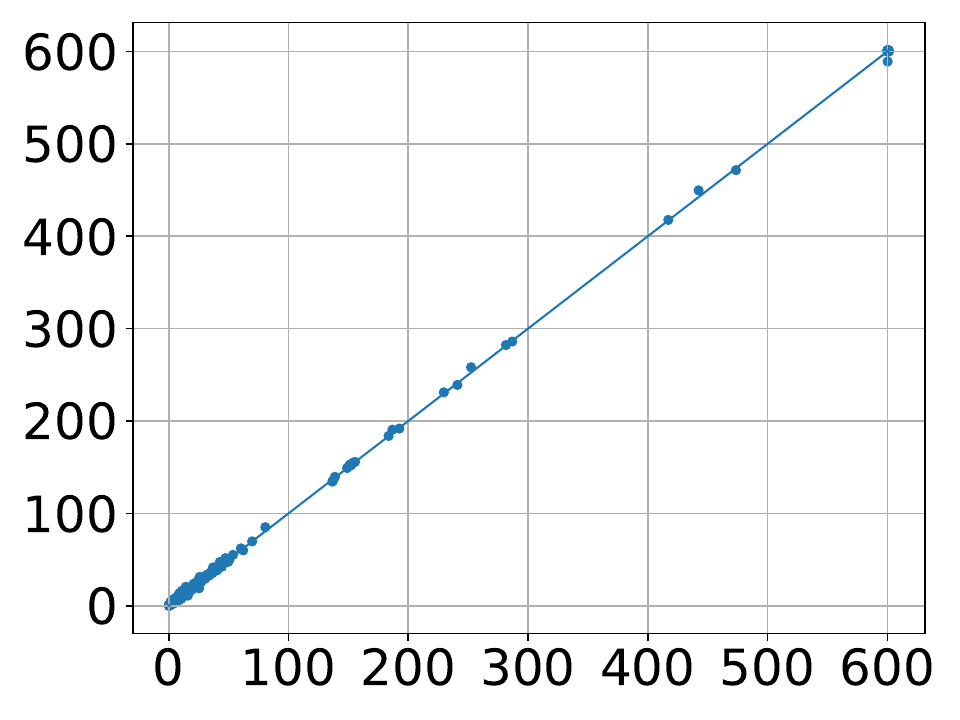}
\includegraphics[width=0.155\textwidth]{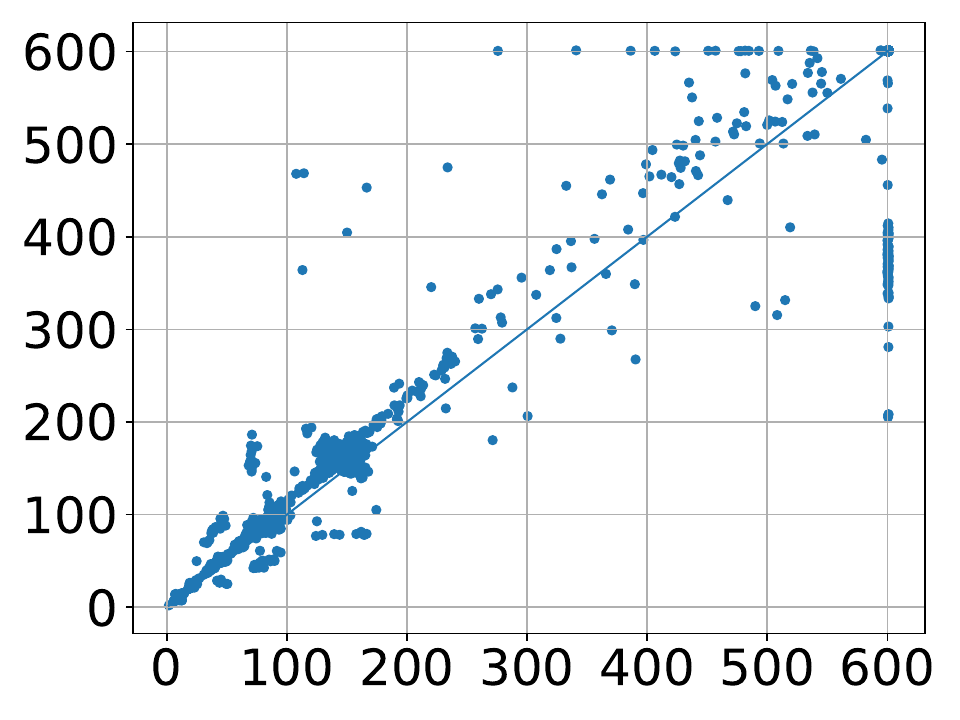}
\caption{Time for the generic encoding (y-axis) vs 
time for the encoding specific for binary conflicts (x-axis), both with localization, to decide for each $q(\ans)$, dataset \mn{uXcY}, and priority relation,  
if $q(\ans)$ holds under X-AR. (left) X=P (middle) X=C (right) X=G}\label{fig:impact-binary}
\end{figure}

\subsubsection{Tractable Approximations} Table~\ref{tab:grounded-repair-incremental} shows the number of plain (S-)IAR facts (not involved in any conflict, $\cap_{S}$), other trivially P-IAR facts ($\Gamma(\emptyset)$), and facts added at each step of the incremental computation of the grounded repair, for some example datasets and priority relations. 
Over all the cases we consider, 
the proportion of facts that belong to 
the grounded repair but which are not trivially P-IAR (\ie are in $\Gmc\setminus\Gamma(\emptyset)$) varies from 1\%-2\% on the datasets with a low proportion of facts in conflicts (\mn{uXc1}) to 15\%-31\% (depending on the priority relation) on the dataset with the highest proportion of facts in conflicts \mn{u1c50}. 
Table~\ref{tab:time-grounded-compact} gives examples of times needed to compute the grounded repair and the trivially P-IAR facts. Over all the cases we consider, 
computing $\Gmc$ instead of $\Gamma(\emptyset)$ from the attack relation multiplies the runtime by up to 34
but it remains below 600s. However, we encounter memory problems for larger datasets (note that for this task we also used datasets \mn{u20cY} of 2M facts). 
Table~\ref{tab:grounded-vs-iar-facts} shows the size of the grounded repair and how many facts are added by taking the intersections of the optimal repairs.  
Given the cost of checking whether each non-grounded fact holds under X-(I)AR, these numbers suggest the interest of using 
grounded 
rather than X-IAR for $X\in\{P,G,C\}$. This is also the reason that we did not conduct further experimental evaluations of the X-IAR semantics. 

\subsubsection{Impact of Localization} 
Unsurprisingly, localization is crucial. For example, on \mn{u1c10} with $\succ^{nb}$, we manage to decide for 99\% of the pairs of a potential answer (which does not hold under grounded) and a semantics (X-AR or X-brave  with $X\in\{P,G,C\}$) 
whether the answer holds under the semantics, while they all yield a time-out without localization. 
We adopt localization (using weak reachability for globally-optimal repairs and strong one for Pareto- and completion-optimal repairs) in the rest of the experiments.

\subsubsection{Impact of Specific Encoding for Binary Conflicts} 
Figure~\ref{fig:impact-binary} compares for each potential answer (which does not hold under grounded), dataset with binary conflicts, and priority relation, the time needed to decide if it holds under X-AR semantics with and without using the simplified encodings proposed in the case of binary conflicts. 
We can see that this simplification is indeed generally helpful when X=P, and mostly neutral when X=C.  When X=G, we observe a more diverse outcome: even if the binary conflicts-specific encoding improves the performance in a majority of cases, there are a significant number of cases where it is the opposite. Moreover, the difference between the two encodings is more important (in particular, the two encodings do not lead to time-out on the same instances).

\subsubsection{Semantics Comparison in Terms of Runtimes} 
Figure~\ref{fig:overall-runtime} shows the times needed to decide whether a potential answer which does not hold under grounded semantics holds under a given semantics for X-AR and X-brave with $X\in\{P,G,C\}$, over all the cases we consider, using localization and the specific encoding for binary conflicts when possible. 
Recall that in a cactus plot instances are sorted by solving time, so a point $(i,j)$ in the plot indicates that for the considered semantics,
there were $i$ instances which could be (individually) solved within $j$ seconds. 
These results confirm that using globally-optimal repairs is significantly more difficult in practice than using Pareto- or even completion-optimal repairs, and the difference becomes more pronounced as the proportion of facts involved in some conflicts increases.

\subsubsection{Semantics Comparison in Terms of Answers}
Table~\ref{tab:number-answers} shows the number of answers we obtain under the different semantics and the number of potential (non grounded) answers for which we did not manage to determine whether they hold under the considered semantics in our 600s time limit.  
A first observation is that the grounded semantics seems to be a very good approximation of the X-AR semantics ($X\in\{P,G,C\}$). This is key to achieving feasibility as it allows us to reduce the number of potential answers for which we need to call the procedures for the P/C-AR/brave semantics (cf.\ the column Pot$\setminus$GR which gives the number of potential answers for which we actually need to use the encodings for X-AR or X-brave semantics).  For X-brave answers however, we do observe cases where a large share of the answers cannot be obtained without using the designated encodings (in particular in the $\succ^{ss}$ case, which is omitted from the table as in this case the three kinds of optimal repair coincide). 
Another useful information is that the three kinds of repairs seem to often yield the same answers (recall that answers that hold under P-AR semantics are included in those that hold under G-AR semantics, which are included in those that hold under C-AR semantics, and the inclusions for the X-brave semantics are the other way around).

\section{Conclusion}
In this paper, we have presented the first implementation of globally-optimal repair-based 
semantics, using ASP(Q), and of the grounded semantics, making it possible for 
the first time to compare these semantics with the previously implemented semantics
based upon Pareto- and completion-optimal repairs. 
While we are 
able to compute query answers for some instances with globally-optimal repairs, the runtimes are significantly longer, with more timeouts, than for the Pareto-
and completion-optimal repairs, in line with the worst-case complexity. This suggests that even if we aim to compute G-AR and G-brave answers, a better strategy is to first test whether the candidate answer holds under P-AR semantics and P-brave semantics (which are the fastest optimal repair semantics)  and only call the G-AR and G-brave procedures if the status remains unresolved.

Our evaluation also highlighted 
the surprising effectiveness of  
the grounded 
semantics as an approximation of the optimal repair-based semantics, an insight which moreover can be applied also to non-ASP-based implementations. Importantly, the grounded repair can be computed as a preprocessing task, 
thus making it possible to employ a principled inconsistency-tolerant semantics with little overhead compared to standard query answering. 
This suggests the interest of developing highly optimized methods for computing and updating the grounded repair for very large datasets.

\section*{Acknowledgements}
This work was supported by the ANR AI Chair INTENDED (ANR-19-CHIA-0014); by the Italian Ministry of Industrial Development (MISE) under project EI-TWIN n. F/310168/05/X56 CUP B29J24000680005; and by the Italian Ministry of Research (MUR) under project PNRR FAIR - Spoke 9 - WP 9.1 CUP H23C22000860006.

\section*{AI Declaration}
The authors have not employed any Generative AI tools.

\bibliographystyle{kr}
\bibliography{kr-sample}

\appendix

\section{Proofs for Section~\ref{sec:encodings}}

\PropBrave*
\begin{proof}[Proof Sketch]
\noindent$(\Leftarrow)$ Assume that $\Pi^G_{\mi{brave}}$ has a quantified answer set $M$. Since $M\in AS(P_1)$, the set $\Rmc$ that contains all facts whose identifiers are in $\{\mt{A}\mid \mt{inRepair(A)}\in M\}$ (extended with facts that do not appear in any conflict nor cause) is a repair of $\Kmc$ that contains a cause for $q(\ans)$. Moreover, since $P_2\cup\fix{P_1}{M}$ is incoherent, $\Rmc$ does not have any global improvement. Hence $\mathcal{K}_{\succ} \bravemodels{G} q(\vec{a})$. 

\noindent$(\Rightarrow)$ Conversely, assume that $\Pi^G_{\mi{brave}}$ is incoherent. Let $\Rmc\in\reps{\Kmc}$ be such that $\Rmc$ contains some cause for $q(\ans)$. The interpretation $M$ that contains (i) \mt{inRepair(A)} for every identifier \mt{A} of $\alpha\in\Rmc$ and (ii) the atoms (deterministically) derived by $P_1\setminus\{\mt{\{inRepair(A)\}\impl reachable(A)}\}$ is such that $M\in AS(P_1)$. Since $M$ is such that $P_2\cup\fix{P_1}{M}$ is coherent (otherwise $\Pi^G_{\mi{brave}}$ will be coherent), $\Rmc$ admits a global improvement. Thus $\mathcal{K}_{\succ} \not\bravemodels{G} q(\vec{a})$.
\end{proof}

\subsubsection{Proof of Theorem~\ref{ThmLocalizationGlobal}}
We prove Proposition~\ref{prop:ThmLocalizationGlobalFirst} and Proposition~\ref{prop:ThmLocalizationGlobalSecond} below, which correspond two the first and second item of Theorem~\ref{ThmLocalizationGlobal}, respectively. 
Recall that we consider a prioritized KB $\Kmc_\succ$ with $\Kmc=(\Dmc,\Tmc)$, a subset $\Bmc\subseteq\Dmc$, and the set of facts reachable from $\Bmc$ (\wrt `weak' reachability) $\Bmc_r=\reachweak(\Bmc)$, which is used to define the prioritized KB $\Kmc^{\Bmc_r}_{\succ^{\Bmc_r}}$ by $\Kmc^{\Bmc_r}=(\Bmc_r,\Tmc)$ and $\succ^{\Bmc_r}$ the restriction of $\succ$ to $\Bmc^r$. 

In the proofs of the propositions, we will use the \emph{set of minimal (\wrt $\succ$) facts of a conflict}: 
for every $\Cmc \in \conflicts{\Kmc}$ let $\minconf{\Cmc} = \{ \gamma | \forall \delta \in \Cmc, \gamma \not\succ \delta \}$. 

We recall that if $\succ'$ is a total priority relation, there is a unique Pareto-, globally- and completion-optimal repair \cite{DBLP:journals/amai/StaworkoCM12}, so one can define completion-optimal repairs equivalently as those that are Pareto-optimal or globally-optimal \wrt some completion $\succ'$ of $\succ$.

\begin{proposition}\label{prop:ThmLocalizationGlobalFirst}
     For $X\in\{P,G,C\}$, if $\Rmc\in\xreps{\Kmc_\succ}$, then $\Rmc\cap\Bmc_r\in\xreps{\Kmc^{\Bmc_r}_{\succ^{\Bmc_r}}}$.
\end{proposition}
\begin{proof}
\noindent\textbf{Case $\mathbf{X=P}$.} 
Let $\Rmc\in\preps{\Kmc_\succ}$ and suppose for a contradiction that $\Rmc \cap \Bmc_r\notin\preps{\Kmc^{\Bmc_r}_{\succ^{\Bmc_r}}}$: there exists $\Rmc_0 \subseteq \Bmc_r$ such that there exists $\beta \in \Rmc_0 \setminus (\Rmc \cap \Bmc_r)$, such that $\beta \succ \alpha$ for all $\alpha \in  (\Rmc \cap \Bmc_r) \setminus \Rmc_0$.

Let $\widetilde{\Rmc_0} = \{\beta\} \cup (\Rmc_0 \cap \Rmc) \cup (\Rmc \setminus \Bmc_r)$. We build a Pareto improvement of $\Rmc$ \wrt $\Kmc_\succ$ by removing selected elements from $\widetilde{\Rmc_0}$ to obtain a contradiction with $\Rmc\in\preps{\Kmc_\succ}$.

\begin{claim}\label{lem:gamma_pareto}
    For every $\Cmc \in \conflicts{\Kmc}$ such that $\Cmc \subseteq \widetilde{\Rmc_0}$, there exists $\gamma_{\Cmc}\in\Cmc$ such that $\beta \succ \gamma_\Cmc$. 
\end{claim}
\noindent{\emph{Proof of claim.}}
    Let $\Cmc \in \conflicts{\Kmc}$ be such that $\Cmc \subseteq \widetilde{\Rmc_0}$. 
    First, note that $\beta \in \Cmc$ as otherwise $\Cmc \subseteq \Rmc$, which contradicts the $\Tmc$-consistency of $\Rmc$. 
    Moreover, $\Cmc \cap (\Rmc \setminus \Bmc_r) \neq \emptyset$ as otherwise $\Cmc \subseteq \Rmc_0$ which contradicts the $\Tmc$-consistency of $\Rmc_0$. 
    By definition of $\Bmc_r=\reachweak(\Bmc)$, and since $\beta\in\Bmc_r$, if for every $\gamma\in\Cmc$, $\beta\not\succ\gamma$, it would hold that $\Cmc\subseteq \Bmc_r$. Since we showed that $\Cmc \cap (\Rmc \setminus \Bmc_r) \neq \emptyset$, it follows that 
    there exists $\gamma_\Cmc \in \Cmc$ such that $\beta \succ \gamma_\Cmc$. \emph{(end of claim proof)}
    \smallskip

Let $\Gamma_{\widetilde{\Rmc_0}} = \{\gamma_\Cmc \mid \Cmc \in \conflicts{\Kmc}, \Cmc \subseteq \widetilde{\Rmc_0}, \beta\succ\gamma_\Cmc\}$ and \mbox{$\Rmc'_0 = \widetilde{\Rmc_0} \setminus \Gamma_{\widetilde{\Rmc_0}}$.} 
As $\Rmc'_0$ is obtained by removing at least one fact per conflict included in $\widetilde{\Rmc_0}$, it is $\Tmc$-consistent. Also:
    \begin{align*}
        \Rmc'_0 \setminus \Rmc &= (\widetilde{\Rmc_0} \setminus \Gamma_{\widetilde{\Rmc_0}}) \setminus \Rmc\\
        &= ((\{\beta\} \cup (\Rmc_0 \cap \Rmc) \cup (\Rmc \setminus \Bmc_r)) \setminus \Gamma_{\widetilde{\Rmc_0}}) \setminus \Rmc\\
        &= \{\beta\}\text{   since $\beta \not \in \Gamma_{\widetilde{\Rmc_0}}$ and $\beta \not \in \Rmc$;}
        \\
        \Rmc \setminus \Rmc'_0 &= \Rmc \setminus (\widetilde{\Rmc_0} \setminus \Gamma_{\widetilde{\Rmc_0}}) \\
        &= ((\Rmc \cap \Bmc_r) \setminus \Rmc_0) \cup \Gamma_{\widetilde{\Rmc_0}}\text{   since $\Gamma_{\widetilde{\Rmc_0}} \subseteq \Rmc$.}
    \end{align*}

By construction of $\Rmc'_0$, for every $ \alpha \in \Rmc \setminus \Rmc'_0$, $\beta \succ \alpha$. Hence $\Rmc'_0$ is a Pareto improvement of $\Rmc$, which contradicts $\Rmc\in\preps{\Kmc_\succ}$. 
Therefore $\Rmc \cap \Bmc_r \in \preps{\Kmc^{\Bmc_r}_{\succ^{\Bmc_r}}}$.

\smallskip
\noindent\textbf{Case $\mathbf{X=G}$.} 
Let $\Rmc\in\greps{\Kmc_\succ}$ and suppose for a contradiction that $\Rmc \cap \Bmc_r\notin\greps{\Kmc^{\Bmc_r}_{\succ^{\Bmc_r}}}$: there exists $\Rmc_0 \subseteq \Bmc_r$ such that $\Rmc_0\neq \Rmc \cap \Bmc_r$ and for every $\alpha \in  (\Rmc \cap \Bmc_r) \setminus \Rmc_0$, there exists $\beta \in \Rmc_0 \setminus (\Rmc \cap \Bmc_r)$ such that $\beta \succ \alpha$.

Let $\widetilde{\Rmc_0} = \Rmc_0  \cup (\Rmc \setminus \Bmc_r)$. We build a global improvement of $\Rmc$ \wrt $\Kmc_\succ$ by removing selected elements from $\widetilde{\Rmc_0}$ to obtain a contradiction with $\Rmc\in\greps{\Kmc_\succ}$.

We first show that for every $\Cmc \in \conflicts{\Kmc}$ such that $\Cmc \subseteq \widetilde{\Rmc_0}$, there exist $\beta_\Cmc \in \Cmc\cap(\Rmc_0 \setminus (\Rmc \cap \Bmc_r))$ and $\gamma_\Cmc \in\Cmc\setminus (\Rmc_0 \setminus (\Rmc \cap \Bmc_r))$ such that $\beta_\Cmc \succ \gamma_\Cmc$. 
Let $\Cmc \in \conflicts{\Kmc}$ be such that $\Cmc \subseteq \widetilde{\Rmc_0}$. First, note that $\Cmc \cap (\Rmc_0 \setminus (\Rmc \cap \Bmc_r)) \neq \emptyset$ as otherwise $\Cmc \subseteq \Rmc$, which contradicts the $\Tmc$-consistency of $\Rmc$. 
Moreover, $\Cmc \cap (\Rmc \setminus \Bmc_r) \neq \emptyset$ as otherwise $\Cmc \subseteq \Rmc_0$, which contradicts the $\Tmc$-consistency of $\Rmc_0$. 

\begin{claim}\label{lem:min_global}
   For every $\Cmc \in \conflicts{\Kmc}$ such that $\Cmc \subseteq \widetilde{\Rmc_0}$, $\minconf{\Cmc} \subseteq \Rmc \setminus \Bmc_r$
\end{claim}
\noindent{\emph{Proof of claim.}}
By definition of $\Bmc_r=\reachweak(\Bmc)$, if there was $\alpha\in\minconf{\Cmc}$ such that $\alpha\in\Bmc_r$, it would hold that $\Cmc\subseteq\Bmc_r$. Since $\Cmc \cap (\Rmc \setminus \Bmc_r) \neq \emptyset$, it follows that $\minconf{\Cmc} \subseteq \Rmc \setminus \Bmc_r$. 
 \emph{(end of claim proof)}
    \smallskip

We define ${(\beta_i)}_{i \in \mathbb{N}} \in \Cmc^{\mathbb{N}}$ with $\beta_0 \in \Cmc \cap (\Rmc_0 \setminus (\Rmc \cap \Bmc_r))$ (such $\beta_0$ exists since $\Cmc \cap (\Rmc_0 \setminus (\Rmc \cap \Bmc_r)) \neq \emptyset$) and for $i \in \mathbb{N}$:
    \begin{itemize}
        \item if $\beta_i \in \Rmc_0$ then let $\beta_{i+1} \in \Cmc$ be such that $\beta_i \succ \beta_{i+1}$ (such $\beta_{i+1}$ exists because $\minconf{\Cmc} \subseteq \Rmc \setminus \Bmc_r$ and $\Rmc_0\subseteq\Bmc_r$ so $\beta_i\notin\minconf{\Cmc}$);
        \item otherwise ($\beta_i \not \in \Rmc_0$), let $\beta_{i+1} = \beta_i$.
    \end{itemize}
As $\Dmc$ is finite and $\succ$ is acyclic the sequence is ultimately constant. Let $j = \min_{i\in \mathbb{N}}{(\beta_{i+1} \not \in \Rmc_0 \setminus (\Rmc \cap \Bmc_r))}$. 
 Let $\beta_\Cmc = \beta_j $ and $\gamma_\Cmc = \beta_{j+1}$. By definition of $j$ we have $\beta_\Cmc \in \Rmc_0 \setminus (\Rmc \cap \Bmc_r)$ and $\gamma_\Cmc \not \in \Rmc_0 \setminus (\Rmc \cap \Bmc_r)$, and by construction of $(\beta_i)_{i\in\mathbb{N}}$, it holds that $\beta_\Cmc \succ \gamma_\Cmc$.

 Let $\Gamma_{\widetilde{\Rmc_0}} = \{\gamma_\Cmc \mid \Cmc \in \conflicts{\Kmc}, \Cmc \subseteq \widetilde{\Rmc_0}\}$ and \mbox{$\Rmc'_0 = \widetilde{\Rmc_0} \setminus \Gamma_{\widetilde{\Rmc_0}}$}. As $\Rmc'_0$ is obtained by removing at least one fact per conflict included in $\widetilde{\Rmc_0}$,
it is $\Tmc$-consistent. Also:
\begin{align*}
         \Rmc'_0 \setminus \Rmc &= (\widetilde{\Rmc_0} \setminus \Gamma_{\widetilde{\Rmc_0}}) \setminus \Rmc \\
        &= (\Rmc_0 \setminus (\Rmc \cap \Bmc_r)) \setminus \Gamma_{\widetilde{\Rmc_0}} \\
        &= \Rmc_0 \setminus (\Rmc \cap \Bmc_r)
        \end{align*}
        as for all $\gamma_\Cmc \in \Gamma_{\widetilde{\Rmc_0}}$ we have $\gamma_\Cmc \not \in \Rmc_0 \setminus (\Rmc \cap \Bmc_r)$, and
\begin{align*}
        \Rmc \setminus \Rmc'_0 &= \Rmc \setminus (\widetilde{\Rmc_0} \setminus \Gamma_{\widetilde{\Rmc_0}})\\
        &= \Rmc \setminus ((\Rmc_0  \cup (\Rmc \setminus \Bmc_r)) \setminus \Gamma_{\widetilde{\Rmc_0}}) \\
         &= ((\Rmc \cap \Bmc_r) \setminus \Rmc_0) \cup \Gamma_{\widetilde{\Rmc_0}}.
        \end{align*}

Let $\alpha \in \Rmc \setminus \Rmc'_0$:
\begin{itemize}
\item  If $\alpha\in (\Rmc \cap \Bmc_r) \setminus \Rmc_0$, by assumption on $\Rmc_0$, there exists $\beta \in \Rmc_0 \setminus (\Rmc \cap \Bmc_r)=\Rmc'_0 \setminus \Rmc$ such that $\beta \succ \alpha$. 

\item Otherwise, $\alpha\in \Gamma_{\widetilde{\Rmc_0}}$ is equal to some $\gamma_\Cmc$ and $\beta=\beta_\Cmc$ is such that $\beta \in \Rmc_0 \setminus (\Rmc \cap \Bmc_r)= \Rmc'_0 \setminus \Rmc$ (by definition of the $\beta_\Cmc$) and $\beta \succ \alpha$ (since $\beta_\Cmc\succ\gamma_\Cmc$). 
\end{itemize}
Hence $\Rmc'_0$ is a global improvement of $\Rmc$, which contradicts $\Rmc\in\greps{\Kmc_\succ}$. 
Therefore $\Rmc \cap \Bmc_r \in \greps{\Kmc^{\Bmc_r}_{\succ^{\Bmc_r}}}$.

\smallskip
\noindent\textbf{Case $\mathbf{X=C}$.} 
Let $\Rmc\in\creps{\Kmc_\succ}$, \ie $\Rmc\in\preps{\Kmc_{\succ'}}$ for some $\succ'$ completion of $\succ$. Thus, by the case $X=P$ of the proposition that we have already shown, we have $\Rmc \cap \Bmc_r\in\preps{\Kmc_{{\succ'}^{\Bmc_r}}^{\Bmc_r}}$. As ${\succ'}^{\Bmc_r}$ is a completion of ${\succ}^{\Bmc_r}$ in $\Kmc^{\Bmc_r}$ we directly obtain that $\Rmc \cap \Bmc_r\in\creps{\Kmc_{{\succ}^{\Bmc_r}}^{\Bmc_r}}$.
\end{proof}

\begin{proposition}\label{prop:ThmLocalizationGlobalSecond}
    For $X\in\{P,G,C\}$, if $\Rmc\in \xreps{\Kmc^{\Bmc_r}_{\succ^{\Bmc_r}}}$, then there exists $\Rmc'\in\xreps{\Kmc_\succ}$ such that $\Rmc=\Rmc'\cap\Bmc_r$. 
\end{proposition}
\begin{proof}

\smallskip
\noindent\textbf{Case $\mathbf{X=P}$.} 
Let $\Rmc\in\preps{\Kmc_{{\succ}^{\Bmc_r}}^{\Bmc_r}}$. 
Following the proof of Lemma 9 from the extended version of \cite{DBLP:conf/kr/BienvenuB22}, we build some $\Rmc^*\subseteq\Dmc$ as required by the proposition as follows: 
let $\Rmc' = \Rmc$, $\Dmc' = \Dmc$, and repeat the following steps until $\Dmc'$ is empty:
    \begin{itemize}
        \item  Choose $\alpha \in \Dmc'$ such that $\beta \not \succ \alpha$ for all $\beta \in  \Dmc'$.
        \item If $\Rmc' \cup \{\alpha\}$ is $\Tmc$-consistent, add $\alpha$ to $\Rmc'$.
        \item Remove $\alpha$ from $\Dmc'$.
    \end{itemize}   
Let $\Rmc^*$ be the final $\Rmc'$. 
First, note that $\Rmc\subseteq\Rmc^*\cap\Bmc_r$ by initialization of $\Rmc^*$ and definition of $\Rmc$. Also $\Rmc^*\cap\Bmc_r\subseteq\Rmc$ as otherwise there would exist $\alpha \in (\Rmc^*\cap\Bmc_r) \setminus \Rmc$ but then $\Rmc \cup \{\alpha\}$ would be $\Tmc$-consistent and hence a Pareto improvement of $\Rmc$. Thus $\Rmc=\Rmc^*\cap\Bmc_r$.

We show that $\Rmc^* \in \preps{\Kmc_{\succ}}$. Suppose for a contradiction there exists $\Rmc_0 \subseteq \Dmc$ such that there exists $\beta \in \Rmc_0 \setminus \Rmc^*$ such that $\beta \succ \alpha$ for every $\alpha \in \Rmc^* \setminus \Rmc_0$. Since $\beta \not \in \Rmc^*$, by construction of $\Rmc^*$, there exists $\Cmc \in \conflicts{\Kmc}$ such that $\beta \in \Cmc$ and $\Cmc \setminus \{\beta\} \subseteq \Rmc '$ when $\beta$ was chosen in the construction of $\Rmc^*$ (otherwise $\beta$ would have been added to $\Rmc '$). 

\begin{claim}\label{lem:minconf_pareto}
    For every $\Cmc \in \conflicts{\Kmc}$ such that $\beta \in \Cmc$, $\minconf{\Cmc} \cap \Rmc = \emptyset$.
\end{claim}
\noindent{\emph{Proof of claim.}}
    Suppose for a contradiction $\minconf{\Cmc} \cap \Rmc \neq \emptyset$, \ie there exists $\gamma\in\Cmc\cap\Rmc\subseteq \Cmc\cap\Bmc_r$ such that $\gamma\not\succ\delta$ for every $\delta\in\Cmc$. By definition of $\Bmc_r=\reachweak(\Bmc)$, $\Cmc \subseteq \Bmc_r$, and in particular $\beta \in \Bmc_r$. Hence $\Rmc_0 \cap \Bmc_r$ is a Pareto improvement of $\Rmc$, which contradicts the assumption $\Rmc\in \preps{\Kmc^{\Bmc_r}_{\succ^{\Bmc_r}}}$. 
 \emph{(end of claim proof)}
    \smallskip

As $\Rmc_0$ is $\Tmc$-consistent, $\Cmc \setminus \{\beta\} \not \subseteq \Rmc_0$ so there is $\alpha_0 \in \Cmc \cap (\Rmc^* \setminus \Rmc_0)$. Hence $\beta \succ \alpha_0$ (by definition of $\Rmc_0$ and $\beta$). 

Since $\alpha_0$ has been chosen after $\beta$ in the construction of $\Rmc^*$ and was already in $\Rmc'$ when $\beta$ was considered, we necessarily have that $\alpha_0 \in \Rmc$. Hence, since $\minconf{\Cmc} \cap \Rmc = \emptyset$, there exists $\alpha_1 \in \Cmc$ with $\alpha_0 \succ \alpha_1$. 

Because $\beta \succ \alpha_0 \succ \alpha_1$, $\alpha_1$ was chosen after $\beta$ in the construction of $\Rmc^*$ thus necessarily $\alpha_1 \in \Rmc$. Thus we can iterate the argument to create an infinite chain $\beta \succ \alpha_0 \succ \alpha_1 \succ \ldots $ with all $\alpha_i \in \Rmc$. As $\succ$ is acyclic and $\Dmc$ is finite we obtain a contradiction. 
Therefore $\Rmc^* \in \preps{\Kmc_{\succ}}$.

\smallskip
\noindent\textbf{Case $\mathbf{X=G}$.} 
Let $\Rmc\in\greps{\Kmc_{{\succ}^{\Bmc_r}}^{\Bmc_r}}$.
As we did in the case $X=P$, let $\Rmc' = \Rmc$, $\Dmc' = \Dmc$, and repeat the following steps until $\Dmc'$ is empty:
    \begin{itemize}
        \item  Choose $\alpha \in \Dmc'$ such that $\beta \not \succ \alpha$ for all $ \beta \in  \Dmc'$.
        \item If $\Rmc' \cup \{\alpha\}$ is $\Tmc$-consistent, add $\alpha$ to $\Rmc'$.
        \item Remove $\alpha$ from $\Dmc'$.
    \end{itemize}
Let $\Rmc^*$ be the final $\Rmc'$. 
As in the case $X=P$, we obtain  
$\Rmc=\Rmc^*\cap\Bmc_r$.

We show that $\Rmc^* \in \greps{\Kmc_{\succ}}$. Suppose for a contradiction there exists $\Rmc_0 \subseteq \Dmc$ such that $\Rmc_0\neq \Rmc^*$ and for every $\alpha \in \Rmc^* \setminus \Rmc_0$, there exists $\beta \in \Rmc_0 \setminus \Rmc^*$ such that $ \beta \succ \alpha$.

Assume for a contradiction that there exists $\alpha \in \Rmc \setminus \Rmc_0$. Since $\Rmc \setminus \Rmc_0\subseteq\Rmc^*\setminus\Rmc_0$, there exists $\beta \in \Rmc_0 \setminus \Rmc^*$ such that $ \beta \succ \alpha$. 
By definition of $\succ$, $ \beta \succ \alpha$ implies that there exists $\Cmc \in \conflicts{\Kmc_\succ}$ such that $\{\alpha,\beta\} \subseteq \Cmc$. 
In particular, as $\alpha \not \succ \beta$ and $\alpha \in \Bmc_r$ we must have $\Cmc \subseteq \Bmc_r$ hence $\beta \in \Bmc_r$.
Therefore for every $\alpha \in \Rmc \setminus (\Rmc_0 \cap \Bmc_r)$ there exists $\beta \in (\Rmc_0 \cap \Bmc_r) \setminus \Rmc$ such that $\beta \succ \alpha$, which contradicts $\Rmc \in \greps{\Kmc_{{\succ}^{\Bmc_r}}^{\Bmc_r}}$.
Thus $\Rmc \subseteq \Rmc_0$. Hence $\Rmc = \Rmc_0 \cap \Bmc_r$ (otherwise we have a contradiction with $\Rmc$ being inclusion-maximal).

Let $\alpha_0$ be the first element of $(\Rmc^* \cup \Rmc_0) \setminus \Bmc_r$ considered in the construction of $\Rmc^*$.
    \begin{itemize}
        \item If $\alpha_0 \in \Rmc^* \setminus \Rmc_0$ then there exists $\beta \in \Rmc_0 \setminus \Rmc^*$ such that $\beta \succ \alpha_0$ but then $\beta \not \in \Bmc_r$ because as we previously showed $\Rmc^* \cap \Bmc_r = \Rmc = \Rmc_0 \cap \Bmc_r$ thus $\Rmc_0 \setminus \Rmc^* \subseteq \Dmc \setminus \Bmc_r$. Hence $\beta$ has been considered before $\alpha_0$ which contradicts our assumption.
        \item If $\alpha_0 \in \Rmc_0 \setminus \Rmc^*$ then there exists $\Cmc \in \conflicts{\Kmc}$ such that $\alpha_0 \in \Cmc$ and $\Cmc \setminus \{\alpha_0\} \subseteq \Rmc = \Rmc^* \cap \Bmc_r$ (otherwise $\alpha_0$ would have been added) but then $\Cmc \setminus \{\alpha_0\} \subseteq \Rmc_0 \cap \Bmc_r$ as $\Rmc = \Rmc_0 \cap \Bmc_r$. Hence $\Cmc \subseteq \Rmc_0$ and as $\Cmc \in \conflicts{\Kmc}$ this contradicts the fact that $\Rmc_0$ is $\Tmc$-consistent.
        \item Thus $\alpha_0 \in \Rmc^* \cap \Rmc_0$. 
    \end{itemize}

    Suppose we proved that the $i$ first elements $\alpha_0, \ldots, \alpha_{i-1}$ of $(\Rmc^* \cup \Rmc_0) \setminus \Bmc_r$ considered in the construction of $\Rmc^*$  are in $\Rmc^* \cap \Rmc_0$, and let $\alpha_i  \in (\Rmc^* \cup \Rmc_0) \setminus \Bmc_r$ be the next one.
    \begin{itemize}
        \item If $\alpha_i \in \Rmc^* \setminus \Rmc_0$ then there exists $\beta \in \Rmc_0 \setminus \Rmc^*$ such that $\beta \succ \alpha_i$ but then $\beta \not \in \Bmc_r$ has been considered before $\alpha_i$ which contradicts our assumption.
        \item If $\alpha_i \in \Rmc_0 \setminus \Rmc^*$ then there exists $\Cmc \in \conflicts{\Kmc}$ such that $\alpha_i \in \Cmc$ and $\Cmc \setminus \{\alpha_i\} \subseteq \Rmc \cup \{\alpha_0, \ldots , \alpha_{i-1}\}$ but then (because $\Rmc = \Rmc_0 \cap \Bmc_r$ and $\{\alpha_0, \ldots, \alpha_{i-1}\} \subseteq \Rmc_0$) we have $\Cmc \subseteq \Rmc_0$. As $\Cmc\in \conflicts{\Kmc}$ this contradicts the fact that $\Rmc_0$ is $\Tmc$-consistent.
        \item Thus $\alpha_i \in \Rmc^* \cap \Rmc_0$. 
    \end{itemize}

    Therefore we have recursively that $\Rmc^* = \Rmc_0$ which contradicts our assumption.

\smallskip
\noindent\textbf{Case $\mathbf{X=C}$.} 
Let $\Rmc\in\creps{\Kmc_{{\succ}^{\Bmc_r}}^{\Bmc_r}}$. 
There exists a completion ${\succ'}^{\Bmc_r}$ of ${\succ}^{\Bmc_r}$ in $\Kmc^{\Bmc_r}$ such that $\Rmc\in\preps{\Kmc_{{\succ'}^{\Bmc_r}}^{\Bmc_r}}$. Let $\succ'$ be a completion of $\succ$ in $\Kmc$ that agrees with ${\succ'}^{\Bmc_r}$ on $\Kmc^{\Bmc_r}$. 
We showed that there exists $\Rmc' \in \preps{\Kmc_{\succ'}}$ such that $\Rmc = \Rmc'\cap \Bmc_r$. Since $\succ'$ is a completion of $\succ$ in $\Kmc$, we obtain $\Rmc' \in \creps{\Kmc_{\succ}}$. 
\end{proof}

\subsubsection{Proof of Theorem~\ref{ThmLocalizationParetoComp}}
Proposition~\ref{prop:ThmLocalizationParetoCompFirst} and Proposition~\ref{prop:ThmLocalizationParetoCompSecond} below 
correspond to the first and second item of Theorem~\ref{ThmLocalizationParetoComp}, respectively. 
Here $\Bmc_r=\reachstrong(\Bmc)$, is the set of facts reachable from $\Bmc$ \wrt `strong' reachability. 

\begin{proposition}\label{prop:ThmLocalizationParetoCompFirst}
    For $X\in\{P,G,C\}$, if $\Rmc\in\xreps{\Kmc_\succ}$, then $\Rmc\cap\Bmc_r\in\xreps{\Kmc^{\Bmc_r}_{\succ^{\Bmc_r}}}$.
\end{proposition}
\begin{proof}
The proof is the same as that of Proposition~\ref{prop:ThmLocalizationGlobalFirst}, except that the proofs of Claims~\ref{lem:gamma_pareto} and~\ref{lem:min_global}, which are the only places where 
$\Bmc_r = \reachweak(\Bmc)$ is used, now use $\Bmc_r = \reachstrong(\Bmc)$. 
\end{proof}

\begin{proposition}\label{prop:ThmLocalizationParetoCompSecond}
    For $X\in\{P,C\}$, if $\Rmc\in \xreps{\Kmc^{\Bmc_r}_{\succ^{\Bmc_r}}}$, then there exists $\Rmc'\in\xreps{\Kmc_\succ}$ such that $\Rmc=\Rmc'\cap\Bmc_r$.
\end{proposition}
\begin{proof}
\smallskip
\noindent\textbf{Case $\mathbf{X=P}$.} 
The proof is the same as the one for Proposition~\ref{prop:ThmLocalizationGlobalSecond} in case $X=P$, except that the proof of Claim~\ref{lem:minconf_pareto}, which is the only place where the fact that $\Bmc_r = \reachweak(\Bmc)$ is used, now use $\Bmc_r = \reachstrong(\Bmc)$.

\smallskip
\noindent\textbf{Case $\mathbf{X=C}$.}
Same proof as Proposition~\ref{prop:ThmLocalizationGlobalSecond}.
\end{proof}

\section{Alternative ASP(Q) Encoding for G-AR}

For G-AR semantics, we considered the following alternative ASP(Q) encoding.
\begin{align*}
\Pi_{\mi{AR}}^{G\ 2}=&\forall^{st}\Pi_{\mi{ReachAll}}\Pi_{\mi{Rep}}\exists^{st}\Pi_{\mi{PartGImp}} : C_{AR}^{G\ 2}
\end{align*}
\begin{align*}
&\text{with }&C_{AR}^{G\ 2}=&\Pi_{\mi{SatIfCause}}\\&&&\{\impl\mt{not~sat, not~global\_improvement.}\}
\end{align*}
and $\Pi_{\mi{PartGImp}}$ is obtained from $\Pi_{\mi{GImp}}$ by removing the constraint ($ \impl $ \mt{not~global\_improvement}). 

This second program for G-AR, $\Pi_{\mi{AR}}^{G\ 2}$, intuitively checks whether for every $\Rmc\in\reps{\Kmc}$, there exists some $\Bmc\subseteq\Dmc$ such that either $\Rmc$ contains some cause for $q(\ans)$ or $\Bmc$ is a global improvement for $\Rmc$, so that $\Rmc\in\greps{\Kmc_\succ}$ implies $(\Rmc,\Tmc)\models q(\ans)$. Specifically, if we let $P_1=\Pi_{\mi{ReachAll}}\cup\Pi_{\mi{Rep}}$ and $P_2=\Pi_{\mi{PartGImp}}$, so that $\Pi_{\mi{AR}}^{G\ 2}=\forall^{st}P_1\exists^{st}P_2 : C_{AR}^{G\ 2}$, we know that $\Pi_{\mi{AR}}^{G\ 2}$ is coherent if for every $M_1\in AS(P_1)$, $\exists^{st}P_2\cup\fix{P_1}{M_1}:C_{AR}^{G\ 2}$ is coherent, \ie if for $P'_2=P_2\cup\fix{P_1}{M_1}$, there exists $M_2\in AS(P'_2)$ such that $C_{AR}^{G\ 2}\cup\fix{P'_2}{M_2}$ is coherent. We have seen that answer sets of $P_1$ correspond to repairs, and that given a repair $\Rmc$, $P_2$ guesses a $\Tmc$-consistent $\Bmc$ and derives $\mt{global\_improvement}$ iff $\Bmc$ is a global improvement of $\Rmc$, and the part $\Pi_{\mi{SatIfCause}}$ of the constraint $C_{AR}^{G\ 2}$ derives $\mt{sat}$ iff $\Rmc$ contains some cause. Hence, the remaining constraint,  
$\impl\mt{not~sat, not~global\_improvement}$, which ensures that either $\mt{sat}$ or $\mt{global\_improvement}$ has been derived, enforces the above characterization of $\Kmc_\succ\armodels{G}q(\ans)$.
\begin{proposition}
    $\mathcal{K}_{\succ} \armodels{G} q(\vec{a})$  iff $\Pi_{\mi{AR}}^{G\ 2}$ is coherent.
\end{proposition}

\section{Experimental Setting}
This section recalls statistics on the benchmarks we used in our experiments (\cf \cite{DBLP:phd/hal/Bourgaux16,DBLP:conf/kr/BienvenuB22,DBLP:conf/kr/BienvenuBIJ2025}). 

\subsubsection{Datasets} Table~\ref{tab:datasets} provides information on the size and conflicts of the datasets from the \mn{CQAPri} benchmark that we used (binary conflicts).  Regarding the density of the conflict graph, in \mn{u1c1}, each of the facts involved in some conflict is in conflict with between 1 and 614 facts with an average of 2, and in \mn{u20c50}, each of the facts involved in some conflict is in conflict with between 1 and 744 facts with an average of 6.6. 
The non-binary case adds the same 40 conflicts of size 10 to all datasets \cite{DBLP:conf/kr/BienvenuBIJ2025}. 

\subsubsection{Priority Relations} 
For the binary conflicts case, we use two priority relations from the \mn{ORBITS} benchmark: $\succ^{ss}$ is score-structured and was built using 5 levels of priority to which facts were randomly assigned, and $\succ^{ns}$ was built by considering each (binary) conflict and assigning a random preference between the facts with a probability $0.8$, except if doing so created a cycle, then checking that $\succ^{ns}$ was indeed not score-structured \cite{DBLP:conf/kr/BienvenuB22}. 
These priority relations are such that $\succ^{ss}$ assigns a priority between the two facts of about 80\% of the conflicts and this proportion is about 60\%-65\% for $\succ^{ns}$. 
For the non-binary case, we use one priority relation $\succ^{nb}$ resulting from the set of preference rules $\Sigma_1^a\cup\Sigma_2^a\cup\Sigma_3^a$ and ``going down'' cycle resolution strategy defined by \citeauthor{DBLP:conf/kr/BienvenuBIJ2025}~\shortcite{DBLP:conf/kr/BienvenuBIJ2025}. This priority relation assigns a priority between two facts that belong to some conflict in about 90\% of the cases.

\subsubsection{Queries} We use 8 queries from the \mn{CQAPri} benchmark: \mn{q3}, \mn{q5}, \mn{q7}, \mn{q10}, \mn{q11}, \mn{q14}, \mn{q15}, and \mn{q20} (see \cite[Section~3.3.1]{DBLP:phd/hal/Bourgaux16} for the queries and characteristics). 
Table~\ref{tab:number-answers-per-query} shows the number of potential answers per query on the datasets we use for the query answering task.

\subsubsection{Inputs Considered for Each Task} 
For the task of computing the grounded repair, we considered the 18 datasets presented in Table~\ref{tab:datasets}, with the 2 priority relations given by the \mn{ORBITS} benchmark, plus the 8 datasets with non-binary conflicts \mn{u1cY} and \mn{u5c1}, \mn{u5c5} with the priority relation $\succ^{nb}$ (\cf Table~\ref{tab:time-grounded}). The reason we were not able to use larger datasets in the case of non-binary conflicts is that we fail to compute the priority relation $\succ^{nb}$ for datasets larger than \mn{u5c5}. 
For the task of query answering, we always use the 8 queries and considered a subset of these cases, eliminating the bigger (\mn{u20cY} and \mn{u5c30}, \mn{u5c50}) datasets, which yields 10 datasets with 2 priority relations plus 8 datasets with non-binary conflicts and one priority relation (cf. Table~\ref{tab:number-answers-appendix}).  

\begin{table*}
\centering
\begin{tabular*}{\textwidth}{l@{\extracolsep{\fill}}rrrr}
\toprule
& \# facts & \# facts in conflicts &\% facts in conflicts & \# conflicts
\\
\midrule
\mn{u1c1}& 75,724 & 2,373 & 3 & 2,314
\\
\mn{u1c5}& 75,951 & 6,412 & 8 & 8,476 
\\
\mn{u1c10}&76,201 & 10,891 & 14 & 14,261 
\\
\mn{u1c20}&76,821  & 20,175 &  26 & 28,232
\\
\mn{u1c30}&77,447 & 26,086 &  34 & 45,484
\\
\mn{u1c50}&78,593 & 34,814 &  44 & 81,304
\\
\midrule
\mn{u5c1}& 463,691  & 12,191 & 3 & 11,984 
\\
\mn{u5c5}&465,157 & 45,906 & 10 & 53,398 
\\
\mn{u5c10}&466,919  & 83,263 & 18 & 109,453
\\
\mn{u5c20}&470,674  & 137,836 & 29 & 231,771 
\\
\mn{u5c30}&474,368 & 172,245 & 36 & 345,461
\\
\mn{u5c50}&481,400 & 221,900 & 46 & 583,714
\\
\midrule
\mn{u20c1}& 1,983,493  & 69,597 & 4 & 73,212
\\
\mn{u20c5}& 1,989,788 & 253,141 & 13 & 335,194
\\
\mn{u20c10}& 1,997,445  & 408,398 & 20 & 662,725
\\
\mn{u20c20}& 2,013,048  & 610,271 & 30 & 1,314,991
\\
\mn{u20c30}& 2,028,069  & 748,664 & 37 & 1,933,956
\\
\mn{u20c50}& 2,056,957  & 946,819 & 46 &   3,130,377
\\
\bottomrule
\end{tabular*}
\caption{Size, number and percentage of facts involved in some conflict, and number of conflicts for each dataset of the \mn{CQAPri} benchmark (binary conflicts) that we used in our experiments.}
\label{tab:datasets}
\end{table*}

\begin{table*}
\centering
\begin{tabular*}{\textwidth}{l@{\extracolsep{\fill}}rrrrrrrrr}
\toprule
&  \mn{q3}& \mn{q5}& \mn{q7}& \mn{q10}& \mn{q11}& \mn{q14}& \mn{q15}& \mn{q20} & Total
\\
\midrule
\mn{u1c1} & 85 & 10 & 137 & 3 & 538 & 195 & 507 & 50 & 1525\\
\mn{u1c5} & 85 & 10 & 137 & 3 & 544 & 195 & 508 & 50 & 1532\\
\mn{u1c10} & 85 & 10 & 138 & 6 & 551 & 195 & 508 & 50 & 1543\\
\mn{u1c20} & 85 & 10 & 138 & 6 & 564 & 195 & 508 & 50 & 1556\\
\mn{u1c30} & 87 & 10 & 142 & 6 & 585 & 195 & 508 & 50 & 1583\\
\mn{u1c50} & 87 & 10 & 149 & 7 & 628 & 195 & 515 & 50 & 1641\\
\midrule
\mn{u5c1} & 85 & 10 & 137 & 20 & 3366 & 1112 & 3183 & 50 & 7963\\
\mn{u5c5} & 85 & 10 & 137 & 20 & 3408 & 1112 & 3188 & 50 & 8010\\
\mn{u5c10} & 85 & 10 & 138 & 23 & 3472 & 1112 & 3190 & 50 & 8080\\
\mn{u5c20} & 85 & 10 & 138 & 24 & 3584 & 1112 & 3203 & 50 & 8206\\
\bottomrule
\end{tabular*}
\caption{Number of potential answers.}
\label{tab:number-answers-per-query}
\end{table*}

\section{Experimental Results}
We present here additional results from our experimental evaluation. 
\subsubsection{Tractable Approximations} 
Tables~\ref{tab:grounded_trivial_prio_non_score}, \ref{tab:grounded_trivial_prio_score} and \ref{tab:grounded_trivial_non_binary} show the number of plain (S-)IAR facts (which are not involved in any conflict), other trivially P-IAR facts ($\Gamma(\emptyset)$), and facts in the grounded repair ($\Gmc$). 
Tables~\ref{tab:number-trivial-grounded-answers-non-score}, \ref{tab:number-trivial-grounded-answers-score} and \ref{tab:number-trivial-grounded-answers-non-binary} show the number of \emph{query answers} that are trivially P-IAR, grounded, and the number of answers additionally found by each incremental computation step of the grounded repair. 
Table~\ref{tab:time-grounded} shows the time needed to compute the facts that belong to the grounded repair or to $\Gamma(\emptyset)$ (among those that belong to some conflict) from the precomputed attack relation.

\subsubsection{Impact of Specific Encoding for Binary Conflicts} 
Figures~\ref{fig:impact-binary-AR-appendix} and \ref{fig:impact-binary-brave-appendix} compare for each potential answer (which does not hold under grounded), dataset with binary conflicts, and priority relation, the time needed to decide if it holds under X-AR and X-brave semantics, respectively, with and without using the optimization proposed in the case of binary conflicts.

\subsubsection{Comparison of the Two ASP(Q) Encodings for G-AR} Figure~\ref{fig:comparisonGAR-appendix} shows how $\Pi_{\mi{AR}}^{G}$ and $\Pi_{\mi{AR}}^{G\ 2}$ compare. 
We choose to focus on $\Pi_{\mi{AR}}^{G}$ which appears to perform a bit better in general.

\subsubsection{Semantics Comparison in Terms of Runtimes} 
Figures~\ref{fig:overall-runtime-appendix} and \ref{fig:overall-runtime-nonbinary-appendix} show the times needed to decide whether a potential answer which does not hold under grounded semantics holds under a given semantics for X-AR and X-brave with $X\in\{P,G,C\}$, in our different scenarios, using localization and the binary conflicts-specific encoding when possible. 
Instances are sorted by solving time, so a point $(i,j)$ indicates that for the considered semantics, $i$ instances  could be (individually) solved within $j$ seconds.

\subsubsection{Semantics Comparison in Terms of Answers}
Table~\ref{tab:number-answers-appendix} shows the number of answers we obtain under the different semantics and the number of potential (non grounded) answers for which we did not manage to determine whether they hold under the considered semantics in our 600s time limit.

\begin{table*}
\centering
\begin{tabular*}{\textwidth}{l@{\extracolsep{\fill}}cccccccccc}
\toprule
 & \mn{u1c1} & \mn{u1c5} & \mn{u1c10} & \mn{u1c20} & \mn{u1c30} & \mn{u1c50} & \mn{u5c1} & \mn{u5c5} & \mn{u5c10} & u5c20 \\
\midrule
$\cap_{S}$ & 73351 & 69539 & 65310 & 56646 & 51361 & 43779      & 451500 & 419251 & 383656 & 332838 \\
$\Gamma(\emptyset)\setminus\cap_{S}$ & 921 & 1705 & 3021 & 5391 & 6105 & 7289 & 3740 & 13149 & 22920 & 32449 \\
$\Gmc\setminus\cap_{S}$ & 2286 & 6009 & 10134 & 18507 & 23339 & 29945 & 11704 & 43435 & 78368 & 126482 \\
\midrule
\% $\Gmc$ not $\Gamma(\emptyset)$  & 1.8\% & 5.7\% & 9.4\% & 17.5\% & 23.1\% & 30.7\% & 1.7\% & 6.5\% & 12\% & 20.5\% \\
\bottomrule
\end{tabular*}
\caption{Number of facts in $\cap_{S}=\bigcap_{\Rmc\in\reps{\Kmc}}\Rmc$, $\Gamma(\emptyset)$ (trivially P-IAR), and the grounded repair $\Gmc$, in $\succ^{ns}$ case.}
\label{tab:grounded_trivial_prio_non_score}
\bigskip

\begin{tabular*}{\textwidth}{l@{\extracolsep{\fill}}cccccccccc}
\toprule
 & \mn{u1c1} & \mn{u1c5} & \mn{u1c10} & \mn{u1c20} & \mn{u1c30} & \mn{u1c50} & \mn{u5c1} & \mn{u5c5} & \mn{u5c10} & u5c20 \\
\midrule
$\cap_{S}$ & 73351 & 69539 & 65310 & 56646 & 51361 & 43779      & 451500 & 419251 & 383656 & 332838 \\
$\Gamma(\emptyset)\setminus\cap_{S}$ & 1225 & 2193 & 3981 & 5988 & 6851 & 8493 & 5354 & 16022 & 27965 & 40829 \\
$\Gmc\setminus\cap_{S}$ & 2106 & 4256 & 7831 & 13020 & 15235 & 17872 & 10036 & 33370 & 59103 & 90164 \\
\midrule
\% $\Gmc$ not $\Gamma(\emptyset)$  & 1.2\% & 2.8\% & 5.3\% & 10.1\% & 12.6\% & 15.2\% & 1\% & 3.8\% & 7\% & 11.7\% \\
\bottomrule
\end{tabular*}
\caption{Number of facts in $\cap_{S}=\bigcap_{\Rmc\in\reps{\Kmc}}\Rmc$, $\Gamma(\emptyset)$ (trivially P-IAR), and the grounded repair $\Gmc$, in $\succ^{ss}$ case.}
\label{tab:grounded_trivial_prio_score}
\bigskip

\begin{tabular*}{\textwidth}{l@{\extracolsep{\fill}}cccccccc}
\toprule
 & \mn{u1c1} & \mn{u1c5} & \mn{u1c10} & \mn{u1c20} & \mn{u1c30} & \mn{u1c50} & \mn{u5c1} & u5c5 \\
\midrule
$\cap_{S}$  & 73131 & 69327 & 65098 & 56439 & 51165 & 43612 & 451280 & 419039\\
$\Gamma(\emptyset)\setminus\cap_{S}$ & 989 & 2561 & 4364 & 7668 & 8883 & 10777 & 5994 & 21399 \\
$\Gmc\setminus\cap_{S}$ & 2399 & 6089 & 10066 & 18279 & 22916 & 29458 & 11684 & 42969 \\
\midrule
\% $\Gmc$ not $\Gamma(\emptyset)$ & 1.8\% & 4.7\% & 7.6\% & 14.2\% & 18.9\% & 25.6\% & 1.2\% & 4.7\% \\
\bottomrule
\end{tabular*}
\caption{Number of facts in $\cap_{S}=\bigcap_{\Rmc\in\reps{\Kmc}}\Rmc$, $\Gamma(\emptyset)$ (trivially P-IAR), and the grounded repair $\Gmc$, in $\succ^{nb}$ case.}
\label{tab:grounded_trivial_non_binary}
\end{table*}

\begin{table}
\centering
\begin{tabular*}{0.47\textwidth}{l@{\extracolsep{\fill}}cccccc}
\toprule
 & Triv & GR & Step 1 & Step 2 & Step 3 \\
\midrule
\mn{u1c1} & 1465 & 1524 & 59 & - & - \\
\mn{u1c5} & 1366 & 1525 & 159 & - & - \\
\mn{u1c10} & 1186 & 1527 & 341 & 0 & - \\
\mn{u1c20} & 937 & 1521 & 555 & 29 & - \\
\mn{u1c30} & 754 & 1518 & 730 & 34 & - \\
\mn{u1c50} & 625 & 1500 & 800 & 75 & 0 \\
\mn{u5c1} & 7637 & 7953 & 316 & - & - \\
\mn{u5c5} & 6958 & 7948 & 990 & 0 & - \\
\mn{u5c10} & 6069 & 7964 & 1894 & 1 & - \\
\mn{u5c20} & 4947 & 7982 & 2972 & 62 & 1 \\
\bottomrule
\end{tabular*}
\caption{Number of answers found trivially P-IAR (Triv) and grounded (GR). Column ``Step $i$'' gives the number of answers \wrt $\Gamma^{i+1}(\emptyset)$ minus those \wrt $\Gamma^i(\emptyset)$. ``-'' indicates that this step was not needed to compute the grounded repair. Case of $\succ^{ns}$.}\label{tab:number-trivial-grounded-answers-non-score}
\bigskip

\begin{tabular*}{0.47\textwidth}{l@{\extracolsep{\fill}}ccccc}
\toprule
 & Triv & GR  & Step 1 & Step 2 \\
\midrule
\mn{u1c1} & 1477 & 1525 & 48 & - \\
\mn{u1c5} & 1385 & 1450 & 65 & - \\
\mn{u1c10} & 1215 & 1447 & 232 & - \\
\mn{u1c20} & 931 & 1305 & 374 & 0 \\
\mn{u1c30} & 731 & 1030 & 299 & 0 \\
\mn{u1c50} & 618 & 770 & 152 & 0 \\
\mn{u5c1} & 7706 & 7901 & 195 & - \\
\mn{u5c5} & 7044 & 7635 & 591 & - \\
\mn{u5c10} & 6213 & 7300 & 1087 & 0 \\
\mn{u5c20} & 5169 & 6649 & 1480 & 0 \\
\bottomrule
\end{tabular*}
\caption{Number of answers found trivially P-IAR (Triv) and grounded (GR). Column ``Step $i$'' gives the number of answers \wrt $\Gamma^{i+1}(\emptyset)$ minus those \wrt $\Gamma^i(\emptyset)$. ``-'' indicates that this step was not needed to compute the grounded repair. Case of $\succ^{ss}$.}\label{tab:number-trivial-grounded-answers-score}
\bigskip

\begin{tabular*}{0.47\textwidth}{l@{\extracolsep{\fill}}ccccc}
\toprule
 & Triv & GR & Step 1 & Step 2 \\
\midrule
\mn{u1c1} & 1447 & 1523 & 75 & 1 \\
\mn{u1c5} & 1369 & 1524 & 155 & 0 \\
\mn{u1c10} & 1167 & 1526 & 359 & 0 \\
\mn{u1c20} & 940 & 1511 & 571 & 0 \\
\mn{u1c30} & 789 & 1515 & 726 & 0 \\
\mn{u1c50} & 680 & 1503 & 817 & 6 \\
\mn{u5c1} & 7708 & 7946 & 237 & 1 \\
\mn{u5c5} & 7209 & 7951 & 742 & 0 \\
\bottomrule
\end{tabular*}
\caption{Number of answers found trivially P-IAR (Triv) and grounded (GR). Column ``Step $i$'' gives the number of answers \wrt $\Gamma^{i+1}(\emptyset)$ minus those \wrt $\Gamma^i(\emptyset)$. ``-'' indicates that this step was not needed to compute the grounded repair. Case of $\succ^{nb}$.}\label{tab:number-trivial-grounded-answers-non-binary}
\end{table}

\begin{table}
    \centering
        \begin{tabular*}{0.47\textwidth}{l@{\extracolsep{\fill}}rr}
\toprule

        	    & $\Gmc\setminus\cap_{S}$	& $\Gamma(\emptyset)\setminus\cap_{S}$\\
        \midrule
        \mn{u1c1} $\succ^{ns}$	    & 2.10	    & 0.16\\
        \mn{u1c5} $\succ^{ns}$	    & 6.63	    & 0.33\\
        \mn{u1c10} $\succ^{ns}$	    & 9.25	    & 0.51\\
        \mn{u1c20} $\succ^{ns}$	    & 21.22	    & 1.39\\
        \mn{u1c30} $\succ^{ns}$	    & 47.65	    & 2.43\\
        \mn{u1c50} $\succ^{ns}$	    & 126.39	    & 3.80\\
        \mn{u5c1} $\succ^{ns}$	    & 10.70	    & 0.68\\
        \mn{u5c5} $\succ^{ns}$	    & 54.50	    & 2.69\\
        \mn{u5c10} $\succ^{ns}$	    & 71.88	    & 4.08\\
        \mn{u5c20} $\succ^{ns}$	    & 240.34	    & 7.11\\
        \mn{u5c30} $\succ^{ns}$	    & 376.27	    & 10.53\\
        \mn{u5c50} $\succ^{ns}$	    & MEMOUT	    & 14.69\\
        \mn{u20c1} $\succ^{ns}$	    & 47.05	    & 2.39\\
        \mn{u20c5} $\succ^{ns}$	    & 219.98	    & 9.88\\
        \mn{u20c10} $\succ^{ns}$	& 441.24	    & 18.09\\
        \mn{u20c20} $\succ^{ns}$	& MEMOUT	    & 33.13\\
        \mn{u20c30} $\succ^{ns}$	& MEMOUT	    & 50.23\\
        \mn{u20c50} $\succ^{ns}$	& MEMOUT	    & 80.3\\
        \midrule
        \mn{u1c1} $\succ^{ss}$        & 1.58	    & 0.16\\
        \mn{u1c5} $\succ^{ss}$        & 4.62	    & 0.50\\
        \mn{u1c10} $\succ^{ss}$        & 7.30	    & 0.69\\
        \mn{u1c20} $\succ^{ss}$        & 14.92	    & 1.26\\
        \mn{u1c30} $\succ^{ss}$        & 32.41	    & 2.40\\
        \mn{u1c50} $\succ^{ss}$        & 61.80	    & 3.53\\
        \mn{u5c1} $\succ^{ss}$        & 7.85	    & 0.68\\
        \mn{u5c5} $\succ^{ss}$        & 27.04	    & 2.65\\
        \mn{u5c10} $\succ^{ss}$        & 52.14	    & 4.33\\
        \mn{u5c20} $\succ^{ss}$        & 129.31	    & 7.12\\
        \mn{u5c30} $\succ^{ss}$        & 201.37	    & 10.60\\
        \mn{u5c50} $\succ^{ss}$        & 304.11	    & 14.55\\
        \mn{u20c1} $\succ^{ss}$        & 23.75	    & 2.41\\
        \mn{u20c5} $\succ^{ss}$        & 106.77	    & 9.47\\
        \mn{u20c10} $\succ^{ss}$    & 341.53	    & 18.24\\
        \mn{u20c20} $\succ^{ss}$    & MEMOUT	    & 36.07\\
        \mn{u20c30} $\succ^{ss}$    & MEMOUT	    & 51.29\\
        \mn{u20c50} $\succ^{ss}$    & MEMOUT	    & 80.31\\
        \midrule
        \mn{u1c1} $\succ^{nb}$	    & 1.19	    & 0.14\\
        \mn{u1c5} $\succ^{nb}$	    & 4.03	    & 0.66\\
        \mn{u1c10} $\succ^{nb}$	    & 6.82	    & 1.02\\
        \mn{u1c20} $\succ^{nb}$	    & 13.81	    & 1.50\\
        \mn{u1c30} $\succ^{nb}$	    & 25.34	    & 2.33\\
        \mn{u1c50} $\succ^{nb}$	    & 53.35	    & 3.97\\
        \mn{u5c1} $\succ^{nb}$	    & 4.55	    & 0.83\\
        \mn{u5c5} $\succ^{nb}$	    & 23.03	    & 2.89\\
        \bottomrule
    \end{tabular*}
    \caption{Time (in seconds) to compute the facts that belong to the grounded repair $\Gmc$ or to $\Gamma(\emptyset)$ (among those that belong to some conflict) from the precomputed attack relation.}
    \label{tab:time-grounded}
\end{table}

\begin{figure*}
    \includegraphics[width=.33\textwidth]{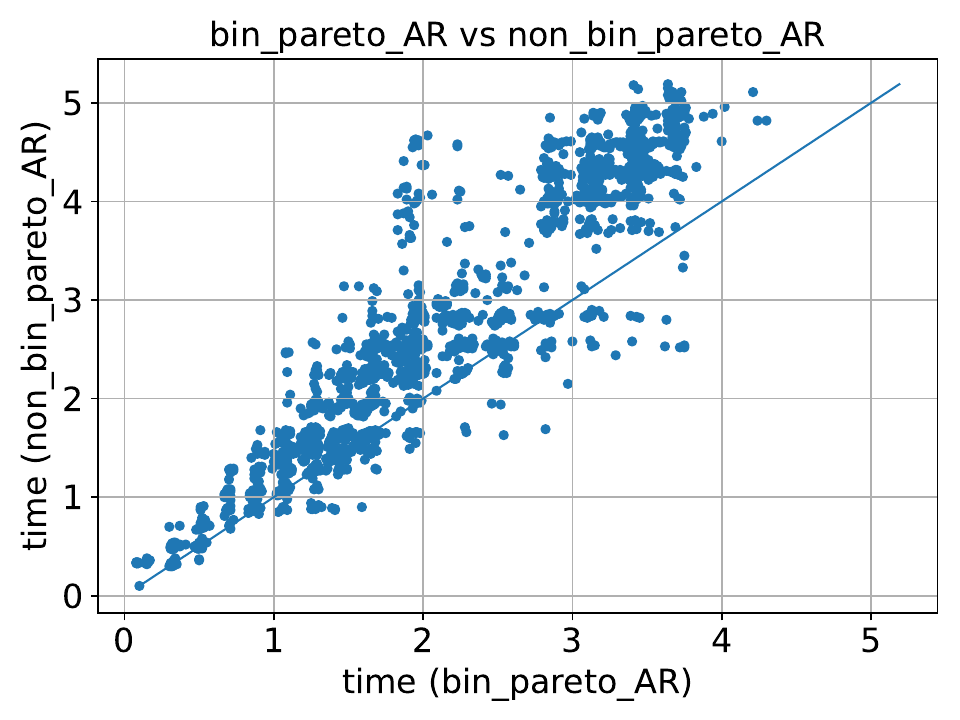}
    \includegraphics[width=.33\textwidth]{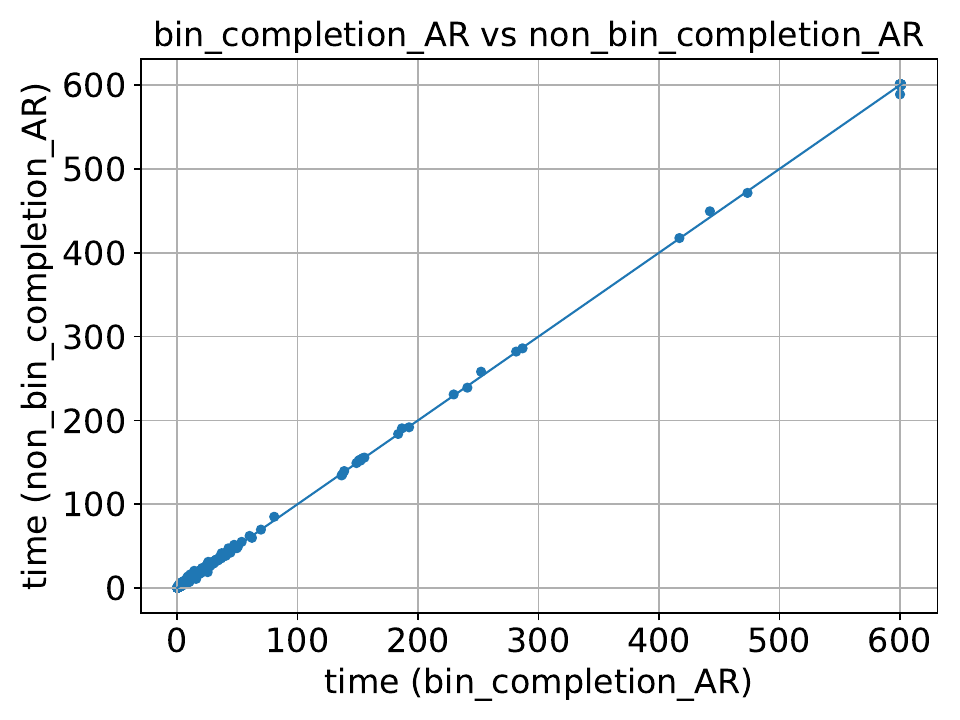}
    \includegraphics[width=.33\textwidth]{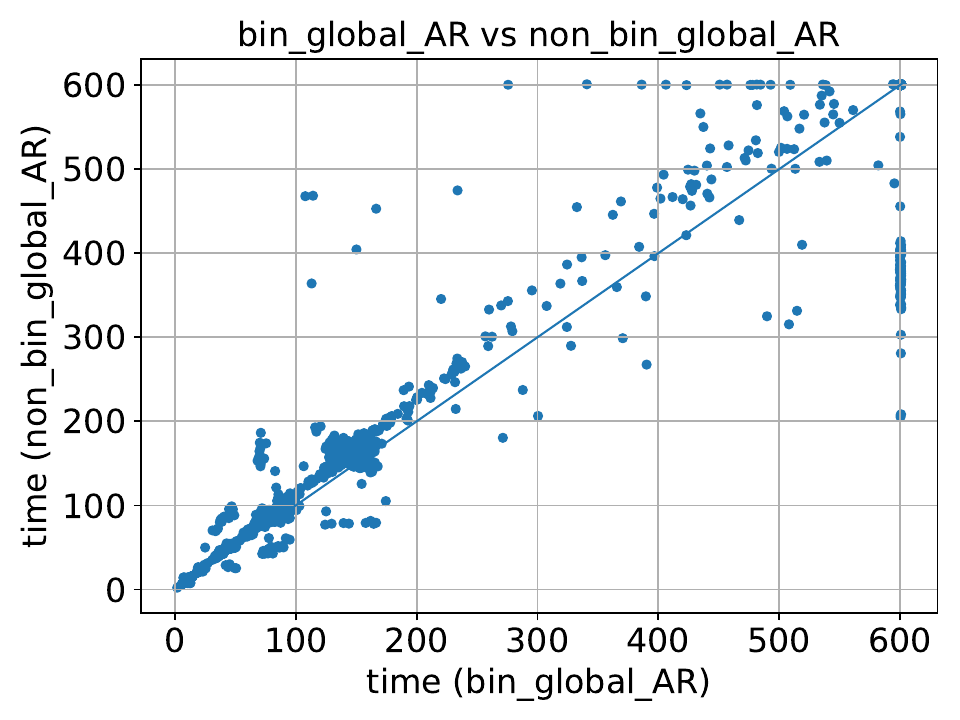}
    \caption{Time for the generic encoding (y-axis) vs 
time for the encoding specific for binary conflicts (x-axis), both with localization, to decide for each $q(\ans)$, dataset \mn{uXcY}, and priority relation, 
if $q(\ans)$ holds under X-AR. (left) X=P, (middle) X=C, (right) X=G}\label{fig:impact-binary-AR-appendix}
\end{figure*}

\begin{figure*}
    \includegraphics[width=.33\textwidth]{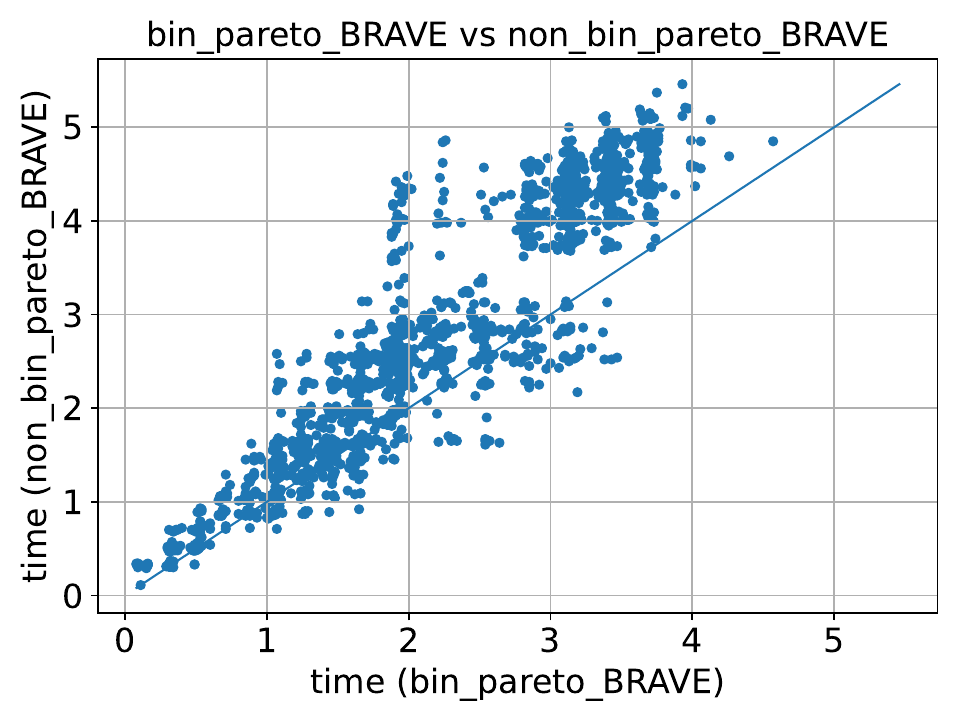}
    \includegraphics[width=.33\textwidth]{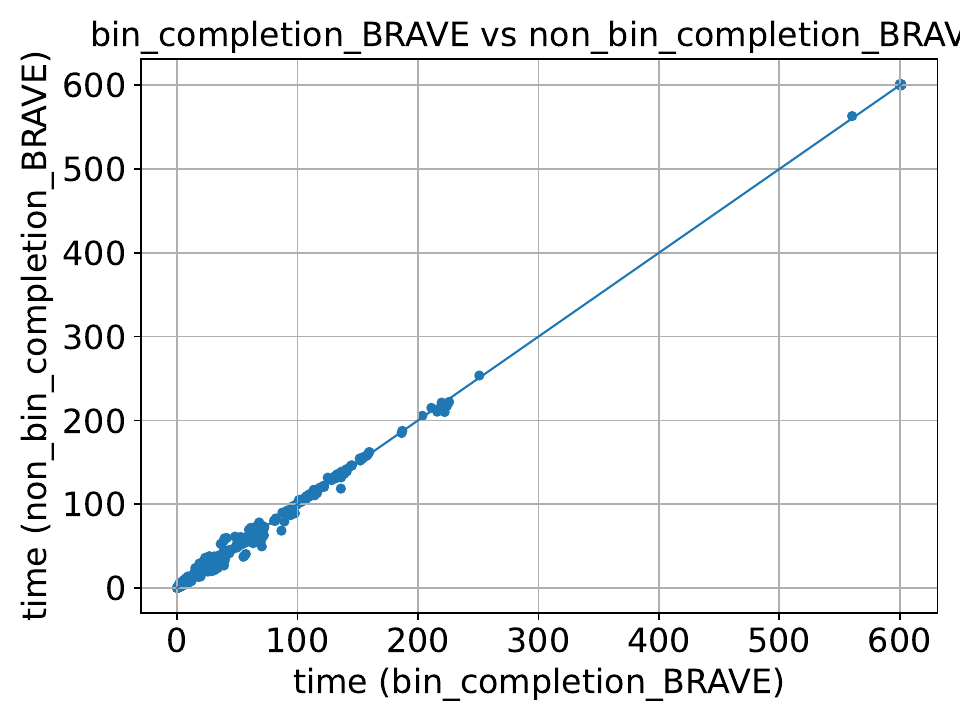}
    \includegraphics[width=.33\textwidth]{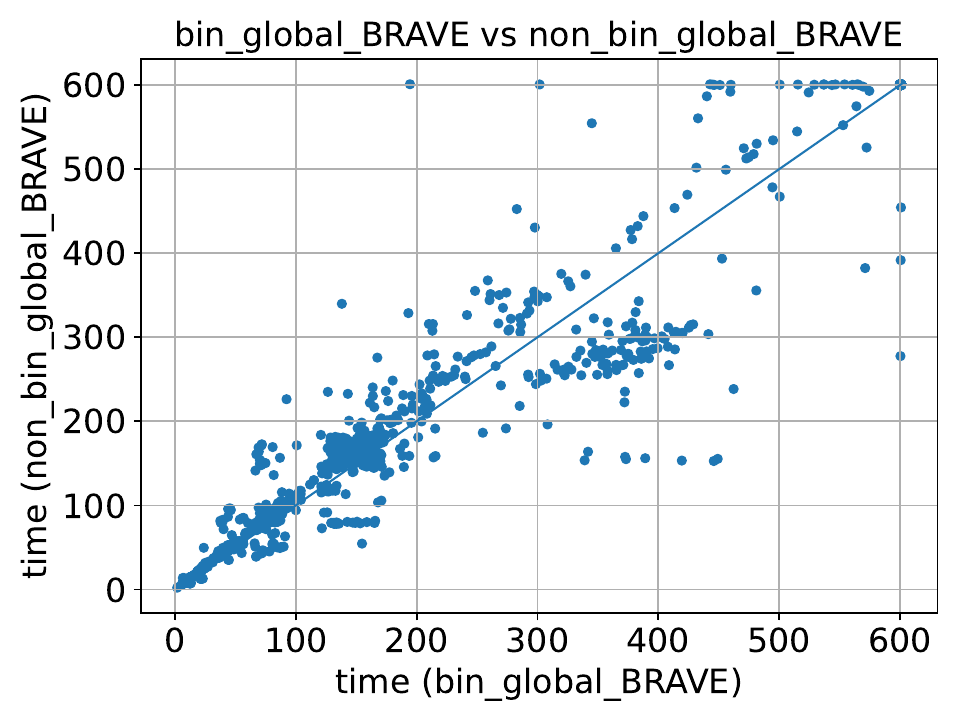}
    \caption{Time for the generic encoding (y-axis) vs 
time for the encoding specific for binary conflicts (x-axis), both with localization, to decide for each $q(\ans)$, dataset \mn{uXcY}, and priority relation, 
if $q(\ans)$ holds under X-brave. (left) X=P, (middle) X=C, (right) X=G}\label{fig:impact-binary-brave-appendix}
\end{figure*}

\begin{figure*}    \includegraphics[width=.49\textwidth,valign=t]{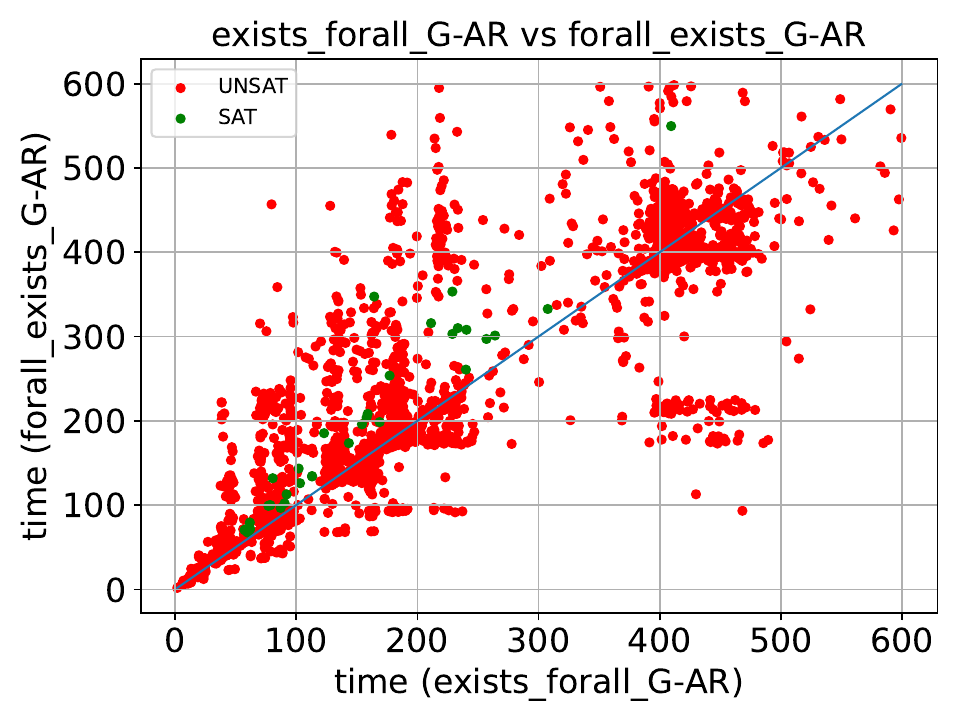}  \includegraphics[width=.49\textwidth,valign=t]{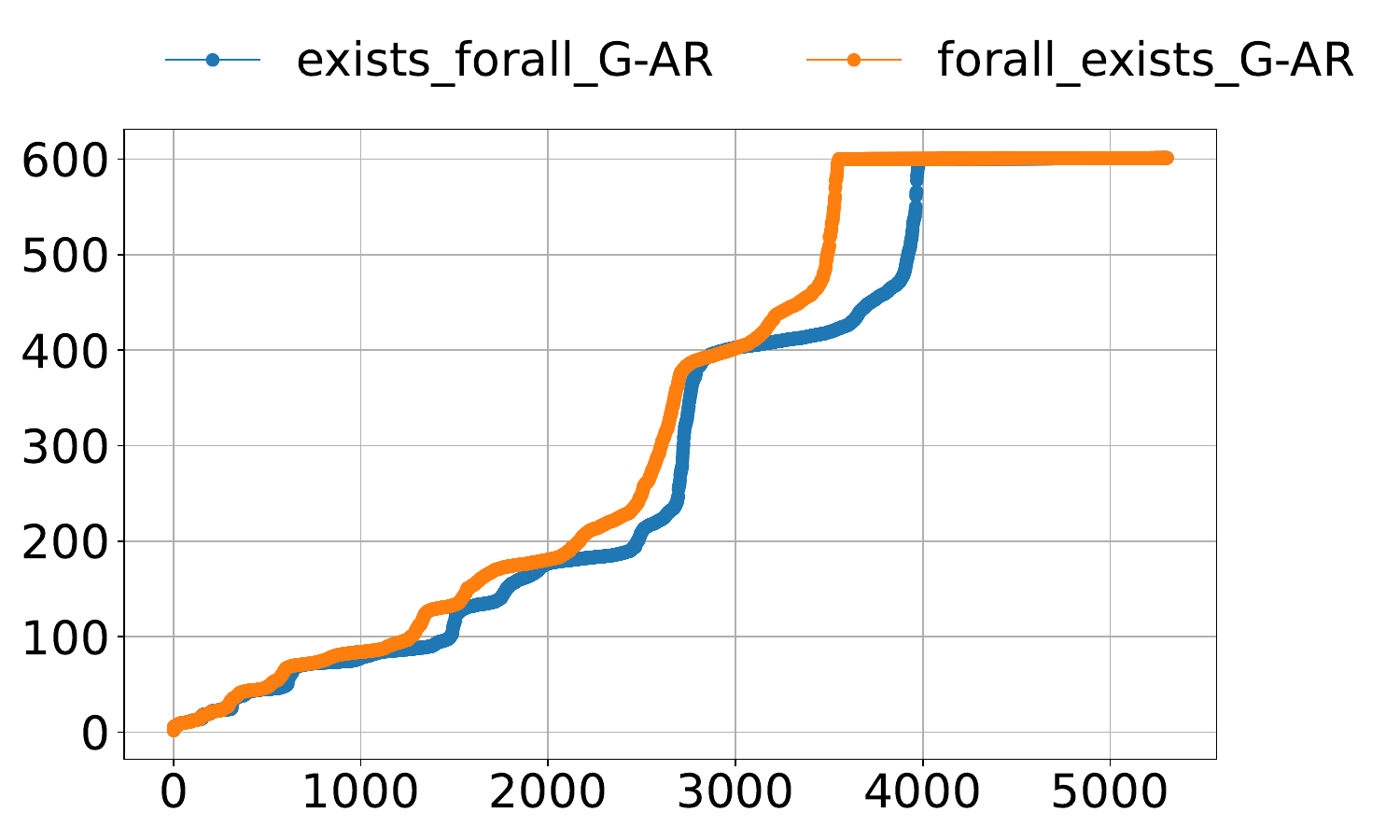}
    \caption{Comparison of $\Pi_{\mi{AR}}^{G}$ (`exists\_forall\_G-AR') and $\Pi_{\mi{AR}}^{G\ 2}$ (`forall\_exists\_G-AR'), both with localization and specific encoding for binary conflicts. (left) Time for $\Pi_{\mi{AR}}^{G\ 2}$ (y-axis) vs time for $\Pi_{\mi{AR}}^{G}$ (x-axis) to decide for each $q(\ans)$, dataset \mn{u1cY} (binary conflict cases), and priority relation, 
if $q(\ans)$ holds under G-AR. (right) Time (in seconds, y-axis) to decide for each $q(\ans)$, dataset \mn{uXcY} (binary conflict cases), and priority relation, whether $q(\ans)$ holds under G-AR using $\Pi_{\mi{AR}}^{G}$ and $\Pi_{\mi{AR}}^{G\ 2}$, with instances sorted by increasing solving time 
along the x-axis (\ie a point $(i,j)$ indicates that $i$ instances could be solved within $j$ seconds).
}\label{fig:comparisonGAR-appendix}
\end{figure*}

\begin{figure*}[t]
    \centering
    \begin{subfigure}{0.24\textwidth}
        \centering
        \includegraphics[width=\linewidth]{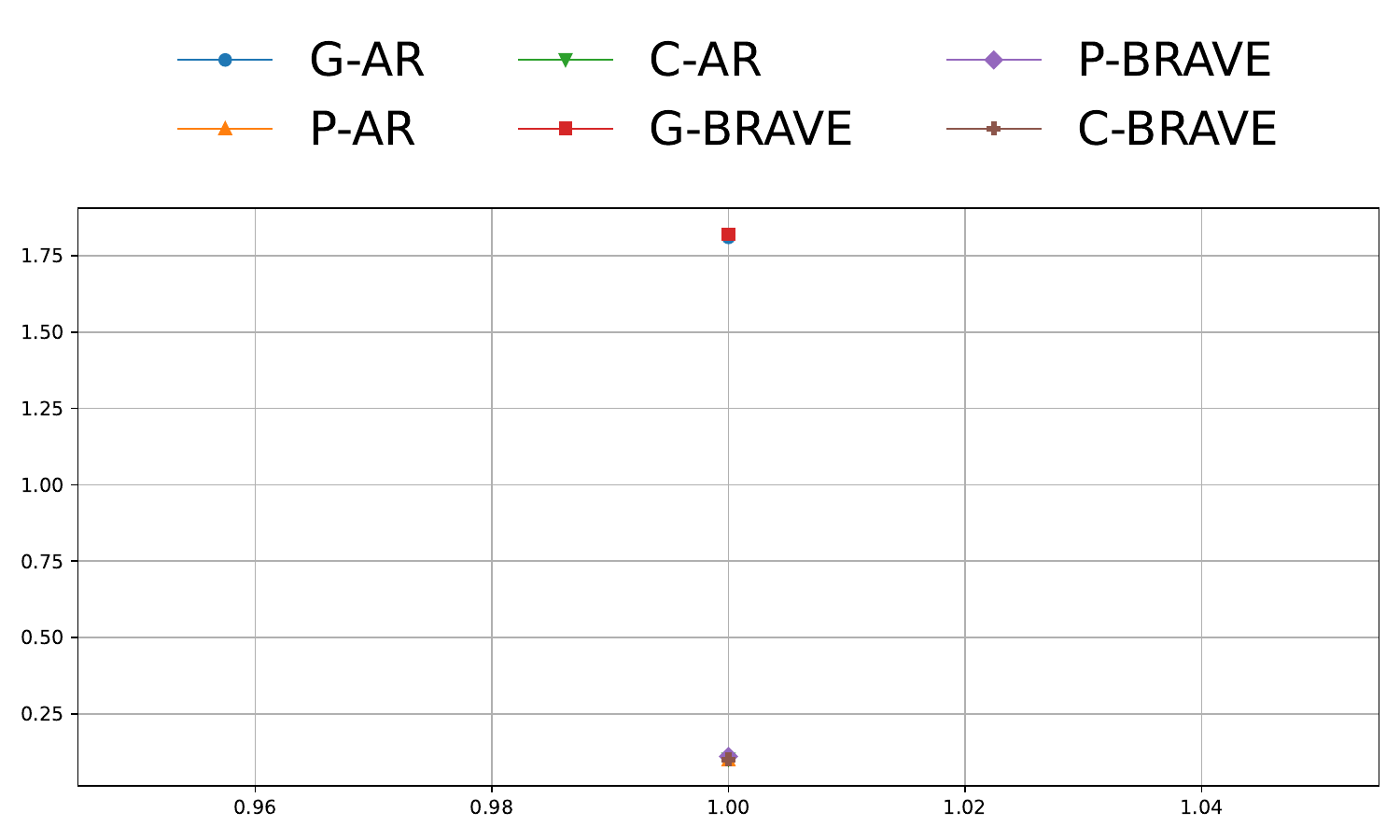}
        \caption{\mn{u1c1} $\succ^{ns}$}
        \label{fig:u1c1_non_score}
    \end{subfigure}
    \begin{subfigure}{0.24\textwidth}
        \centering
        \includegraphics[width=\linewidth]{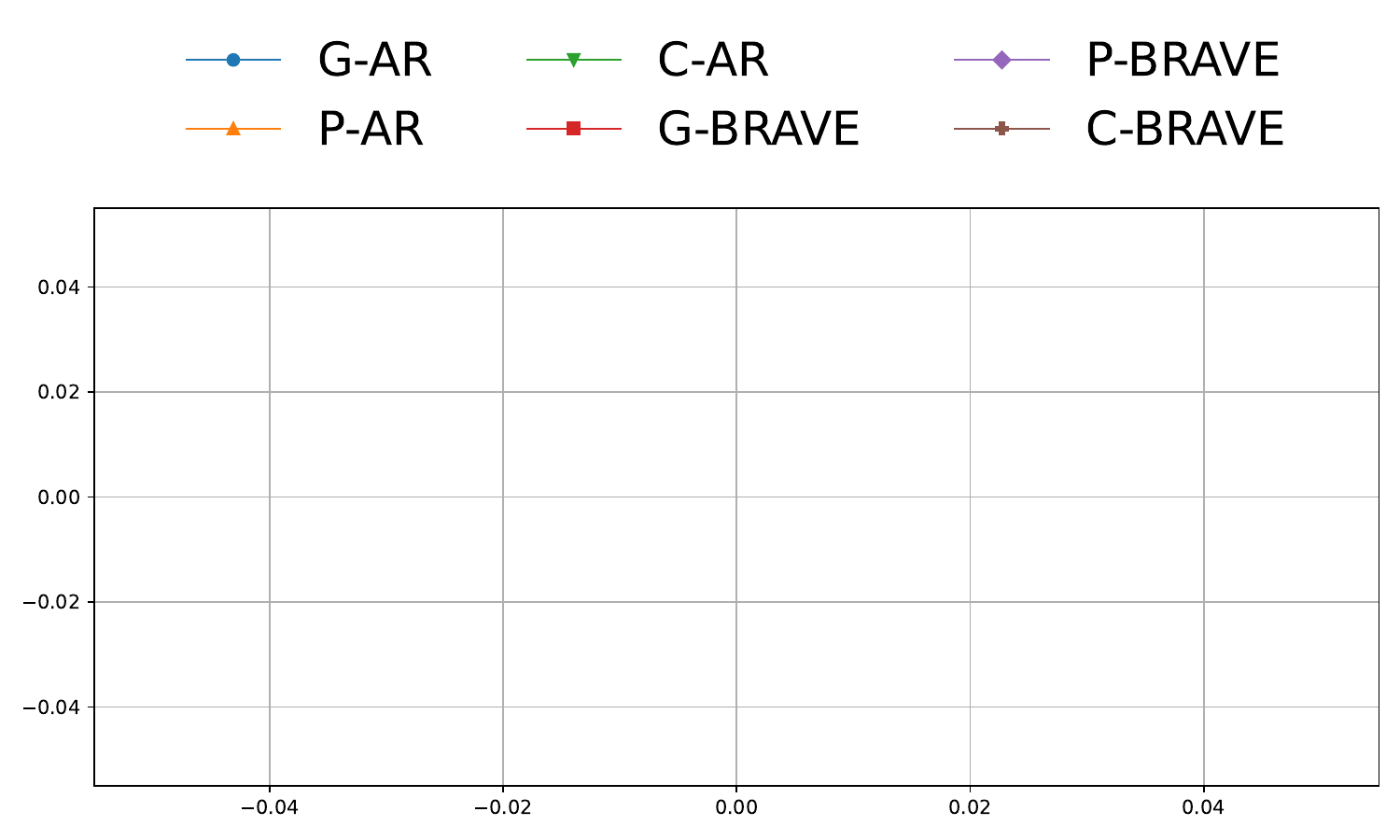}
        \caption{\mn{u1c1} $\succ^{ss}$}
        \label{fig:u1c1_score}
    \end{subfigure}
    \begin{subfigure}{0.24\textwidth}
        \centering
        \includegraphics[width=\linewidth]{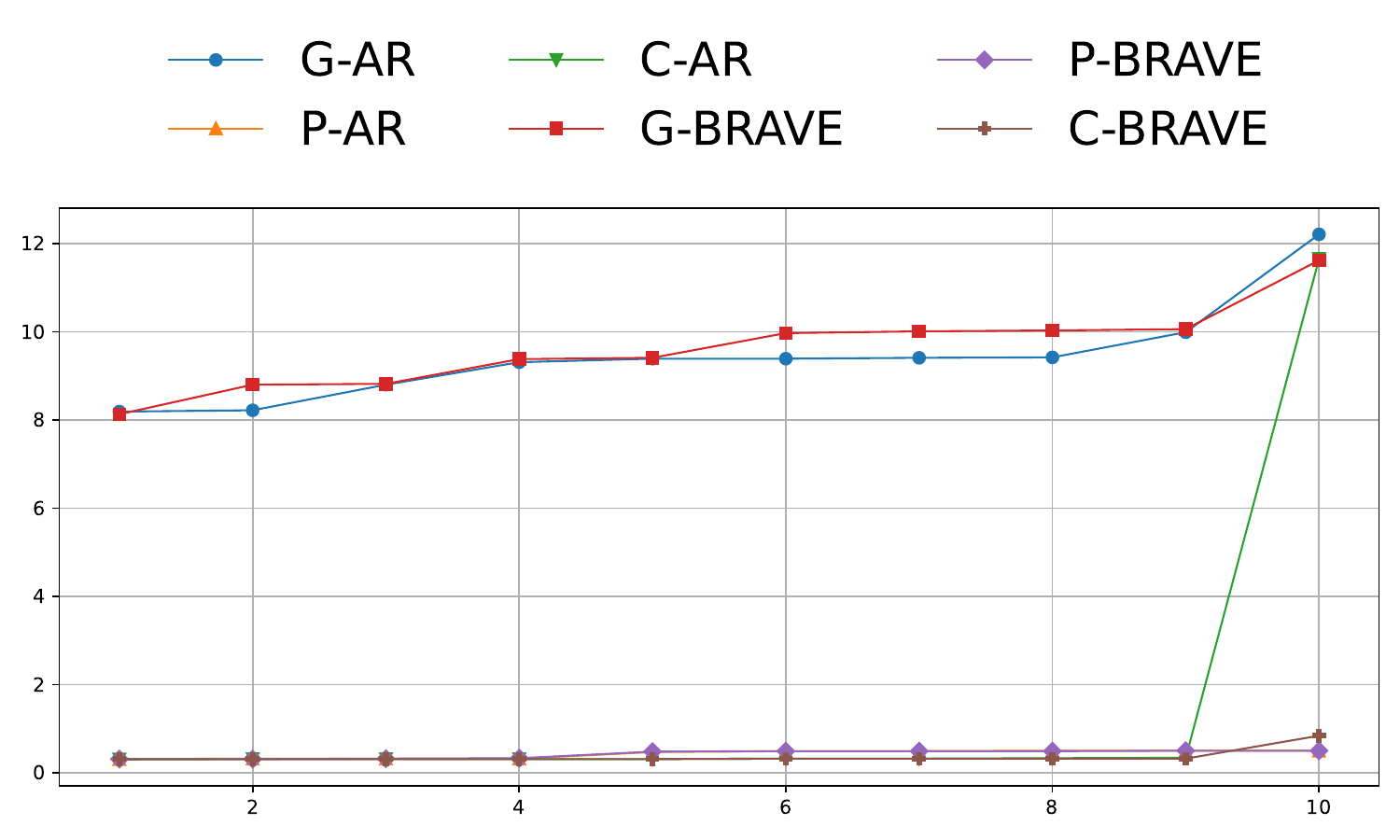}
        \caption{\mn{u5c1} $\succ^{ns}$}
        \label{fig:u5c1_non_score}
    \end{subfigure}
    \begin{subfigure}{0.24\textwidth}
        \centering
        \includegraphics[width=\linewidth]{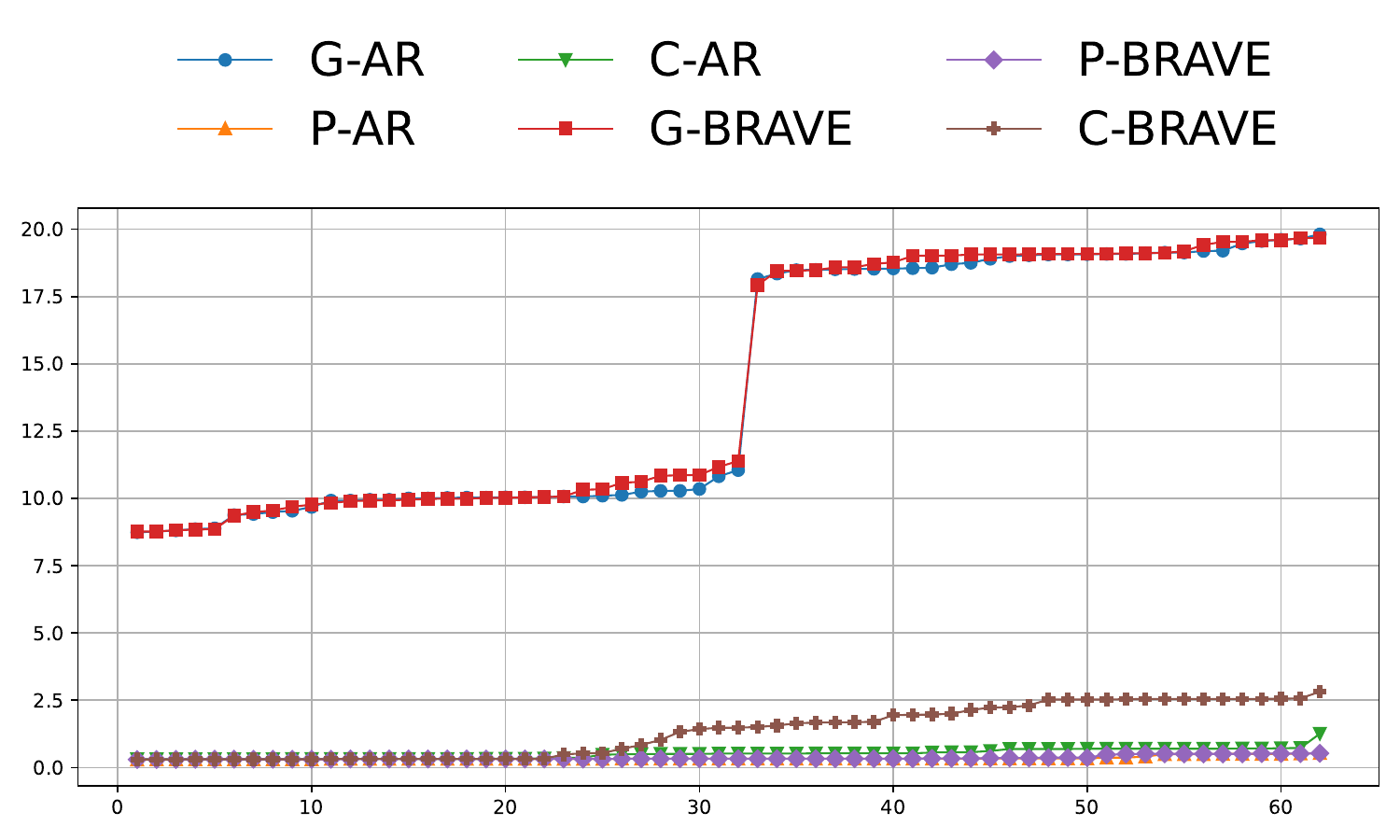}
        \caption{\mn{u5c1} $\succ^{ss}$}
        \label{fig:u5c1_score}
    \end{subfigure}
    \begin{subfigure}{0.24\textwidth}
        \centering
        \includegraphics[width=\linewidth]{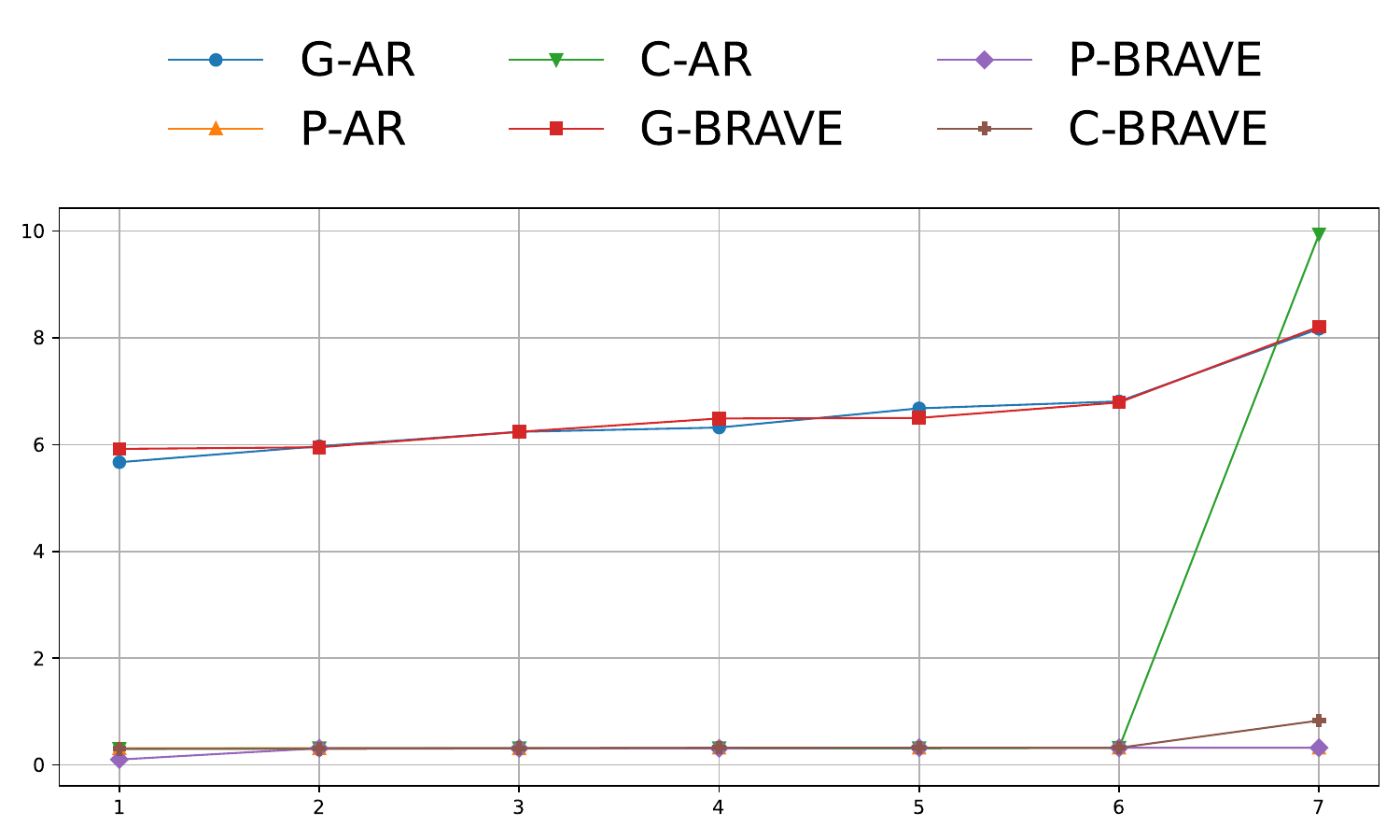}
        \caption{\mn{u1c5} $\succ^{ns}$}
        \label{fig:u1c5_non_score}
    \end{subfigure}
    \begin{subfigure}{0.24\textwidth}
        \centering
        \includegraphics[width=\linewidth]{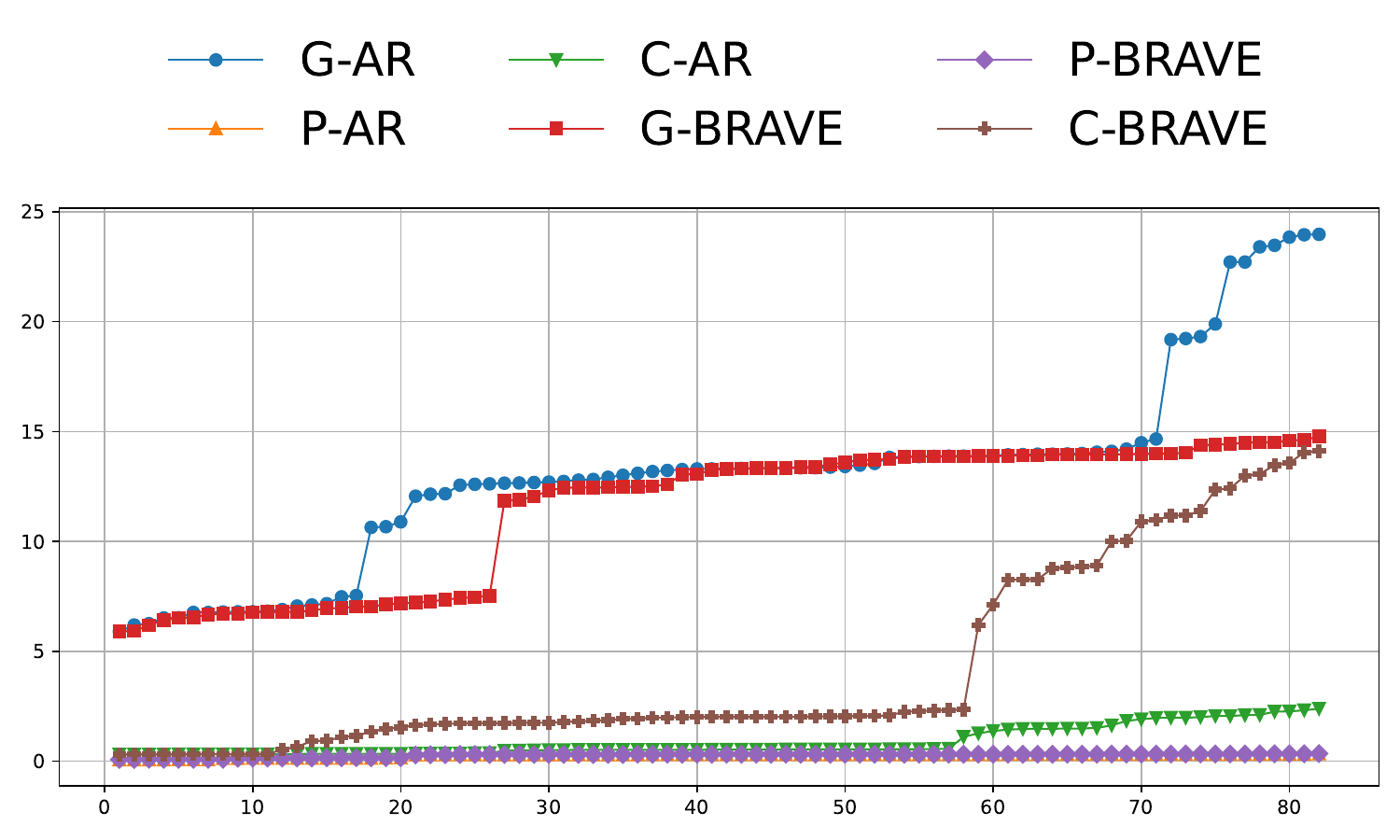}
        \caption{\mn{u1c5} $\succ^{ss}$}
        \label{fig:u1c5_score}
    \end{subfigure}
    \begin{subfigure}{0.24\textwidth}
        \centering
        \includegraphics[width=\linewidth]{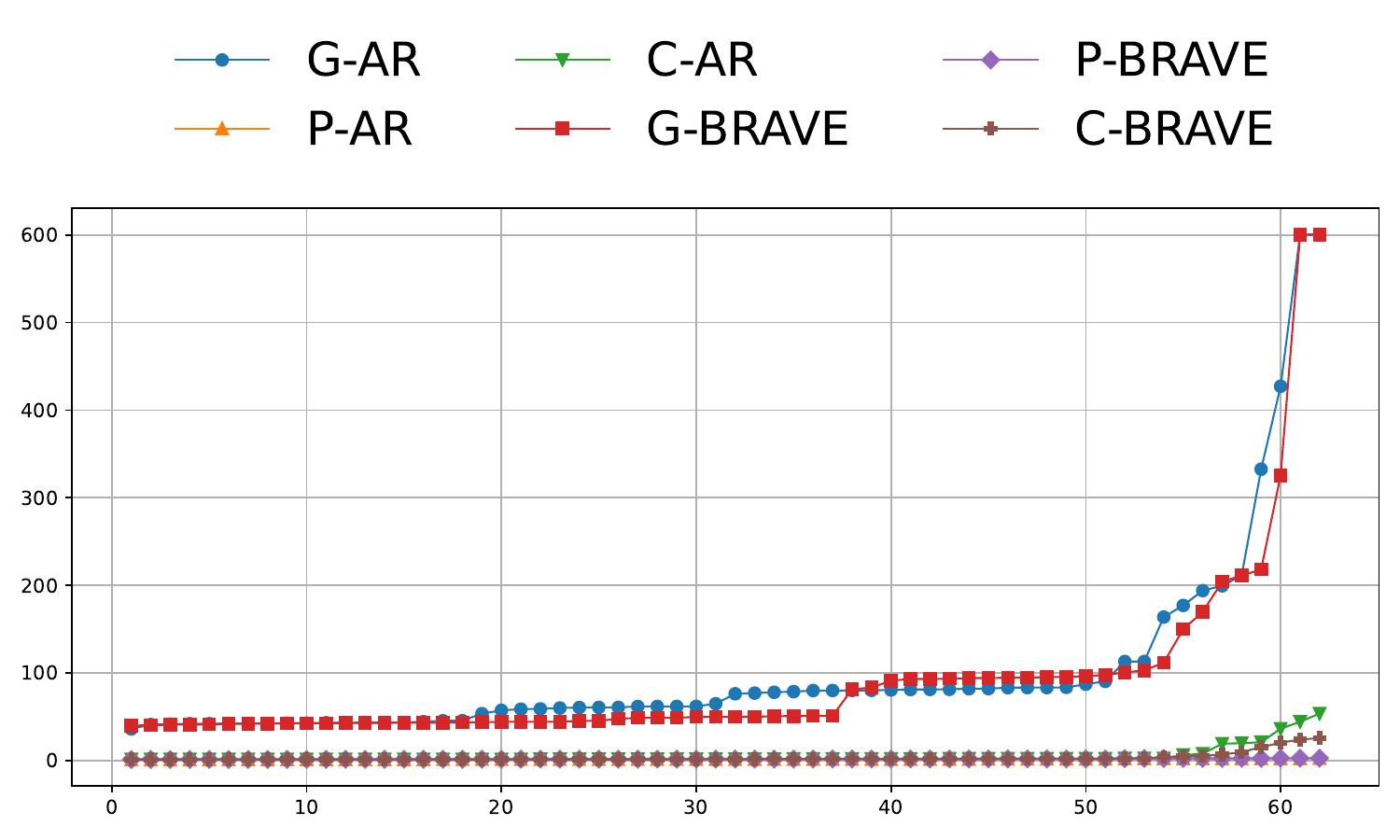}
        \caption{\mn{u5c5} $\succ^{ns}$}
        \label{fig:u5c5_non_score}
    \end{subfigure}
    \begin{subfigure}{0.24\textwidth}
        \centering
        \includegraphics[width=\linewidth]{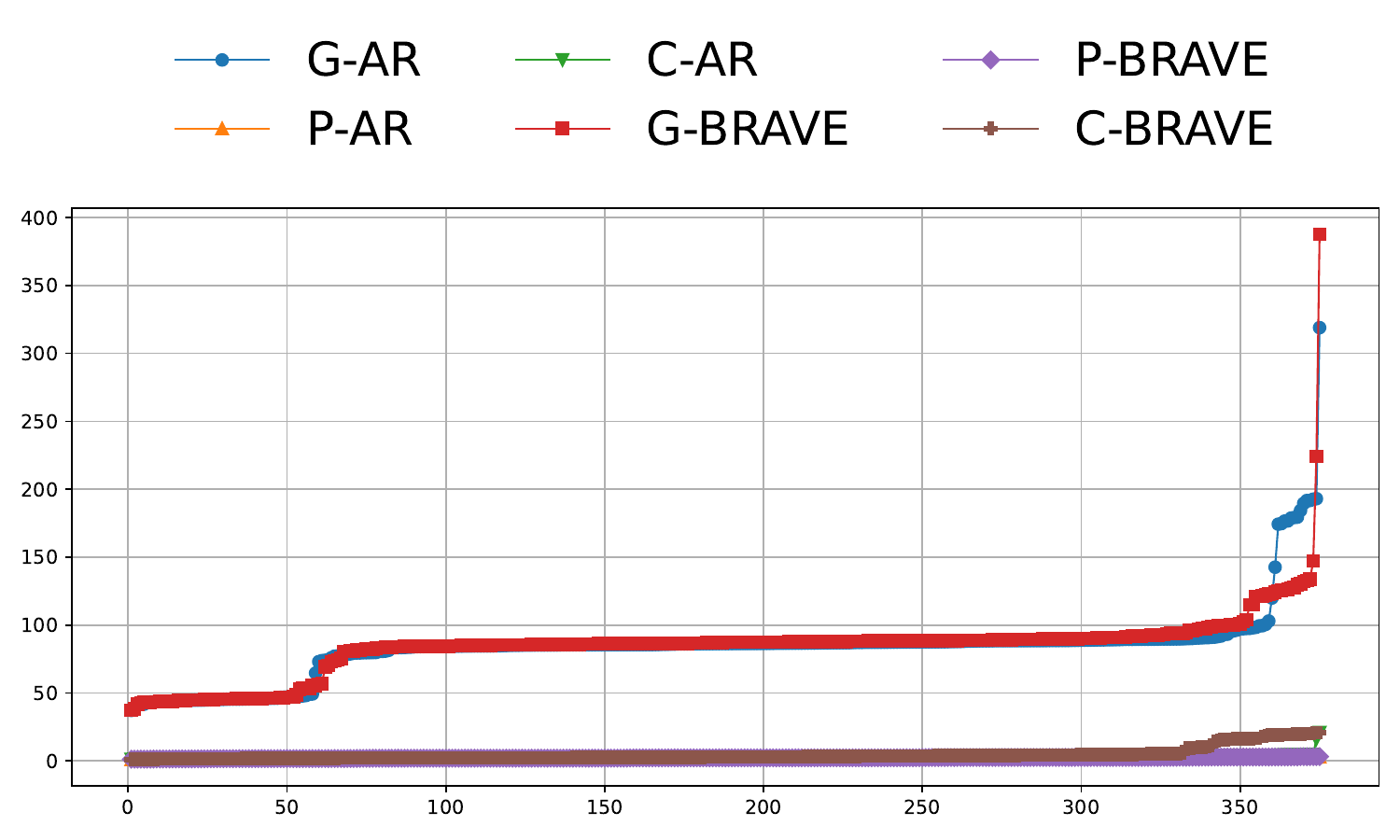}
        \caption{\mn{u5c5} $\succ^{ss}$}
        \label{fig:u5c5_score}
    \end{subfigure}
    \begin{subfigure}{0.24\textwidth}
        \centering
        \includegraphics[width=\linewidth]{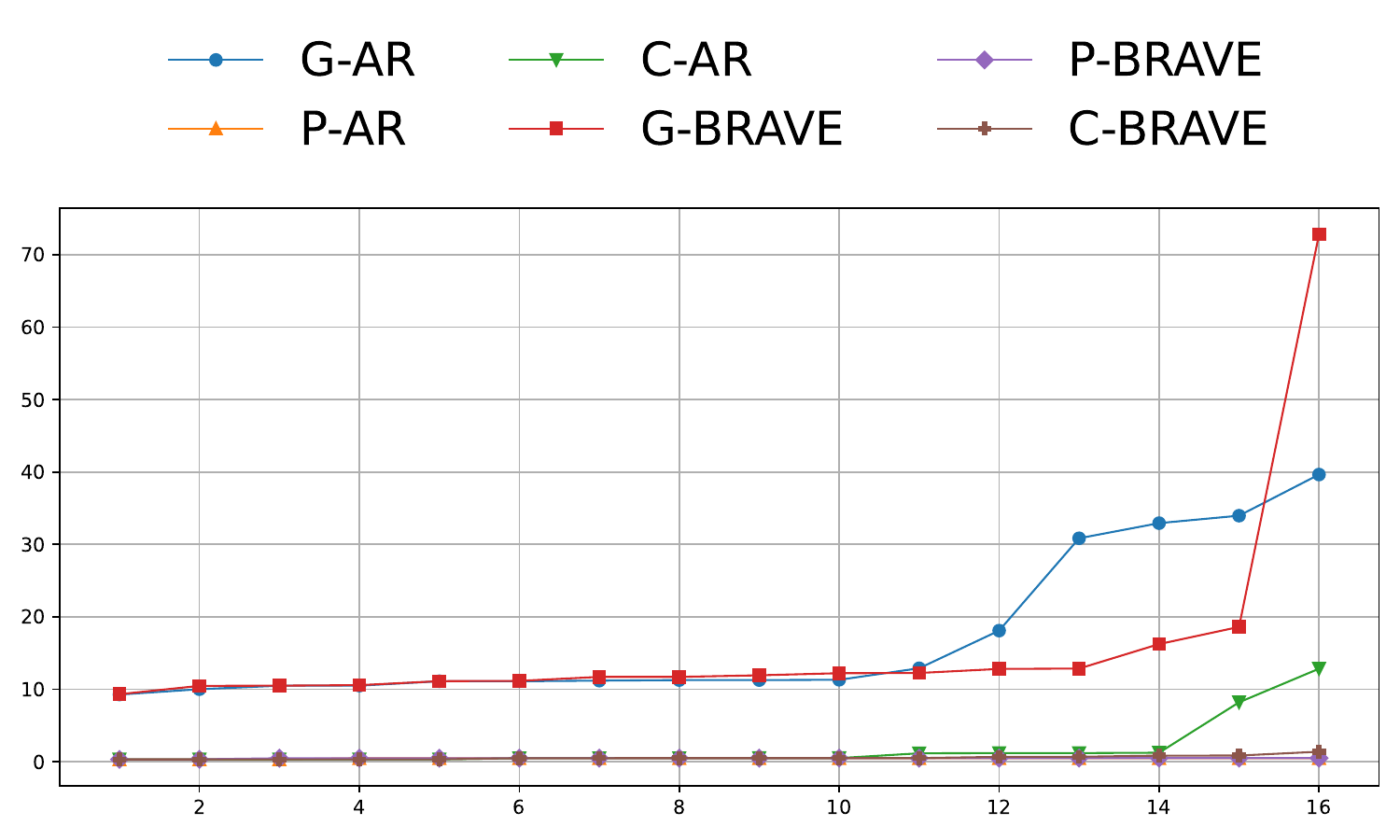}
        \caption{\mn{u1c10} $\succ^{ns}$}
        \label{fig:u1c10_non_score}
    \end{subfigure}
    \begin{subfigure}{0.24\textwidth}
        \centering
        \includegraphics[width=\linewidth]{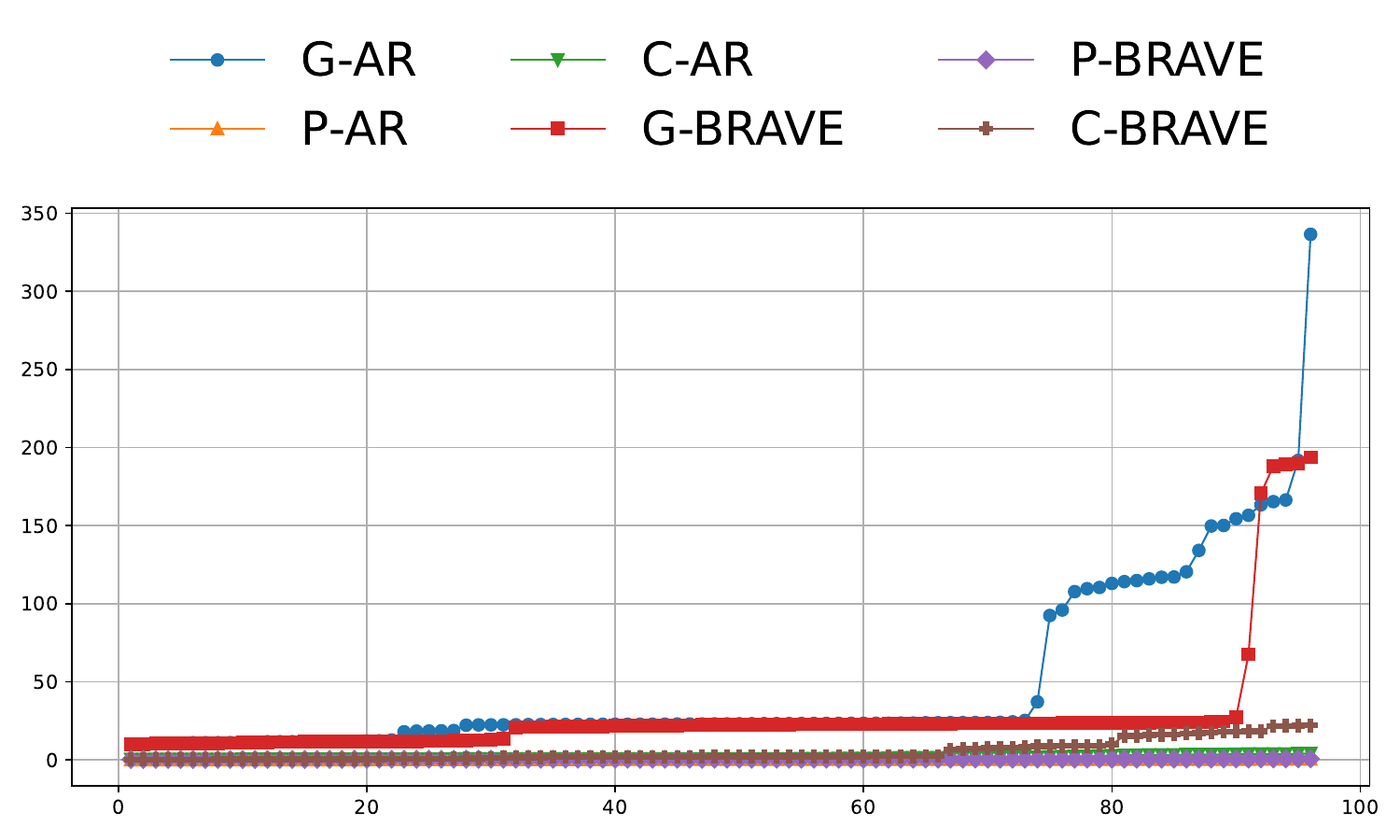}
        \caption{\mn{u1c10} $\succ^{ss}$}
        \label{fig:u1c10_score}
    \end{subfigure}
    \begin{subfigure}{0.24\textwidth}
        \centering
        \includegraphics[width=\linewidth]{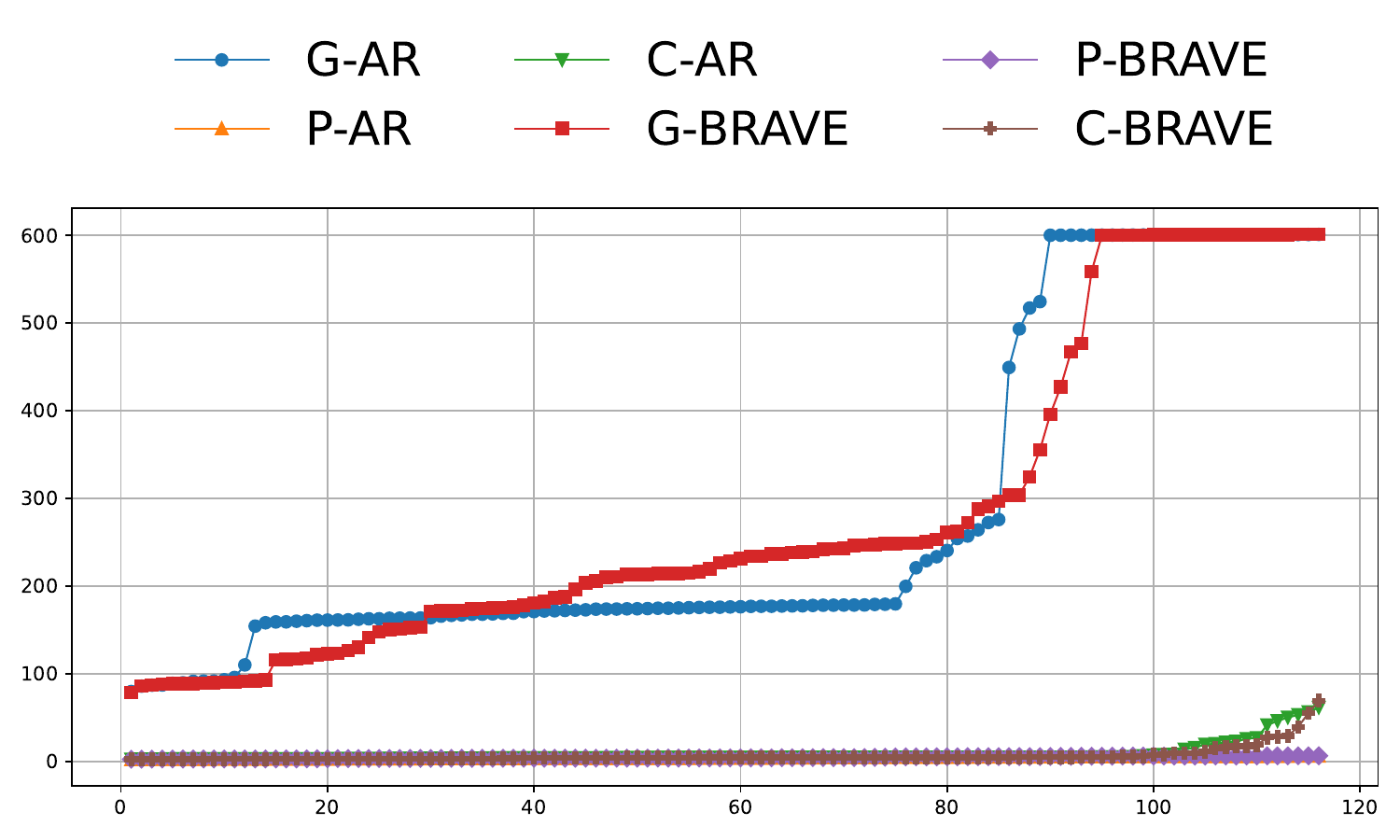}
        \caption{u5c10 $\succ^{ns}$}
        \label{fig:u5c10_non_score}
    \end{subfigure}
    \begin{subfigure}{0.24\textwidth}
        \centering
        \includegraphics[width=\linewidth]{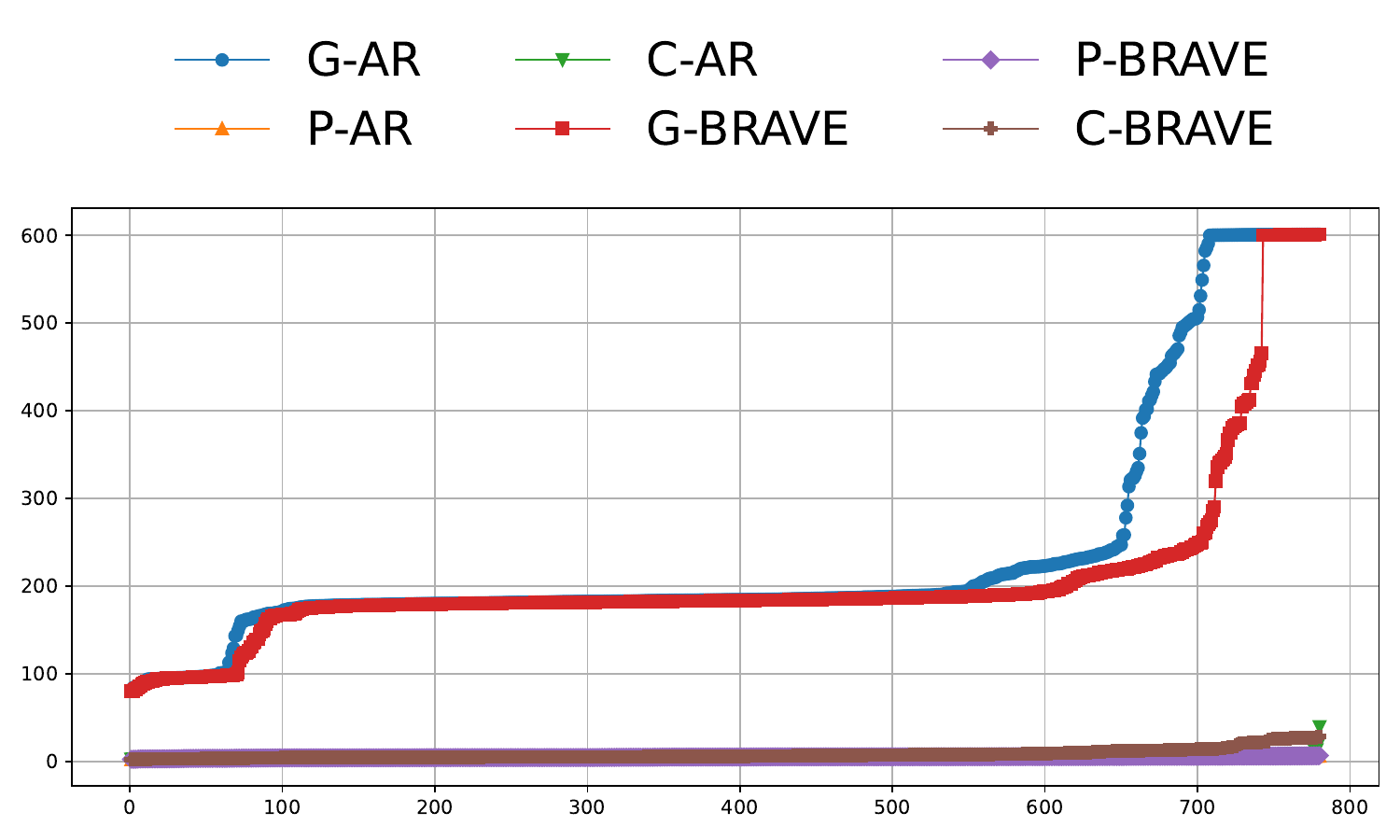}
        \caption{u5c10 $\succ^{ss}$}
        \label{fig:u5c10_score}
    \end{subfigure}
    \begin{subfigure}{0.24\textwidth}
        \centering
        \includegraphics[width=\linewidth]{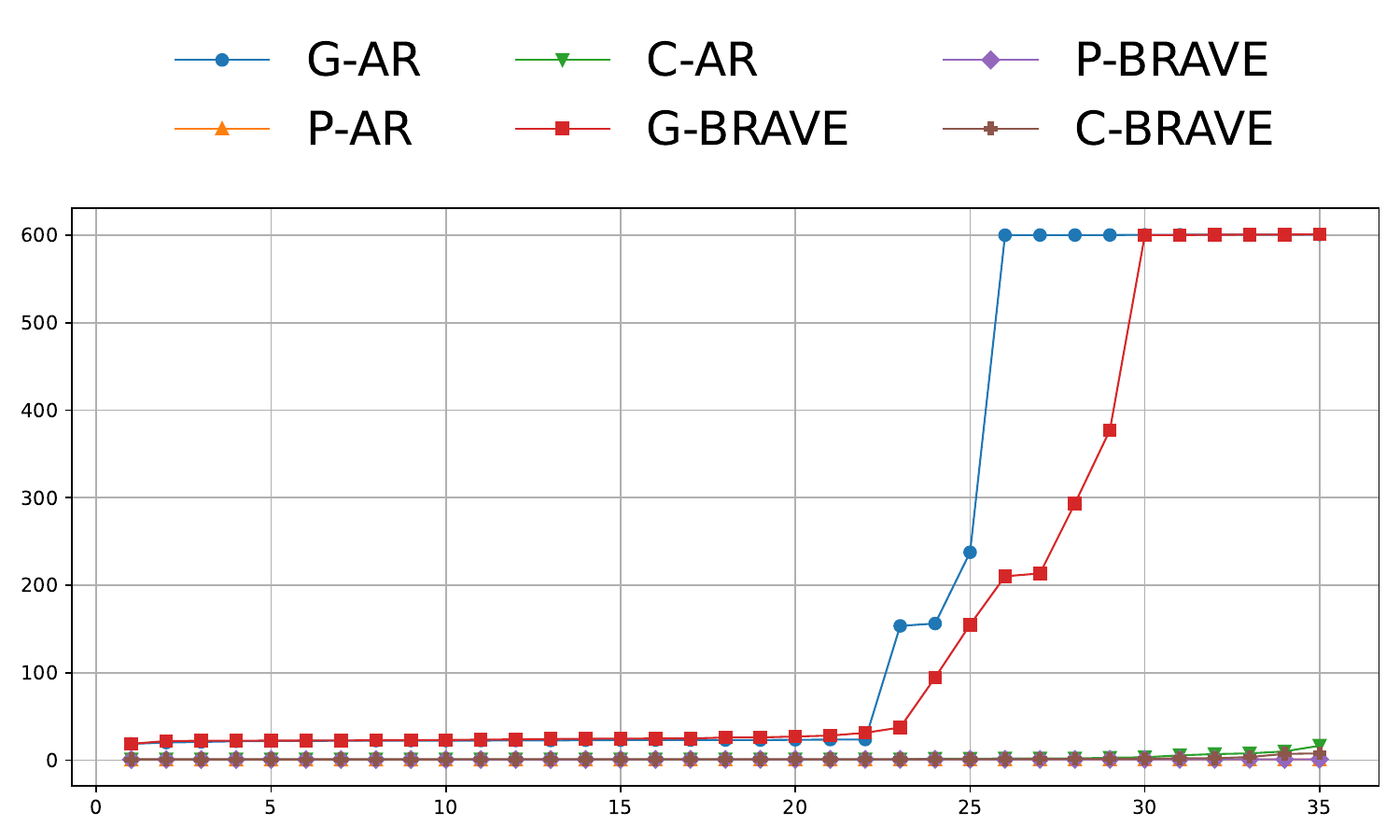}
        \caption{\mn{u1c20} $\succ^{ns}$}
        \label{fig:u1c20_non_score}
    \end{subfigure}
    \begin{subfigure}{0.24\textwidth}
        \centering
        \includegraphics[width=\linewidth]{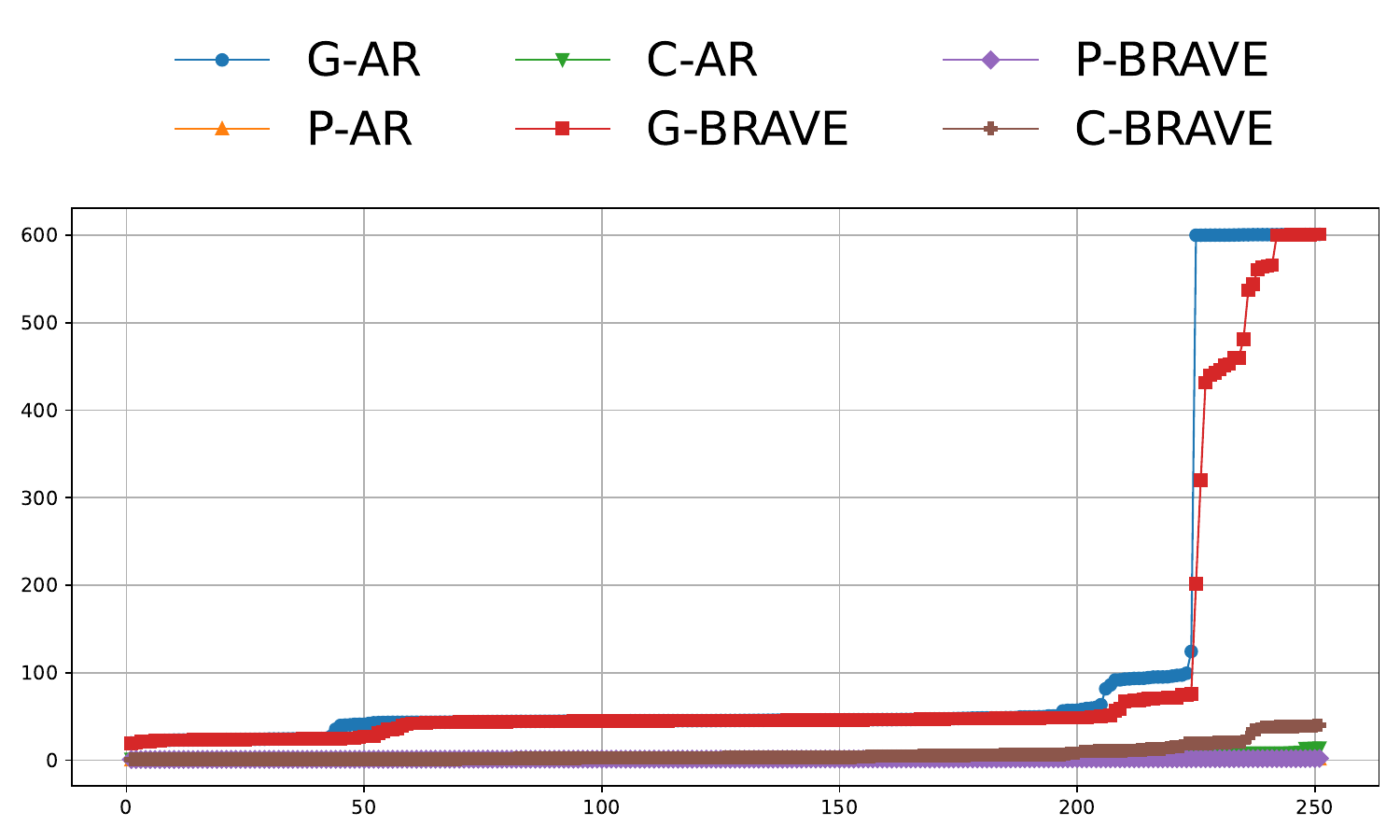}
        \caption{\mn{u1c20} $\succ^{ss}$}
        \label{fig:u1c20_score}
    \end{subfigure}
    \begin{subfigure}{0.24\textwidth}
        \centering
        \includegraphics[width=\linewidth]{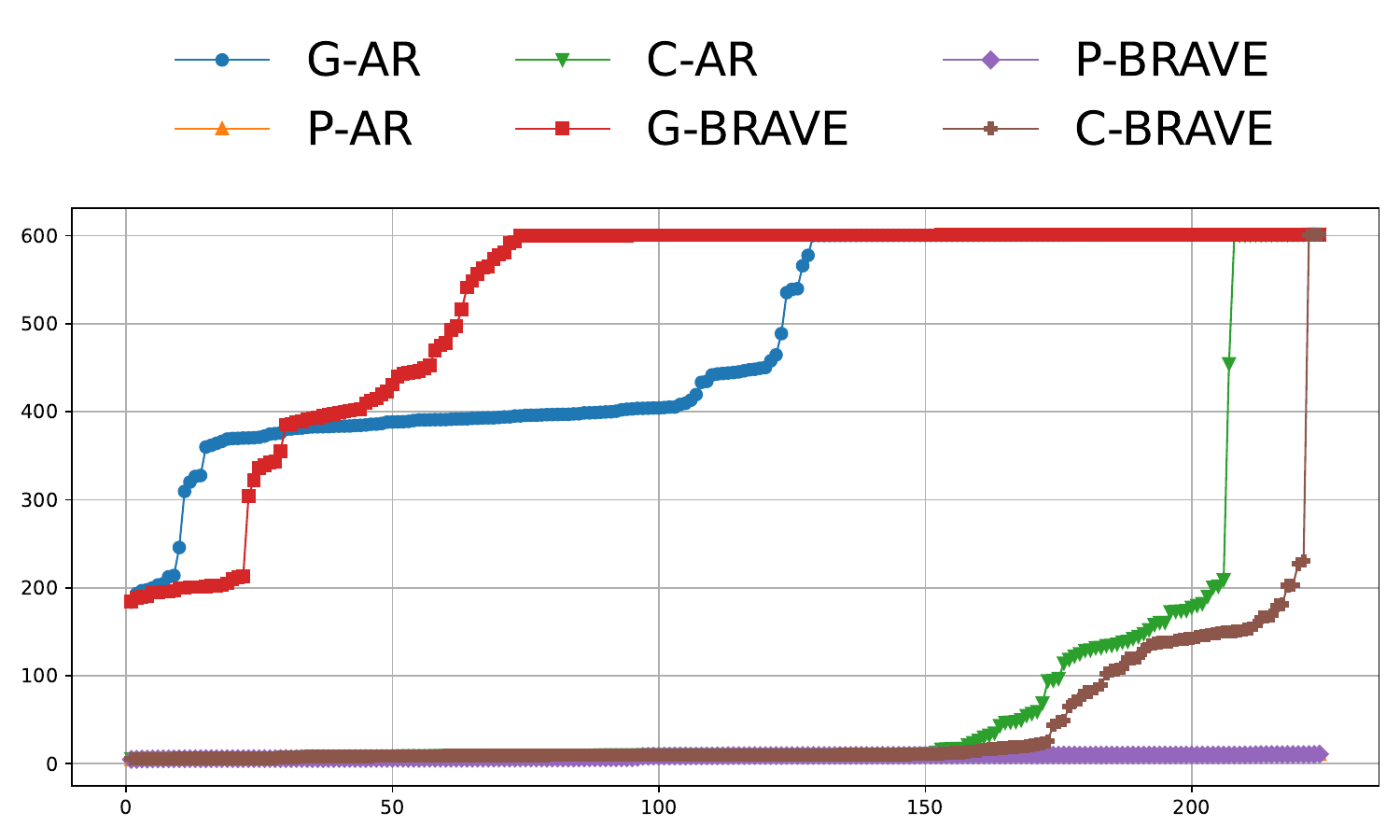}
        \caption{u5c20 $\succ^{ns}$}
        \label{fig:u5c20_non_score}
    \end{subfigure}
    \begin{subfigure}{0.24\textwidth}
        \centering
        \includegraphics[width=\linewidth]{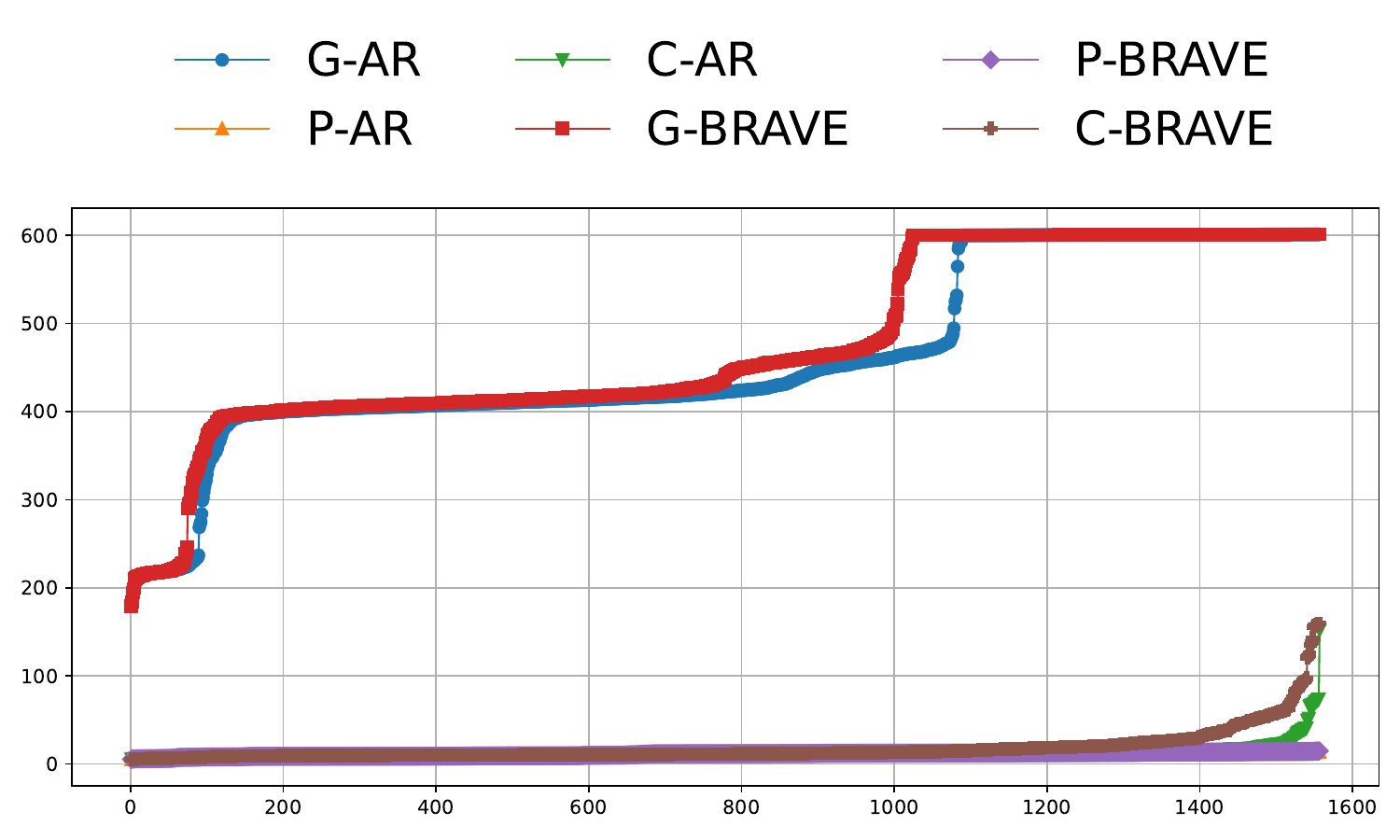}
        \caption{u5c20 $\succ^{ss}$}
        \label{fig:u5c20_score}
    \end{subfigure}
    \begin{subfigure}{0.24\textwidth}
        \centering
        \includegraphics[width=\linewidth]{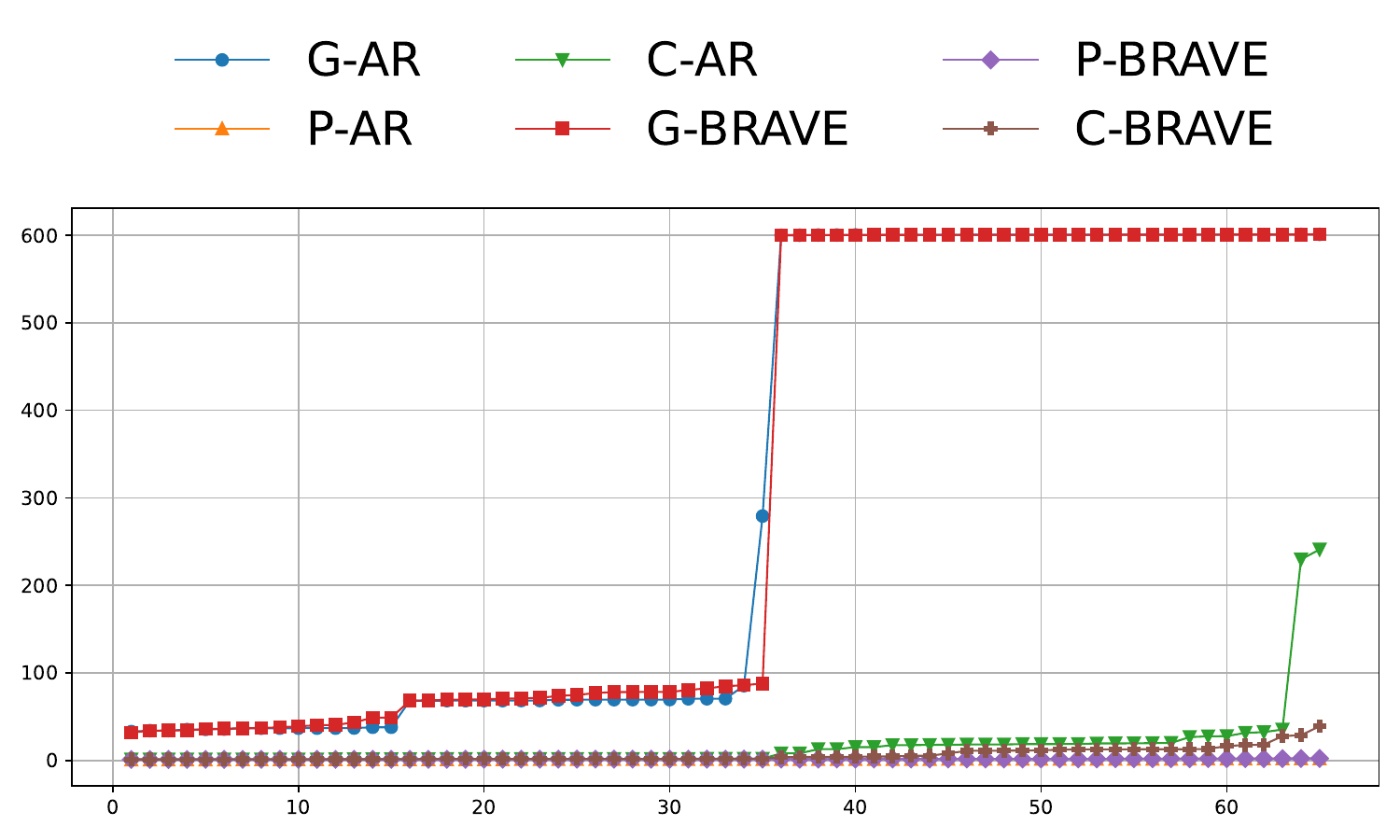}
        \caption{\mn{u1c30} $\succ^{ns}$}
        \label{fig:u1c30_non_score}
    \end{subfigure}
    \begin{subfigure}{0.24\textwidth}
        \centering
        \includegraphics[width=\linewidth]{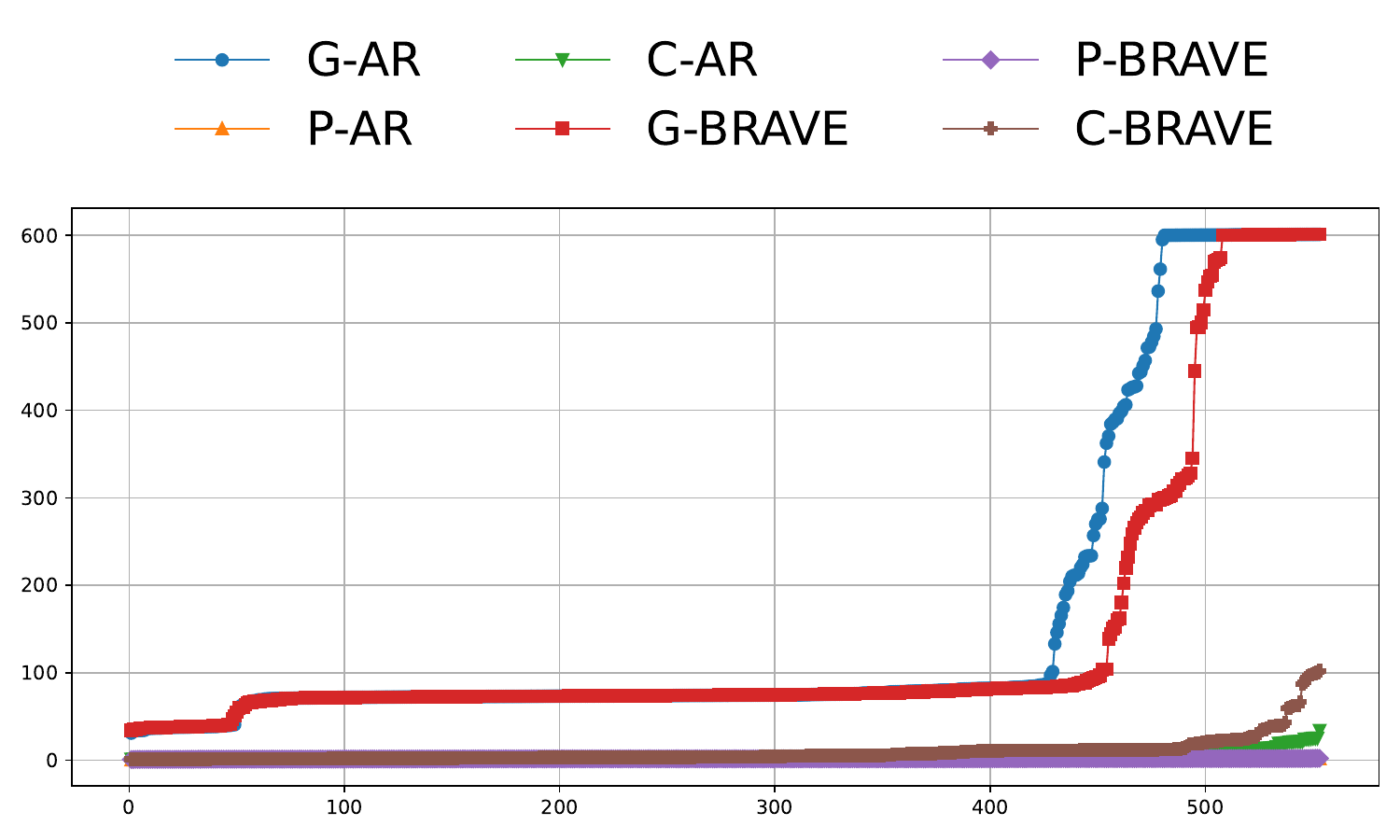}
        \caption{\mn{u1c30} $\succ^{ss}$}
        \label{fig:u1c30_score}
    \end{subfigure}
    \begin{subfigure}{0.24\textwidth}
        \centering
        \includegraphics[width=\linewidth]{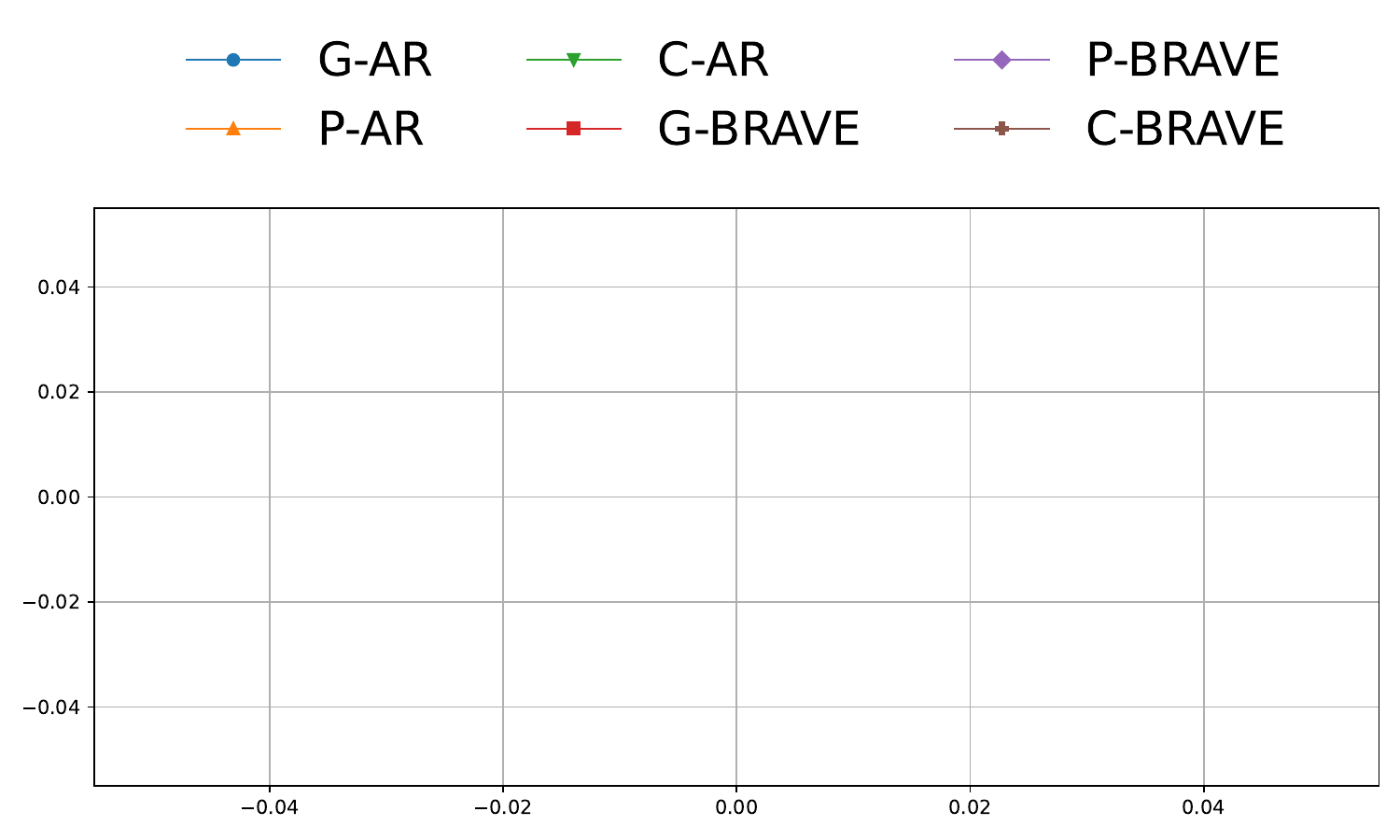}
        \caption{u5c30 $\succ^{ns}$}
        \label{fig:u5c30_non_score}
    \end{subfigure}
    \begin{subfigure}{0.24\textwidth}
        \centering
        \includegraphics[width=\linewidth]{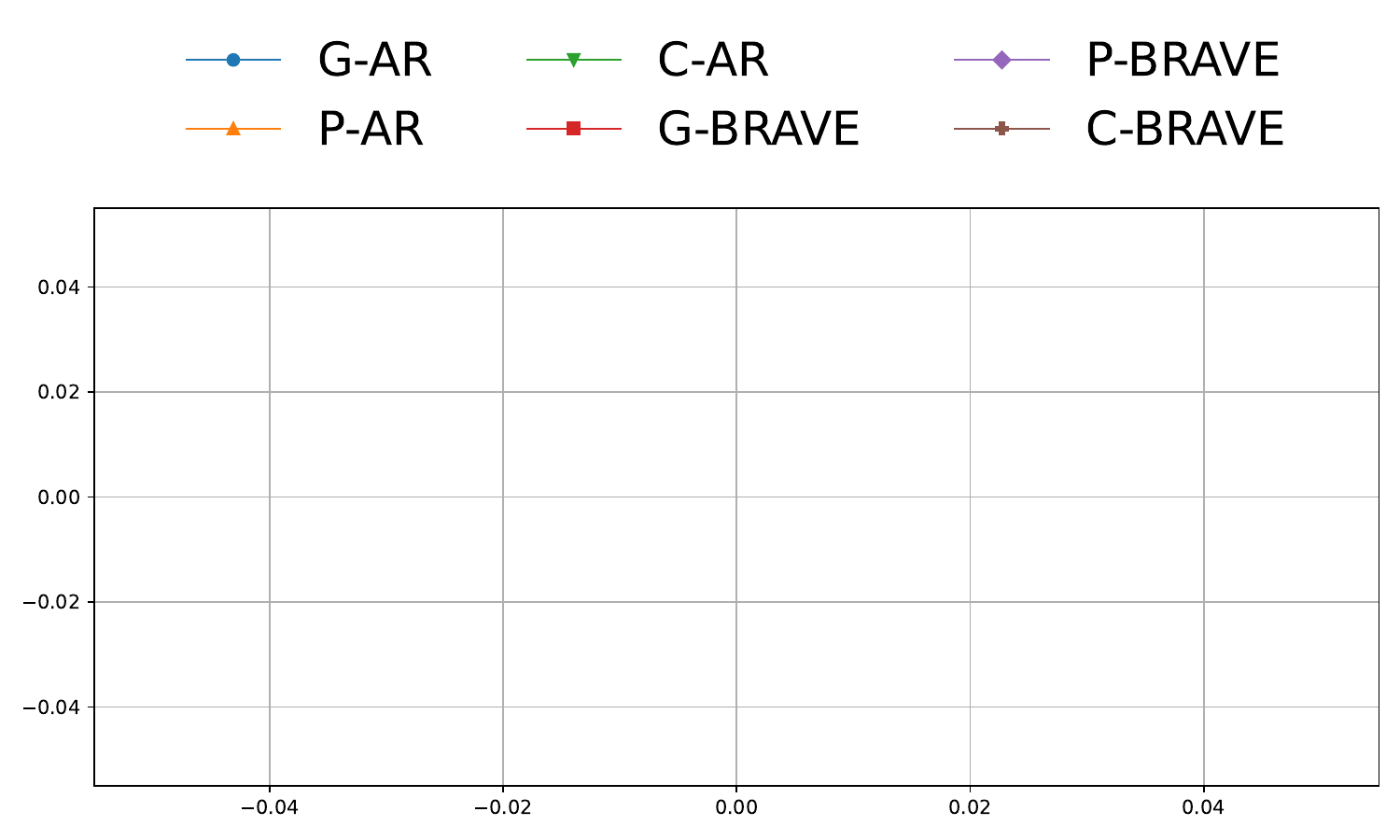}
        \caption{u5c30 $\succ^{ss}$}
        \label{fig:u5c30_score}
    \end{subfigure}
    \begin{subfigure}{0.24\textwidth}
        \centering
        \includegraphics[width=\linewidth]{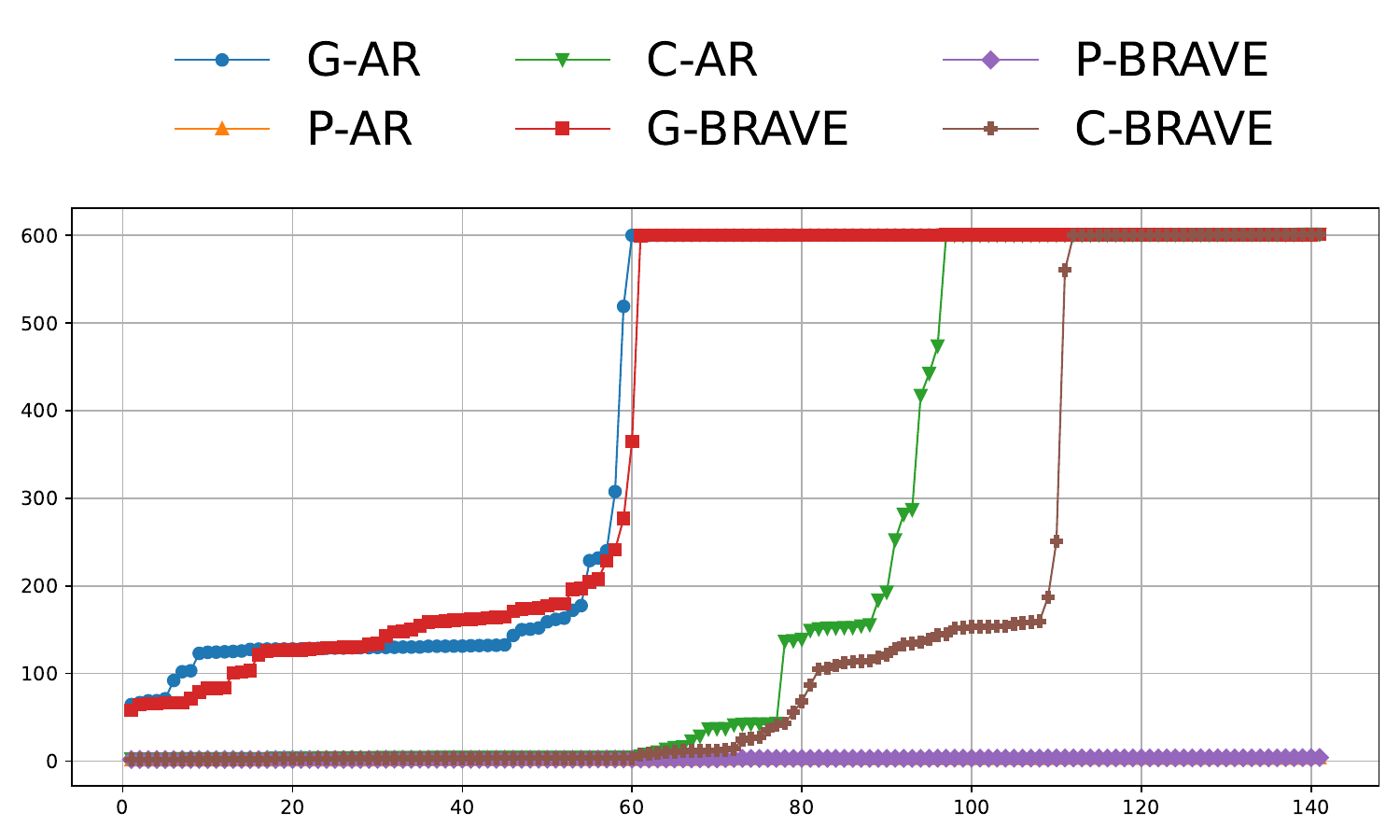}
        \caption{\mn{u1c50} $\succ^{ns}$}
        \label{fig:u1c50_non_score}
    \end{subfigure}
    \begin{subfigure}{0.24\textwidth}
        \centering
        \includegraphics[width=\linewidth]{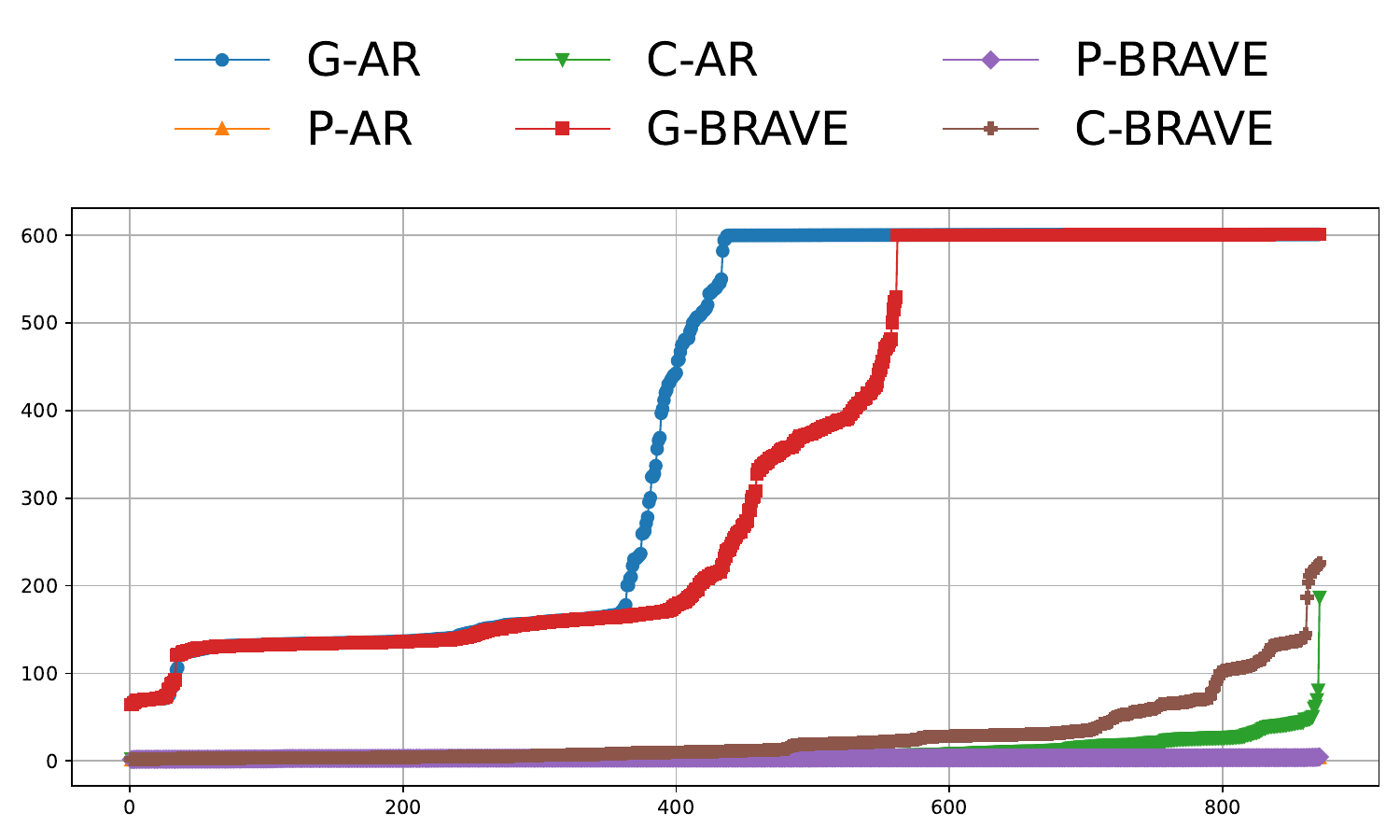}
        \caption{\mn{u1c50} $\succ^{ss}$}
        \label{fig:u1c50_score}
    \end{subfigure}
    \begin{subfigure}{0.24\textwidth}
        \centering
        \includegraphics[width=\linewidth]{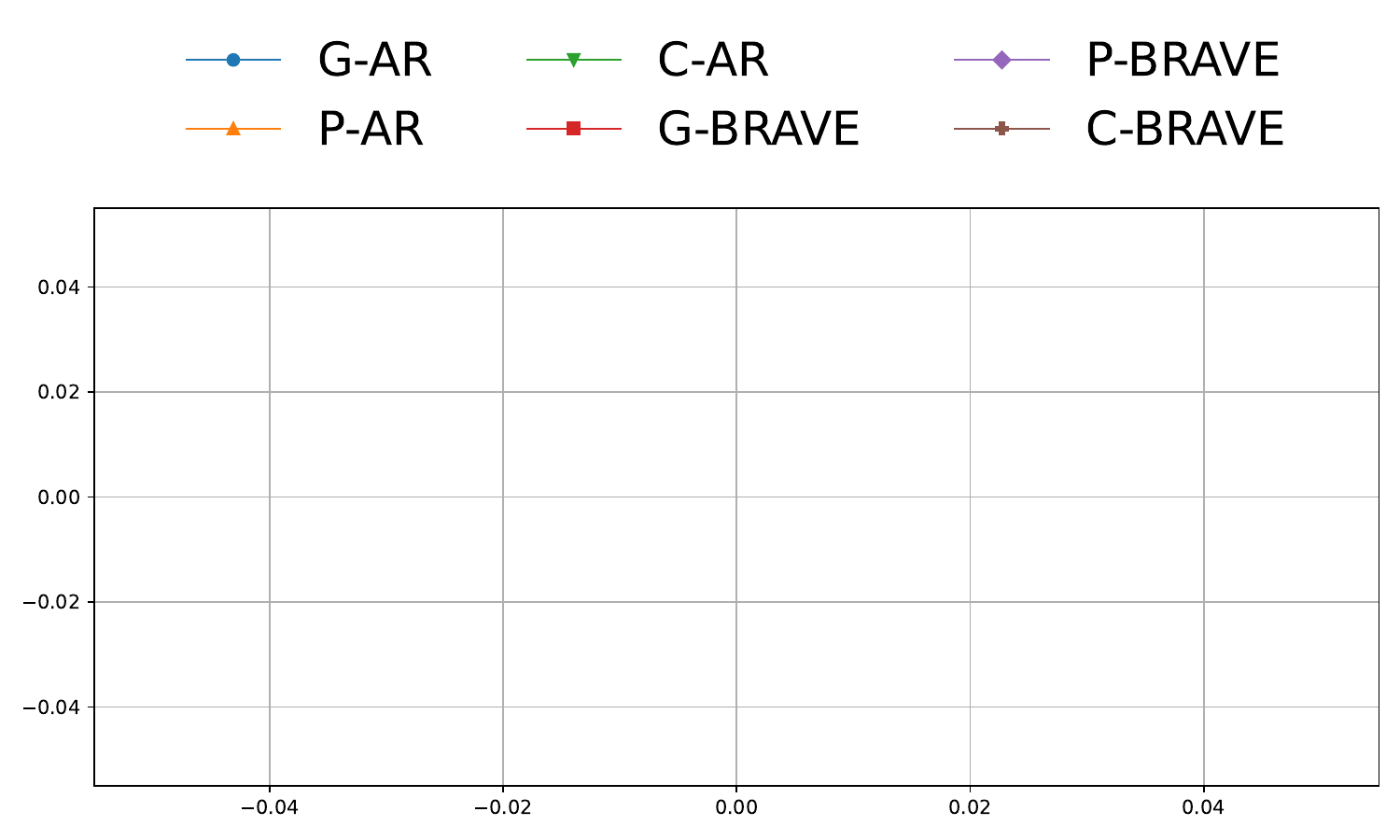}
        \caption{u5c50 $\succ^{ns}$}
        \label{fig:u1c50_non_score}
    \end{subfigure}
    \begin{subfigure}{0.24\textwidth}
        \centering
        \includegraphics[width=\linewidth]{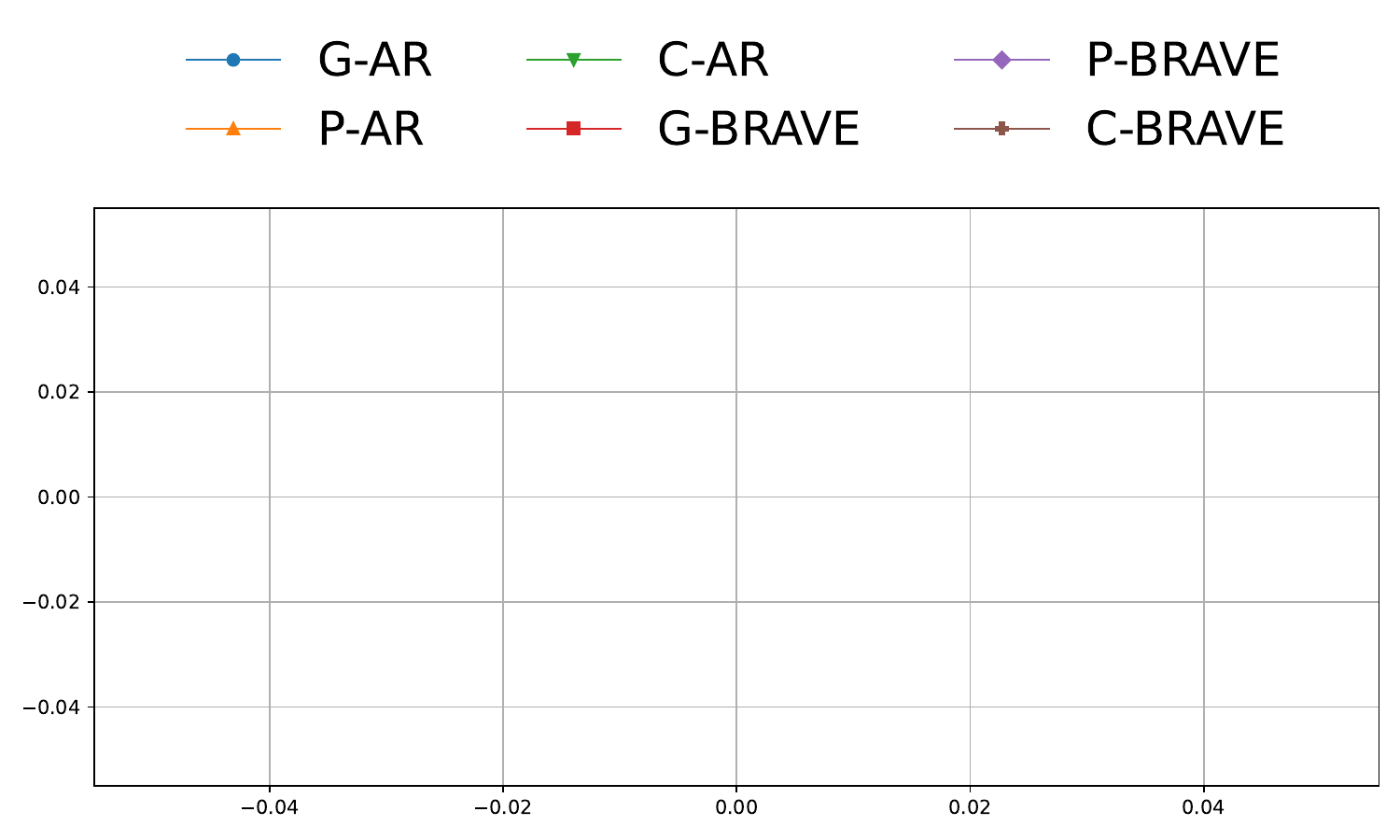}
        \caption{u5c50 $\succ^{ss}$}
        \label{fig:u5c50_score}
    \end{subfigure}
    \caption{Time (in seconds) to decide for each $q(\ans)$ whether $q(\ans)$ holds under X-AR/X-brave, with instances sorted by increasing solving time 
along the x-axis (\ie a point $(i,j)$ indicates that $i$ instances could be solved within $j$ seconds). Empty plots: in the case of \mn{u1c1} $\succ^{ss}$, all potential answers are grounded so there was no potential answer left to check; in the other cases (s,t,w,x), we did not run these experiments as they would have taken too long (lots of answers to check and lots of time-out).}
    \label{fig:overall-runtime-appendix}
\end{figure*}

\begin{figure*}
    \begin{subfigure}{0.24\textwidth}
        \centering
        \includegraphics[width=\linewidth]{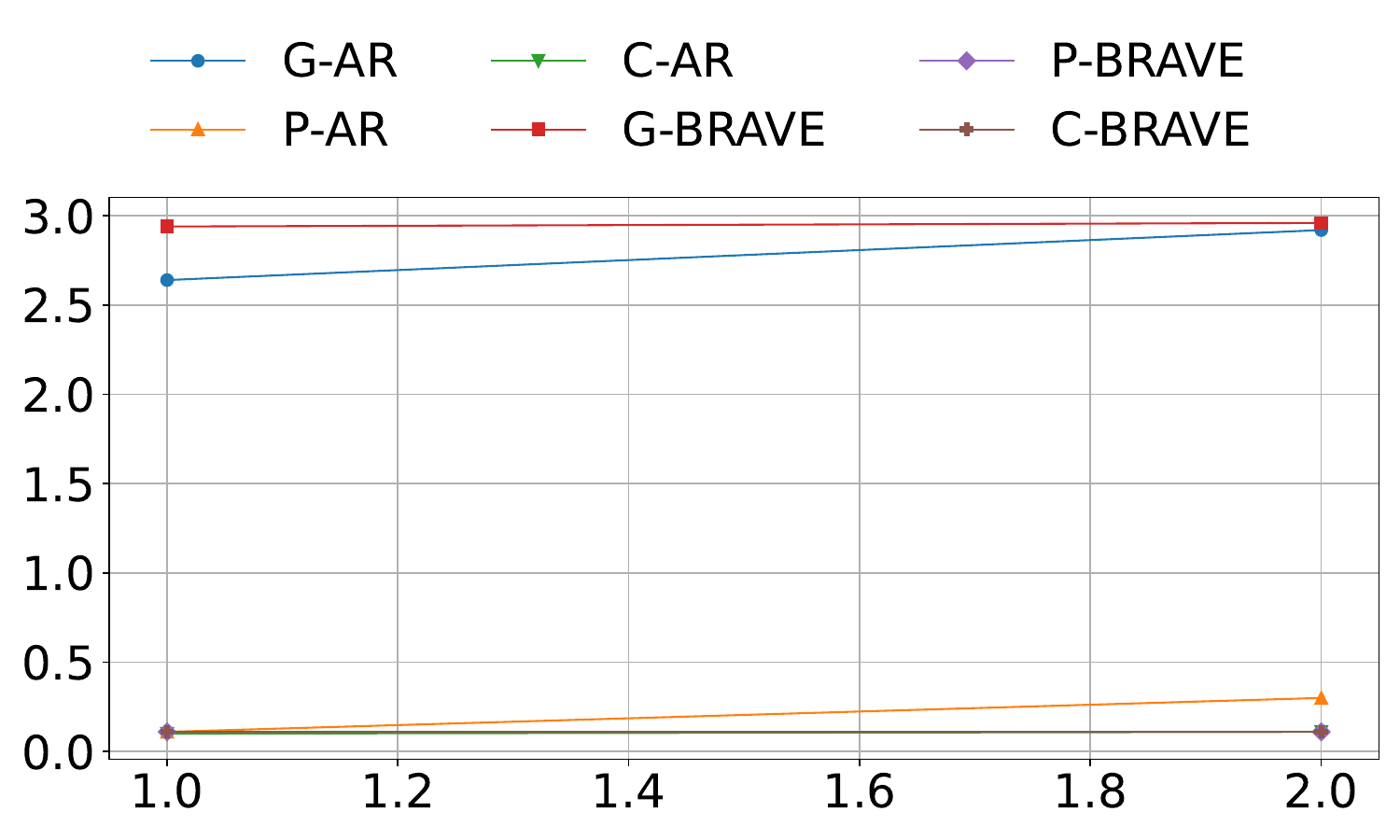}
        \caption{\mn{u1c1} $\succ^{nb}$}
        \label{fig:u1conf1}
    \end{subfigure}
    \begin{subfigure}{0.24\textwidth}
        \centering
        \includegraphics[width=\linewidth]{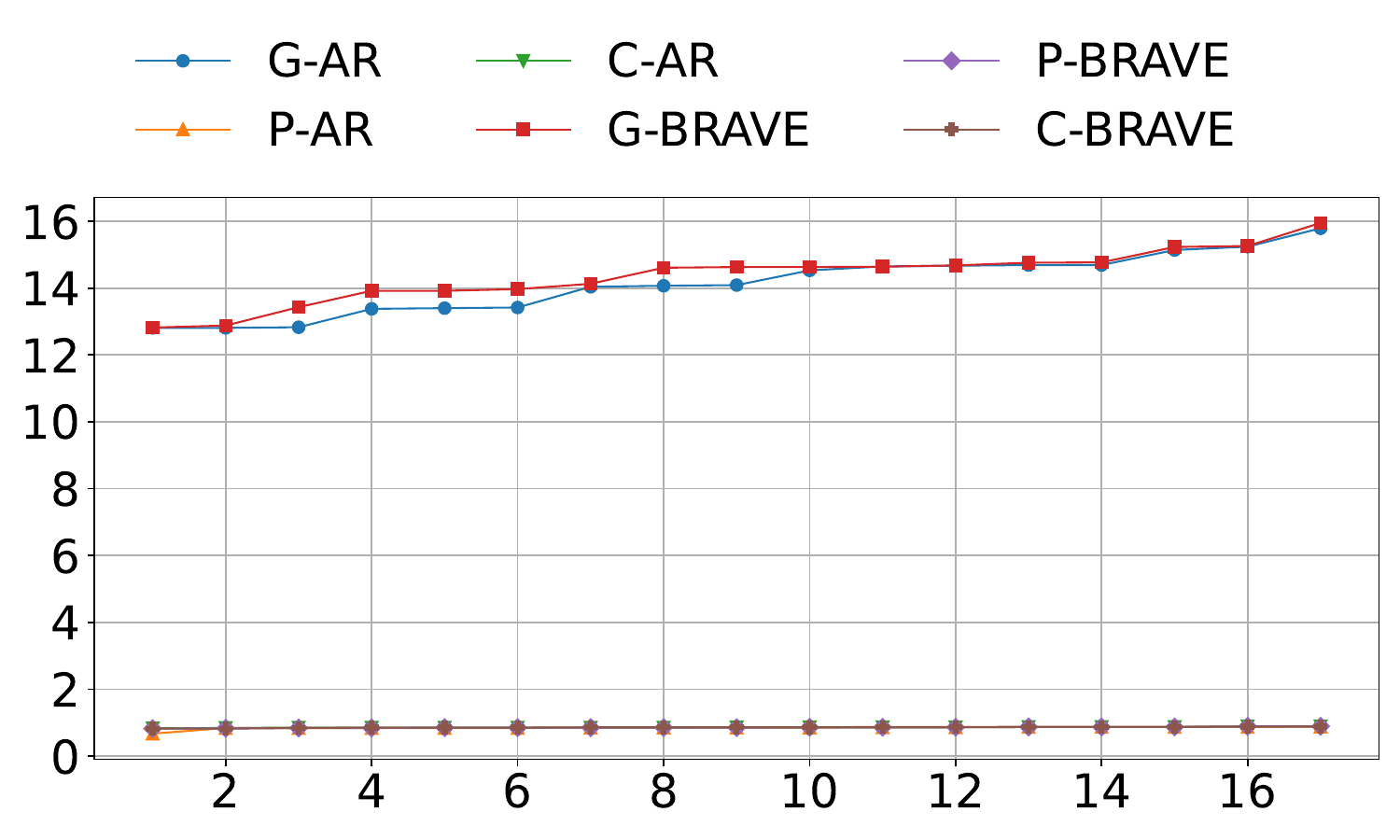}
        \caption{\mn{u5c1} $\succ^{nb}$}
        \label{fig:u5conf1}
    \end{subfigure}
    \\
    \begin{subfigure}{0.24\textwidth}
        \centering
        \includegraphics[width=\linewidth]{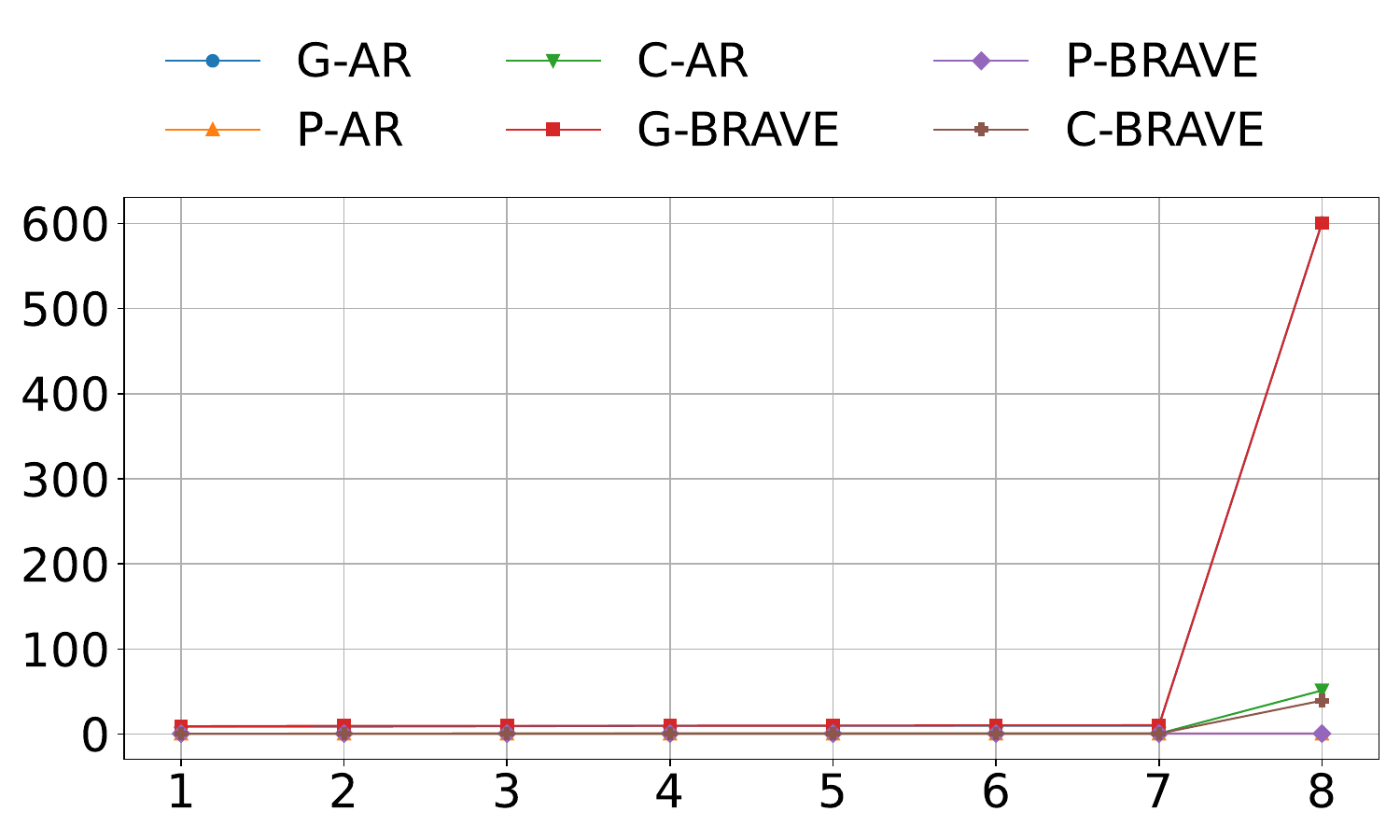}
        \caption{\mn{u1c5} $\succ^{nb}$}
        \label{fig:u1conf5}
    \end{subfigure}
    \begin{subfigure}{0.24\textwidth}
        \centering
        \includegraphics[width=\linewidth]{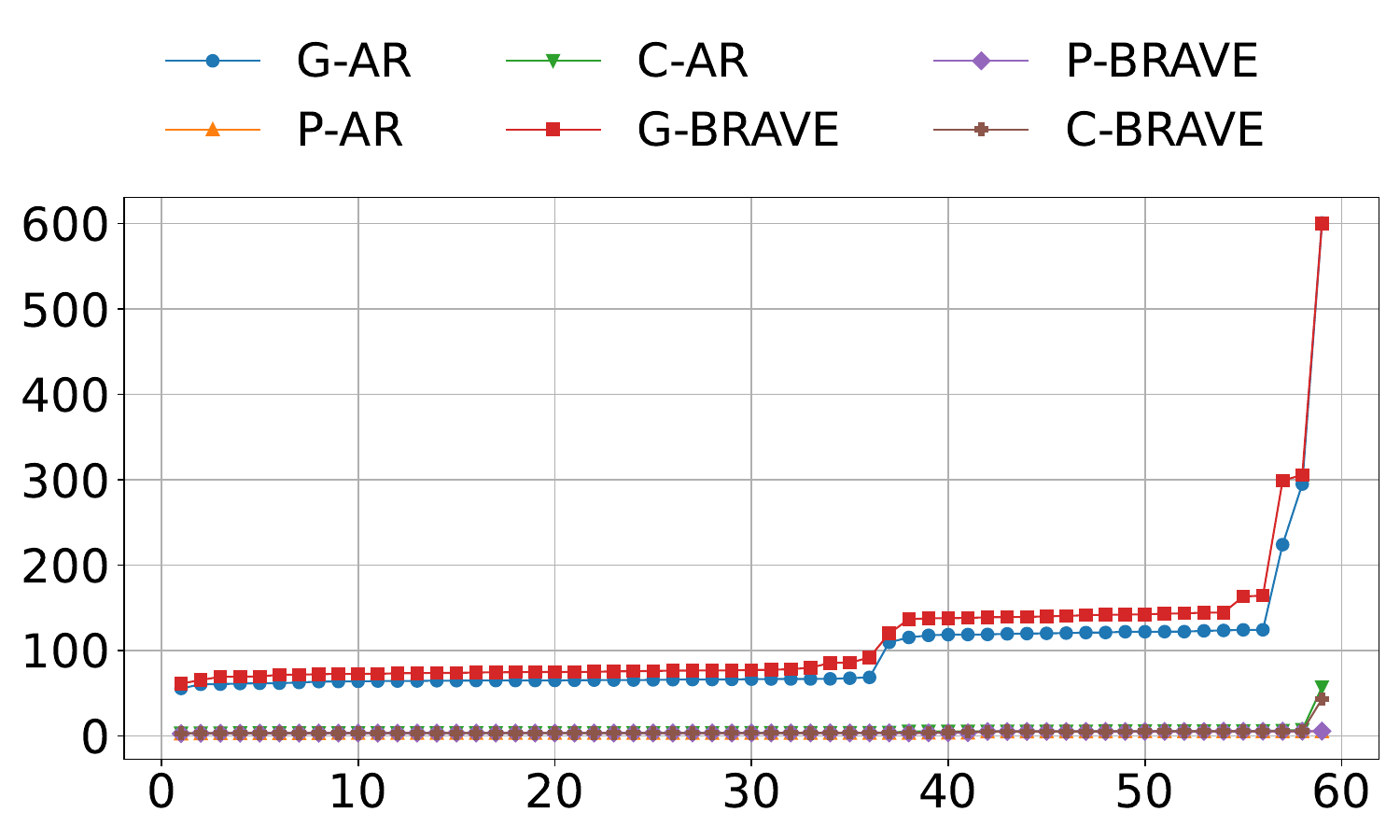}
        \caption{\mn{u5c5} $\succ^{nb}$}
        \label{fig:u5conf5}
    \end{subfigure}
    \\
    \begin{subfigure}{0.24\textwidth}
        \centering
        \includegraphics[width=\linewidth]{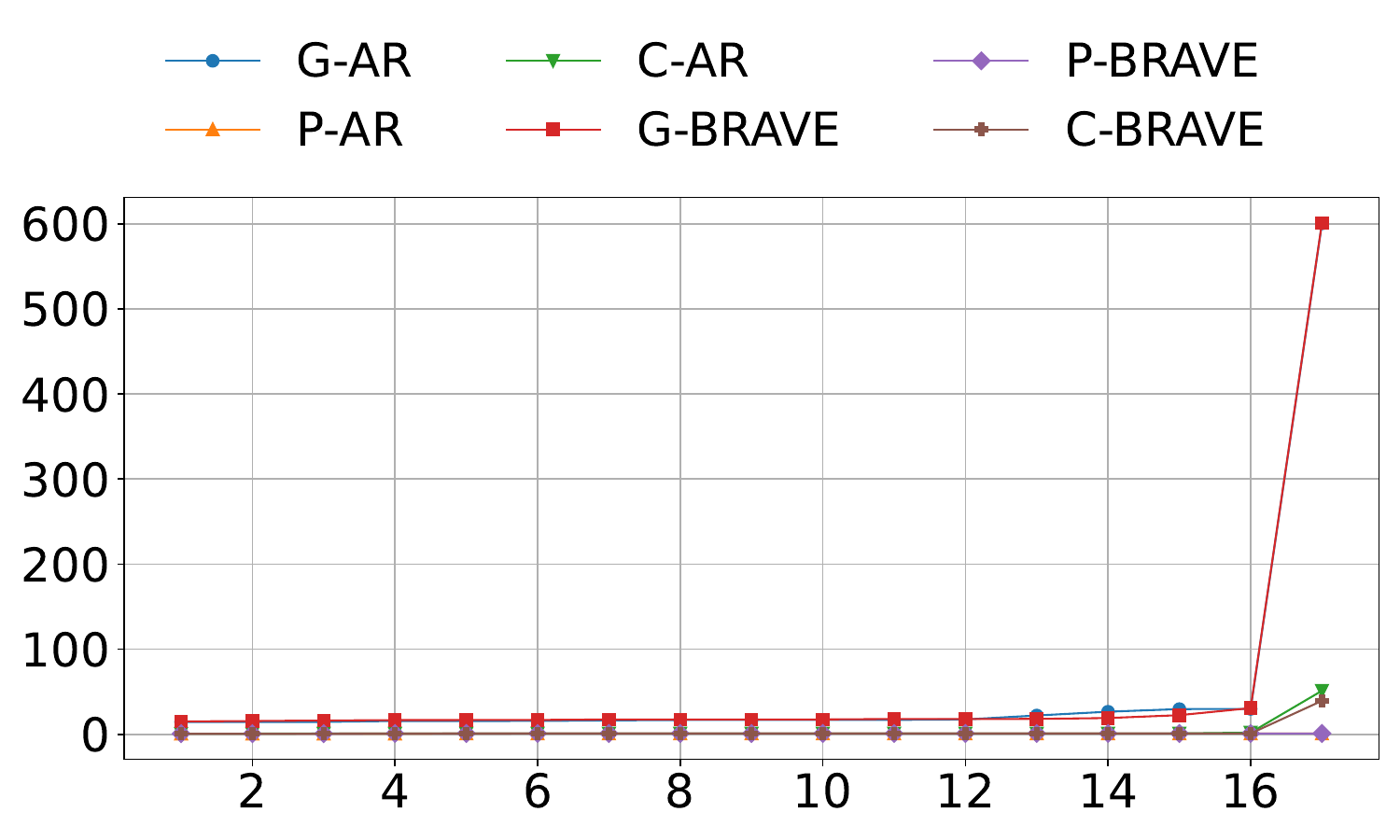}
        \caption{\mn{u1c10} $\succ^{nb}$}
        \label{fig:u1conf10}
    \end{subfigure}
    \\
    \begin{subfigure}{0.24\textwidth}
        \centering
        \includegraphics[width=\linewidth]{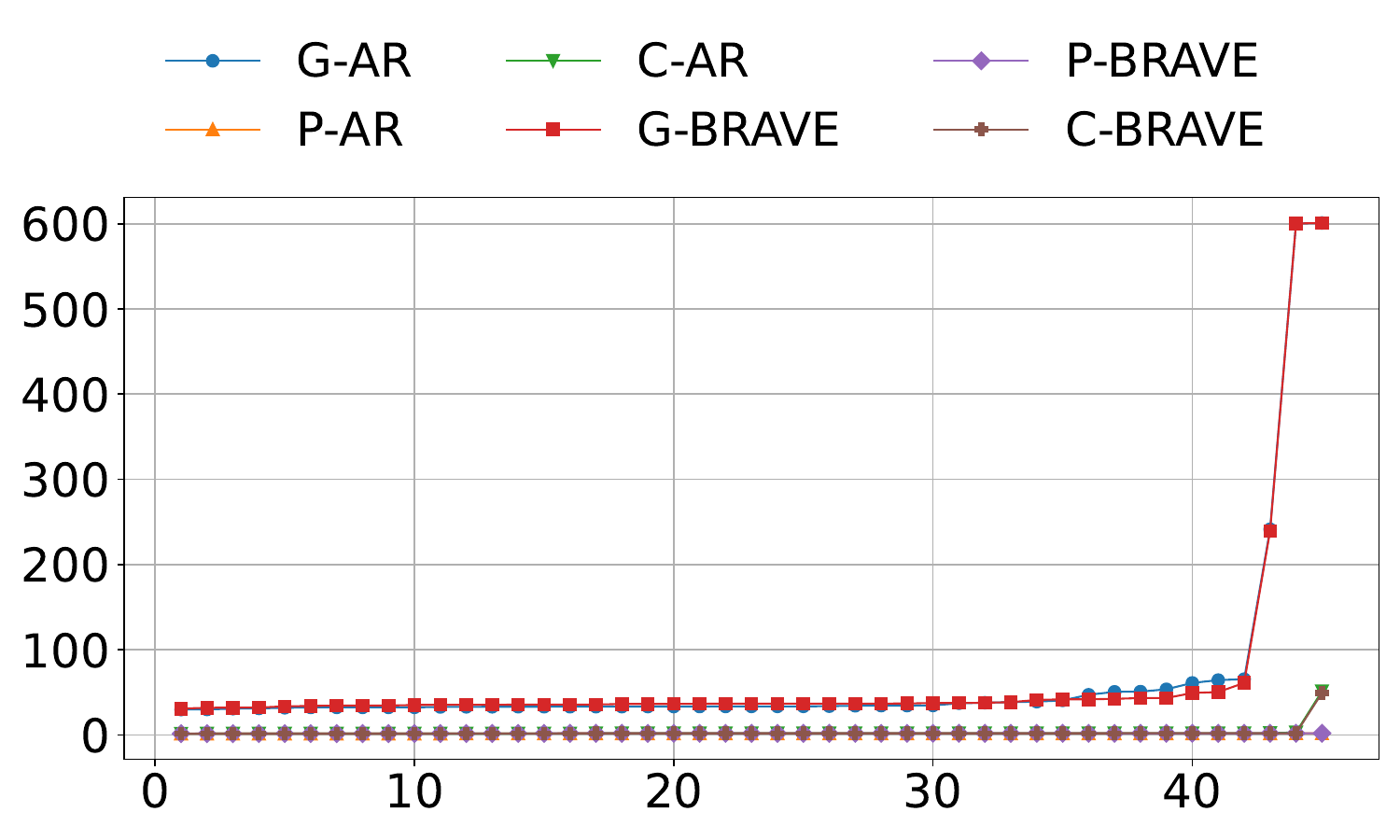}
        \caption{\mn{u1c20} $\succ^{nb}$}
        \label{fig:u1conf20}
    \end{subfigure}
    \\
    \begin{subfigure}{0.24\textwidth}
        \centering
        \includegraphics[width=\linewidth]{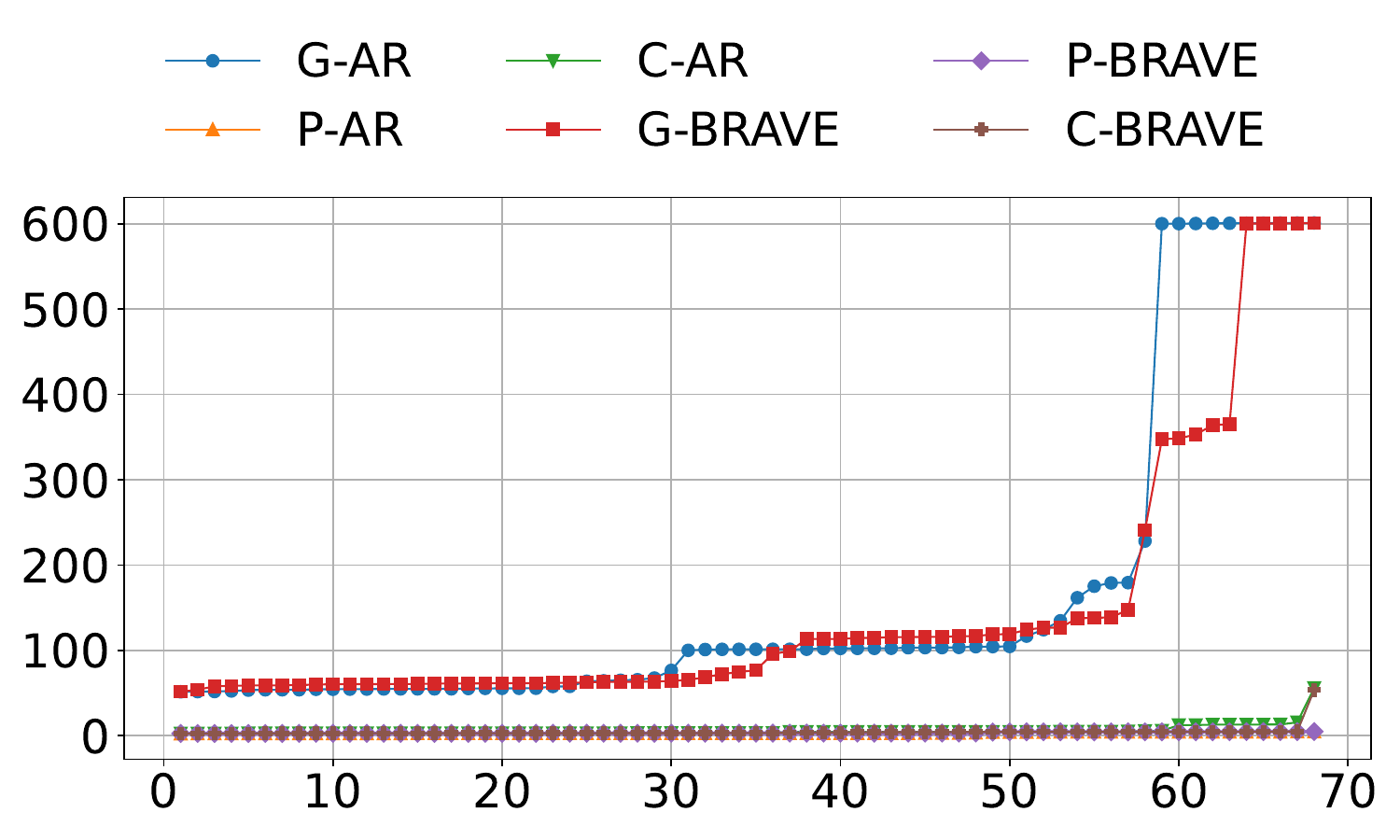}
        \caption{\mn{u1c30} $\succ^{nb}$}
        \label{fig:u1conf30}
    \end{subfigure}
    \\
    \begin{subfigure}{0.24\textwidth}
        \centering
        \includegraphics[width=\linewidth]{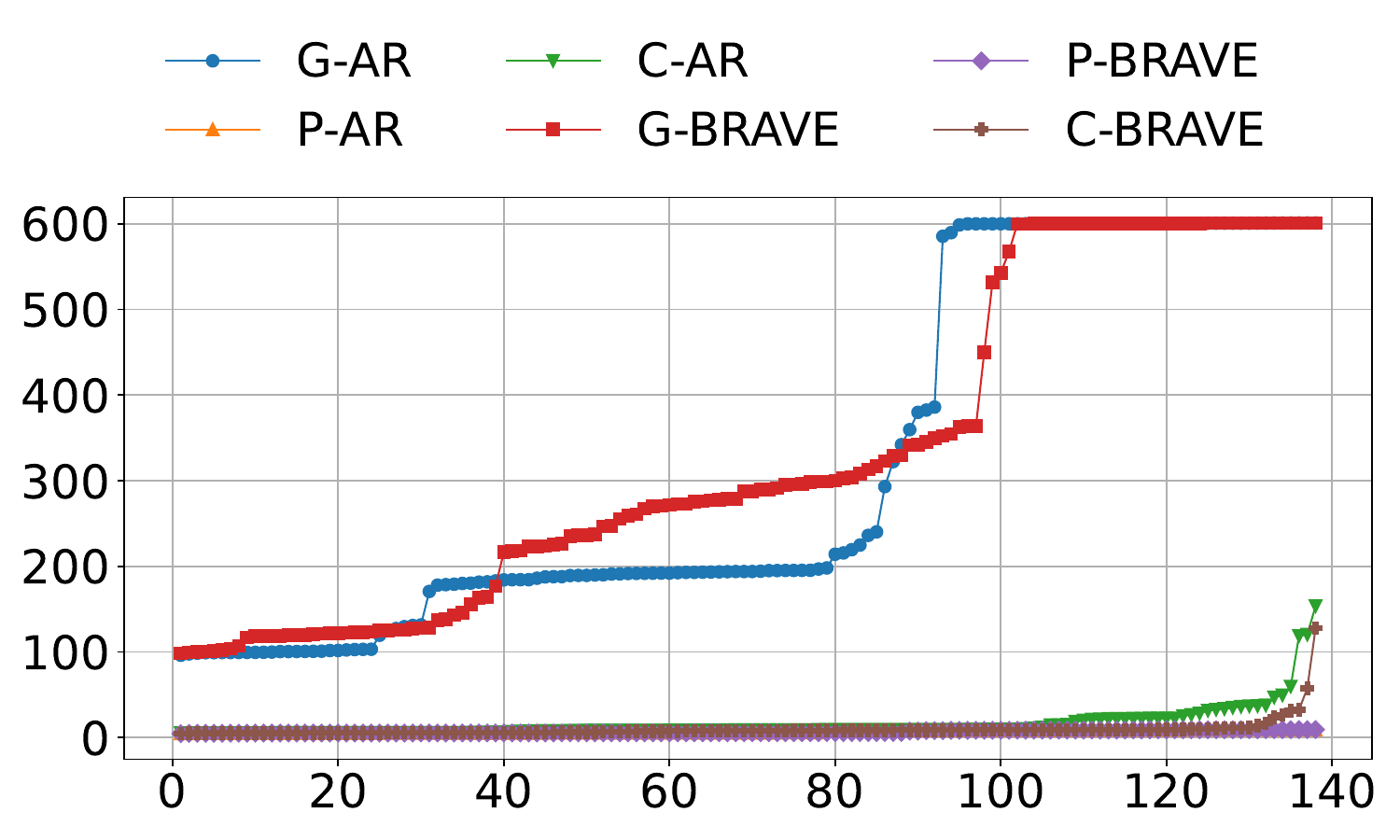}
        \caption{\mn{u1c50} $\succ^{nb}$}
        \label{fig:u1conf50}
    \end{subfigure}
   \caption{Time (in seconds) to decide for each $q(\ans)$ whether $q(\ans)$ holds under X-AR/X-brave, with instances sorted by increasing solving time 
along the x-axis (\ie a point $(i,j)$ indicates that $i$ instances could be solved within $j$ seconds).}\label{fig:overall-runtime-nonbinary-appendix} 
    
\end{figure*}

\begin{table*}
    \centering\footnotesize
    \resizebox{\textwidth}{!}{    
    \begin{tabular}{ll|rr|rrr|rrr}
        \hline
           &                      & Triv & GR$\setminus$Triv  & P-AR$\setminus$GR    & G-AR$\setminus$GR  & C-AR$\setminus$GR  & P-brave$\setminus$GR   & G-brave$\setminus$GR   & C-brave$\setminus$GR\\
        \hline
        \mn{u1c1} & $\succ^{ns}$                     & 1465	       & 59	            & 0 (0)	            & 0 (0)	            & 0 (0)	            & 0 (0)                 & 0 (0)                 & 0 (0)\\
        \mn{u1c5} & $\succ^{ns}$                     & 1366	        & 159	            & 0 (0)	            & 0 (0)	            & 0 (0)	            & 2 (0)                 & 2 (0)                 & 2 (0)\\
        \mn{u1c10} & $\succ^{ns}$                    & 1186	        & 341	            & 0 (0)	            & 0 (0)	            & 0 (0)	            & 6 (0)                 & 6 (0)                 & 6 (0)\\
        \mn{u1c20} & $\succ^{ns}$                    & 937	            & 584	            & 1 (0)	            & 0 (10)            & 1 (0)	            & 17 (0)                & 16 (6)                & 17 (0)\\
        \mn{u1c30} & $\succ^{ns}$                    & 754	            & 764	            & 6 (0)	            & 0 (30)            & 6 (0)	            & 38 (0)                & 19 (30)               & 38 (0)\\
        \mn{u1c50} & $\succ^{ns}$                    & 625	            & 875	            & 26 (0)	            & 9 (82)            & 20 (56)           & 79 (0)                & 28 (81)               & 38 (40)\\
        \mn{u5c1} & $\succ^{ns}$                     & 7637	        & 316	            & 0 (0)	            & 0 (0)	            & 0 (0)             & 2 (0)                 & 2 (0)                 & 2 (0)\\
        \mn{u5c5} & $\succ^{ns}$                     & 6958	        & 990	            & 22 (0)              & 21 (2)	        & 22 (0)            & 31 (0)                & 30 (2)                & 31 (0)\\
        \mn{u5c10} & $\succ^{ns}$                    & 6069	        & 1895	            & 22 (0)              & 8 (27)	        & 22 (0)            & 46 (0)                & 34 (22)               & 46 (0)\\
        \mn{u5c20} & $\succ^{ns}$                    & 4947	        & 3035	            & 37 (0)              & 6 (96)	        & 32 (32)           & 91 (0)                & 51 (151)              & 83 (18)\\
        \hline
        \mn{u1c1} & $\succ^{ss}$                           & 1477	        & 48	            & 0 (0)               & 0 (0)	            & 0 (0)             & 0 (0)                 & 0 (0)                 & 0 (0)\\
        \mn{u1c5} & $\succ^{ss}$                           & 1385	        & 65	            & 0 (0)               & 0 (0)	            & 0 (0)             & 81 (0)                & 81 (0)                & 81 (0)\\
        \mn{u1c10} & $\succ^{ss}$                          & 1215	        & 232	            & 0 (0)               & 0 (0)	            & 0 (0)             & 92 (0)                & 92 (0)                & 92 (0)\\
        \mn{u1c20} & $\succ^{ss}$                          & 931	            & 374	            & 0 (0)               & 0 (27)            & 0 (0)             & 238 (0)               & 228 (10)              & 238 (0)\\
        \mn{u1c30} & $\succ^{ss}$                          & 731	            & 299	            & 0 (0)               & 0 (73)            & 0 (0)             & 519 (0)               & 475 (46)              & 519 (0)\\
        \mn{u1c50} & $\succ^{ss}$                          & 618	            & 152	            & 0 (0)               & 0 (434)           & 0 (0)             & 807 (0)               & 512 (310)             & 807 (0)\\
        \mn{u5c1} & $\succ^{ss}$                           & 7706	        & 195	            & 0 (0)               & 0 (0)             & 0 (0)             & 55 (0)                & 55 (0)                & 55 (0)\\
        \mn{u5c5} & $\succ^{ss}$                           & 7044	        & 591	            & 1 (0)               & 1 (0)             & 1 (0)             & 348 (0)               & 348 (0)               & 348 (0)\\
        \mn{u5c10} & $\succ^{ss}$                          & 6213	        & 1087	            & 1 (0)               & 1 (73)            & 1 (0)             & 716 (0)               & 684 (38)              & 716 (0)\\
        \mn{u5c20} & $\succ^{ss}$                          & 5169	        & 1480	            & 2 (0)               & 2 (468)           & 2 (0)             & 1437 (0)              & 992 (533)             & 1437 (0)\\
        \hline
        \mn{u1c1} & $\succ^{nb}$             & 1447	        & 76	            & 0 (0)               & 0 (0)             & 0 (0)             & 0 (0)                 & 0 (0)                 & 0 (0)\\
        \mn{u1c5} & $\succ^{nb}$             & 1369	        & 155	            & 0 (0)               & 0 (1)             & 0 (0)             & 1 (0)                 & 0 (1)                 & 1 (0)\\
        \mn{u1c10} & $\succ^{nb}$            & 1167	        & 359	            & 0 (0)               & 0 (1)             & 0 (0)             & 1 (0)                 & 0 (1)                 & 1 (0)\\
        \mn{u1c20} & $\succ^{nb}$            & 940	            & 571	            & 0 (0)               & 0 (2)             & 0 (0)             & 1 (0)                 & 0 (2)                 & 1 (0)\\
        \mn{u1c30} & $\succ^{nb}$            & 789	            & 726	            & 0 (0)               & 0 (10)            & 0 (0)             & 6 (0)                 & 5 (5)                 & 6 (0)\\
        \mn{u1c50} & $\succ^{nb}$            & 680	            & 823	            & 0 (0)               & 0 (43)            & 0 (0)             & 20 (0)                & 13 (37)               & 20 (0)\\
        \mn{u5c1} & $\succ^{nb}$             & 7708	        & 238	            & 0 (0)               & 0 (0)             & 0 (0)             & 2 (0)                 & 2 (0)                 & 2 (0)\\
        \mn{u5c5} & $\succ^{nb}$             & 7209	        & 742	            & 0 (0)               & 0 (1)             & 0 (0)             & 3 (0)                 & 2 (1)                 & 3 (0)\\
        \hline
    \end{tabular}
    }
    \caption{Number of answers found trivially P-IAR (Triv), grounded (GR) but not trivially P-IAR, and X-AR or X-brave but not grounded. The number of potential answers for which we ran out of time (600s) is given in parenthesis. When the priority is score-structured ($\succ^{ss}$), note that all optimal repairs coincide so it is expected that the number of X-AR/X-brave answers is the same for every X.
    }\label{tab:number-answers-appendix}
\end{table*}

\end{document}